\documentclass[submission,copyright,creativecommons]{eptcs-neutered}
\providecommand{\event}{TERMGRAPH 2012} 

\usepackage{amsfonts,amsmath,amscd,amssymb,amsbsy,float}
\usepackage{latexsym,MnSymbol}
\usepackage[german,british,english]{babel}
\usepackage{comment}
\usepackage{xspace}
\usepackage{color,graphicx}
\usepackage[utf8]{inputenc}
\usepackage{url,breakurl,hyperref}
\usepackage{enumerate}
\usepackage{tikz}
\usepackage[normalem]{ulem}
\usetikzlibrary{calc,arrows,automata,backgrounds,shapes.callouts,calc,chains,fadings,folding,decorations.fractals,decorations.pathreplacing,fit,patterns,positioning,mindmap,shadows,shapes.geometric,shapes.symbols,through,trees,plotmarks}

\title{Term Graph Representations for Cyclic Lambda-Terms%
       \footnote{This work was started, and in part carried out, 
                 within the framework of the project NWO~project \emph{Realising Optimal Sharing (ROS)}, project number 612.000.935, 
                 under the direction of Vincent von Oostrom and Doaitse Swierstra.}}
\author{Clemens Grabmayer
\institute{Department of Philosophy\\
           Utrecht University\\
           The Netherlands}
\email{clemens@phil.uu.nl}
\and
Jan Rochel 
\institute{Department of Computing Sciences\\
           Utrecht University\\
           The Netherlands}
\email{jan@rochel.info}
}
\def\titlerunning{Term Graph Representations for Cyclic Lambda-Terms}
\def\authorrunning{C. Grabmayer \& J. Rochel}

\newcommand\itemizeprefs{
  \setlength\itemsep{-0.2ex}
  \vspace{-0.4ex}
}

\newcommand{\inMath}[1]{\ensuremath{#1}}

\newcommand\vcentered[1]{\raisebox{-0.5\height}{#1}}

\newcommand\graphs[1]{\par\vspace{2mm}\hfill{#1}\ \vspace{2mm}\par\noindent}
\newcommand\graph[2]{\vcentered{$#1$:~~} \vcentered{\fig{#2}}\hfill}

\newcommand\state{\inMath{s}}

\newcommand\partialTo\rightharpoonup

\newcommand{\tuple}[1]{\langle #1 \rangle}
\newcommand\tuplespace{\hspace*{0.5pt}}
\newcommand\pair[2]{\tuple{#1, \tuplespace #2}}

\newcommand\fig[1]{\includegraphics[scale=0.84]{figs/{{#1}}}}
\newcommand\figsmall[1]{\includegraphics[scale=0.8]{figs/{{#1}}}}
\newcommand{\lambdacalculus}{$\lambda$\nb-calculus}

\newcommand{\lambdaterm}{$\lambda$\nb-term}
\newcommand{\lambdaterms}{\lambdaterm{s}}
\newcommand{\lambdaabstraction}{$\lambda$\nb-ab\-strac\-tion}
\newcommand{\lambdaabstractions}{\lambdaabstraction{s}}

\newcommand{\Sbisimilarity}{$\snlvarsucc$\nb-bi\-si\-mi\-la\-rity}
\newcommand{\Sbisimilar}{$\snlvarsucc$\nb-bi\-si\-mi\-lar}

\let\oldlambda\lambda
\renewcommand\lambda{\inMath\oldlambda}
\let\oldalpha\alpha
\renewcommand\alpha{\inMath\oldalpha}
\let\oldmu\mu
\renewcommand\mu{\inMath\oldmu}

\newcommand{\stxtletrec}{\ensuremath{\text{\normalfont\sf letrec}}}

\newcommand{\sslabs}{\lambda}
\newcommand{\slabs}[1]{\sslabs{#1}}
\newcommand{\labs}[2]{\slabs{#1}.\,{#2}}
\newcommand{\sslapp}{\hspace*{-1pt}@\hspace*{-1pt}}
\newcommand{\slapp}{\hspace*{1.5pt}}
\newcommand{\lapp}[2]{{#1}\slapp{#2}}

\newcommand{\snlvar}{\mathsf{0}}

\newcommand{\snlvarsucc}{\mathsf{S}}

\newcommand{\strailer}{\mathsf{T}}

\newcommand{\ssin}{{\textbf{in}}}

\newcommand{\sletrec}{\textbf{letrec}}

\newcommand{\letrecin}[2]{\sletrec\;{#1}\;\ssin\;{#2}}

\newcommand{\arecvar}{f}
\newcommand{\brecvar}{g}

\makeatletter
\def\a#1{\reflectbox{$\m@th#1{\lambda}$}}
\def\adbmal{\inMath{\mathpalette{\a}{}}}
\makeatother
\newcommand{\avar}{x}
\newcommand{\bvar}{y}
\newcommand{\cvar}{z}
\newcommand{\dvar}{u}

\newcommand{\asig}{\Sigma}
\newcommand{\bsig}{\Delta}
\newcommand{\asigacc}{\Sigma'}

\newcommand{\sarity}{{\mit ar}}
\newcommand{\arity}{\funap{\sarity}}

\newcommand{\afunsym}{\mathsf{f}}

\newcommand{\CRS}{${\text{CRS}}$}

\newcommand{\apath}{\pi}
\newcommand{\apathi}{\indap{\pi}}
\newcommand{\apathacc}{\pi'}

\newcommand\unfold\bigtriangledown
\newcommand{\funin}{\mathrel{:}}
\newcommand{\funap}[2]{{#1}({#2})}
\newcommand{\bfunap}[3]{{#1}({#2},\hspace*{0.5pt}{#3})}

\newcommand{\indap}[2]{#1_{#2}}
\newcommand{\sdefdby}{{:=}}
\newcommand{\defdby}{\mathrel{\sdefdby}}
\newcommand{\length}[1]{\left|{#1}\right|}
\newcommand{\nb}{\nobreakdash}

\newcommand{\bs}{\boldsymbol}
\newcommand{\sdom}{\textrm{dom}}
\newcommand{\dom}{\funap{\sdom}}

\newcommand{\simage}{\textrm{im}}
\newcommand{\image}{\funap{\simage}}
\newcommand{\mcdots}{\hspace*{1pt}\cdots\hspace*{2pt}}

\newcommand{\sfuncomp}{\circ}
\newcommand{\scompfuns}[2]{{#1}\mathrel{\sfuncomp}{#2}}
\newcommand{\compfuns}[2]{\funap{\scompfuns{#1}{#2}}}
\newcommand{\srewrelcomp}{{\cdot}}
\newcommand{\scomprewrels}[2]{{#1}\mathrel{\srewrelcomp}{#2}}
\newcommand{\comprewrels}[2]{\mathrel{\scomprewrels{#1}{#2}}}

\newcommand{\sproji}{\indap{\pi}}

\newcommand{\eqcl}[2]{\indap{[{#1}]}{#2}}
\newcommand{\eqclin}[3]{\bpap{[{#1}]}{#2}{#3}}

\newcommand{\existsst}[2]{\exists{#1}.\;{#2}}
\newcommand{\forallst}[2]{\forall{#1}.\;{#2}}
\newcommand{\existsstzero}[1]{\exists{#1}.\;}
\newcommand{\forallstzero}[1]{\forall{#1}.\;}

\newcommand{\emptyword}{\epsilon}
\newcommand{\sstringcon}{\hspace*{1pt}}
\newcommand{\stringcon}[2]{{#1}{\sstringcon}{#2}}

\newcommand{\srestrictfunto}[2]{{#1}{\mid}_{#2}}
\newcommand{\restrictfunto}[2]{\funap{\srestrictfunto}}

\newcommand{\sidfun}{\text{\normalfont id}}

\newcommand{\sidfunon}{\indap{\sidfun}}
\newcommand{\idfunon}[1]{\funap{\sidfunon{#1}}}

\newcommand{\srestrictto}[2]{{#1}\!\mid_{#2}}
\newcommand{\restrictto}[2]{\funap{\srestrictto}}
\newcommand{\slogand}{{\wedge}}
\newcommand{\logand}{\mathrel{\slogand}}
\newcommand{\slogor}{{\vee}}
\newcommand{\logor}{\mathrel{\slogor}}
\newcommand{\ssbinrelcomp}{\cdot}
\newcommand{\sbinrelcomp}[2]{{#1}\mathrel{\ssbinrelcomp}{#2}}
\newcommand{\binrelcomp}[2]{\mathrel{\sbinrelcomp{#1}{#2}}}
\newcommand{\subap}[2]{{#1}_{#2}}
\newcommand{\supap}[2]{{#1}^{#2}}
\newcommand{\bpap}[3]{{#1}_{#2}^{#3}}
\newcommand{\pbap}[3]{{#1}^{#2}_{#3}}
\newcommand{\descsetexpmid}{\mathrel{\vert}}

\newcommand{\descsetexpBigmid}{\mathrel{\Big\vert}}

\newcommand{\descsetexp}[2]{\left\{{#1}\descsetexpmid{#2}\right\}}

\newcommand{\descsetexpBig}[2]{\Bigl\{{#1}\descsetexpBigmid{#2}\Bigr\}}

\newcommand{\setexp}[1]{\left\{{#1}\right\}}

\newcommand{\spowersetof}{\powerset}
\newcommand{\powersetof}{\funap{\spowersetof}}
\newcommand{\factorset}[2]{{#1}/_{#2}}

\newcommand{\aset}{A}
\newcommand{\bset}{B}

\newcommand{\spo}{\le}
\newcommand{\spoi}{\indap{\spo}}

\newcommand{\poi}[1]{\mathrel{\spoi{#1}}}

\newcommand{\txtlub}{l.u.b.}
\newcommand{\txtglb}{g.l.b.}

\newcommand{\slub}{\bigsqcup}
\newcommand{\sglb}{\bigsqcap}
\newcommand{\lub}[1]{{\slub {#1}}}
\newcommand{\glb}[1]{{\sglb {#1}}}

\newcommand{\punc}[1]{\inMath{\hspace*{3pt}{#1}}}
\newcommand{\nats}{\mathbb{N}}

\definecolor{azure}{rgb}{0.94,1.00,1.00}
\definecolor{blue}{rgb}{0,0,0.5}
\definecolor{brown}{rgb}{.75,.25,.25}
\definecolor{cyan}{rgb}{0.25,0.88,0.82}
\definecolor{chocolate}{rgb}{0.82,0.41,0.12}
\definecolor{darkcyan}{rgb}{0.5,0,1}
\definecolor{darkgreen}{rgb}{0,0.39,0}
\definecolor{darkmagenta}{rgb}{0.5,0,0.5}
\definecolor{firebrick}{RGB}{175,25,25}
\definecolor{forestgreen}{rgb}{0.13,0.55,0.13}
\definecolor{lightcyan}{rgb}{0.88,1.00,1.00}
\definecolor{lightpink}{rgb}{1.00,0.71,0.76}
\definecolor{lightyellow}{rgb}{1.00,1.00,0.88}
\definecolor{lightgoldenrod}{rgb}{0.83,0.97,0.51}
\definecolor{lightgoldenrodyellow}{rgb}{0.98,0.98,0.82}
\definecolor{lightskyblue}{rgb}{0.53,0.81,0.98}
\definecolor{moccasin}{rgb}{1.00,0.89,0.71}
\definecolor{magenta}{rgb}{1,0,1}
\definecolor{navyblue}{rgb}{0,0,0.5}
\definecolor{orange}{rgb}{1.0,0.65,0.0}
\definecolor{orangered}{rgb}{1.0,0.27,0.0}
\definecolor{palegreen}{rgb}{0.60,0.98,0.60}
\definecolor{powderblue}{rgb}{0.69,0.88,0.90}
\definecolor{purple}{rgb}{1,0.5,1}
\definecolor{royalblue}{RGB}{65,105,225}
\definecolor{mediumblue}{RGB}{0,0,205}
\definecolor{cornflowerblue}{RGB}{100,149,237}
\definecolor{springgreen}{rgb}{0.0,1.0,0.5}
\definecolor{turquoise}{rgb}{0.25,0.88,0.82}
\definecolor{snow}{rgb}{1.00,0.98,0.98}
\definecolor{tan}{rgb}{0.82,0.71,0.55}
\definecolor{red}{rgb}{1,0,0}

\newcommand{\safun}{f}

\newcommand{\atg}{G}
\newcommand{\atgi}{\indap{\atg}}
\newcommand{\atgacc}{G'}

\newcommand{\atgiso}{\textbf{\textit{G}}}

\newcommand{\succsoford}[2]{({#1}{#2})}

\newcommand{\succsofordin}[3]{({#1}{#2})^{#3}}

\newcommand{\verts}{V}
\newcommand{\vertsof}{\funap{\verts\!}}
\newcommand{\svlab}{\mathit{lab}}
\newcommand{\vlab}{\funap{\svlab}}
\newcommand{\svargs}{\mathit{args}}
\newcommand{\vargs}{\funap{\svargs}}
\newcommand{\sroot}{\mathit{r}}
\newcommand{\vertsacc}{V'}

\newcommand{\svlabacc}{\mathit{lab}'}
\newcommand{\vlabacc}{\funap{\svlabacc}}
\newcommand{\svargsacc}{\mathit{args}'}

\newcommand{\srootacc}{\mathit{r}'}
\newcommand{\vertsi}{\indap{V}}
\newcommand{\vertsiof}[1]{\funap{\vertsi{#1}\!}}
\newcommand{\svlabi}{\indap{\mathit{lab}}}
\newcommand{\vlabi}[1]{\funap{\svlabi{#1}}}
\newcommand{\svargsi}{\indap{\mathit{args}}}
\newcommand{\vargsi}[1]{\funap{\svargsi{#1}}}
\newcommand{\srooti}{\indap{\mathit{r}}}

\newcommand{\avert}{v}
\newcommand{\bvert}{w}

\newcommand{\averti}{\indap{\avert}}
\newcommand{\bverti}{\indap{\bvert}}

\newcommand{\bvertbp}{\bpap{\bvert}}

\newcommand{\avertacc}{v'}
\newcommand{\bvertacc}{w'}

\newcommand{\avertacci}{\indap{\avertacc}}
\newcommand{\bvertacci}{\indap{\bvertacc}}

\newcommand{\siglambda}{\supap{\asig}{\lambda}}

\newcommand{\siglambdai}[1]{\pbap{\asig}{\sslabs}{{#1}}}
\newcommand{\siglambdaij}[2]{\pbap{\asig}{\sslabs}{{#1},{#2}}}

\newcommand{\stgsucc}{{\rightarrowtail}}
\newcommand{\tgsucc}{\mathrel{\stgsucc}}
\newcommand{\stgsucci}{\indap{\rightarrowtail}}
\newcommand{\tgsucci}[1]{\mathrel{\stgsucci{#1}}}
\newcommand{\stgsuccstar}{{\rightarrowtail^*}}
\newcommand{\tgsuccstar}{\mathrel{\stgsuccstar}}

\newcommand{\stgsuccacc}{\stgsucc'}
\newcommand{\tgsuccacc}{\mathrel{\stgsuccacc}}
\newcommand{\stgsuccacci}{\indap{\rightarrowtail'}}
\newcommand{\tgsuccacci}[1]{\mathrel{\stgsuccacci{#1}}}

\newcommand{\stgsuccis}[2]{{{}^{#2}\hspace*{-1.5pt}\stgsucc_{#1}}}
\newcommand{\tgsuccis}[2]{\mathrel{\stgsuccis{#1}{#2}}}
\newcommand{\stgsuccisstar}[2]{(\stgsuccis{#1}{#2})^*}

\newcommand{\stgsuccins}[3]{{{}^{#3}\hspace*{-1.5pt}\stgsucc_{#1}^{#2}}}

\newcommand{\snodels}{{\#\text{\normalfont del}}}
\newcommand{\nodels}{\bfunap{\snodels}}

\newcommand{\spathto}{\rightarrowtail\hspace*{-5pt}^*}
\newcommand{\pathto}{\mathrel{\spathto}}

\newcommand{\sahom}{h}
\newcommand{\ahom}{\funap{\sahom}}
\newcommand{\sahomext}{\bar{\sahom}}
\newcommand{\ahomext}{\funap{\sahomext}}
\newcommand{\sahomextext}{\bar{\sahomext}}

\newcommand{\sbhom}{g}

\newcommand{\saiso}{i}

\newcommand{\saohom}{H}
\newcommand{\sbohom}{I}
\newcommand{\aohom}{\funap{\saohom}}
\newcommand{\bohom}{\funap{\sbohom}}

\newcommand{\siso}{{\sim}}
\newcommand{\iso}{\mathrel{\sim}}
\newcommand{\sisoi}{\indap{\siso}}
\newcommand{\isoi}[1]{\mathrel{\sisoi{#1}}}

\newcommand{\sbisim}{%
    \setbox0=\hbox{\kern-.1ex{$\leftrightarrow$}\kern-.1ex}
    \setbox1=\vbox{\hbox{\raise .1ex \box0}\hrule}%
    \inMath{\mathrel{\hbox{\kern.1ex\box1\kern.1ex}}}
  }
\newcommand{\bisim}{\mathrel{\sbisim}}
\newcommand{\bisimi}[1]{\mathrel{\sbisim_{#1}}}

\newcommand{\sbisims}{\supap{\sbisim}}
\newcommand{\bisims}[1]{\mathrel{\sbisims{#1}}}
\newcommand{\sbisimS}{\sbisims{\snlvarsucc}}
\newcommand{\bisimS}{\mathrel{\sbisimS}}

\newcommand{\sbisimis}{\bpap{\sbisim}}
\newcommand{\bisimis}[2]{\mathrel{\sbisimis{#1}{#2}}}

\newcommand{\sbisimssubscript}{\supap{\sbisimsubscript}}

\newcommand{\sbisimSsubscript}{\sbisimssubscript{\snlvarsucc}}

\newcommand{\sbisimsubscript}{
    \setbox0=\hbox{\kern-.1ex{$\leftrightarrow$}\kern-.1ex}
    \setbox1=\vbox{\hbox{\raise .1ex \box0}\hrule}%
    \inMath{\mathrel{\hbox{\scalebox{0.75}{\box1}}}}
  }

\newcommand{\sfunbisim}{%
    \setbox0=\hbox{\kern-.1ex{$\rightarrow$}\kern-.1ex}
    \setbox1=\vbox{\hbox{\raise .1ex \box0}\hrule}%
    {\hbox{\kern.1ex\box1\kern.1ex}}
  }
\newcommand{\funbisim}{\mathrel{{\sfunbisim}}}

\newcommand{\sfunbisimi}{\indap{\sfunbisim}}
\newcommand{\funbisimi}[1]{\mathrel{\sfunbisimi{#1}}}

\newcommand{\sfunbisims}{\supap{\sfunbisim}}
\newcommand{\funbisims}[1]{\mathrel{\sfunbisims{#1}}}

\newcommand{\sfunbisimis}{\bpap{\sfunbisim}}
\newcommand{\funbisimis}[2]{\mathrel{\sfunbisimis{#1}{#2}}}

\newcommand{\sfunbisimS}{\sfunbisims{\snlvarsucc}}

\newcommand{\sconvfunbisim}[1][]{%
    \setbox0=\hbox{\kern-.1ex{$\leftarrow$}\kern-.1ex}
    \setbox1=\vbox{\hbox{\raise .1ex \box0}\hrule}%
    \mathrel{\hbox{\kern.1ex\box1\kern.1ex}}
  }
\newcommand{\convfunbisim}{\mathrel{\sconvfunbisim}}

\newcommand{\sconvfunbisimi}{\indap{\sconvfunbisim}}
\newcommand{\convfunbisimi}[1]{\mathrel{\sconvfunbisimi{#1}}}

\newcommand{\sconvfunbisims}{\supap{\sconvfunbisim}}
\newcommand{\convfunbisims}[1]{\mathrel{\sconvfunbisims{#1}}}

\newcommand{\sconvfunbisimis}{\bpap{\sconvfunbisim}}
\newcommand{\convfunbisimis}[2]{\mathrel{\sconvfunbisimis{#1}{#2}}}

\newcommand{\sSfunbisim}{\sfunbisim^{\snlvarsucc}}
\newcommand{\Sfunbisim}{\mathrel{\sSfunbisim}}

\newcommand{\sfunbisimsubscript}{
    \setbox0=\hbox{\kern-.1ex{$\rightarrow$}\kern-.1ex}
    \setbox1=\vbox{\hbox{\raise .1ex \box0}\hrule}%
    \inMath{\mathrel{\hbox{\scalebox{0.75}{\box1}}}}
  }

\newcommand{\abisim}{R}

\newcommand{\scollC}[1]{{{\text{\small\textbar}}\hspace{-0.73ex}\downarrow}_{#1}}
\newcommand{\collC}[1]{\funap{\scollC{#1}}}

\newcommand{\sScope}{\mathit{Sc}}
\newcommand{\Scope}{\funap{\sScope}}
\newcommand{\sScopemin}{\mathit{Sc}^{-}}
\newcommand{\Scopemin}{\funap{\sScopemin}}
\newcommand{\sScopei}{\indap{\mathit{Sc}}}
\newcommand{\Scopei}[1]{\funap{\sScopei{#1}}}

\newcommand{\sbinders}{\mathrm{bds}}
\newcommand{\binders}{\funap{\sbinders}}

\newcommand{\sscopeforgetfullhotgsi}[1]{\sScope\mathit{F}_{#1}^{\sslabs}}
\newcommand{\scopeforgetfullhotgsi}[1]{\funap{\sscopeforgetfullhotgsi{#1}}}
\newcommand{\sabspreforgetfullaphotgsi}[1]{{\sabspre\mathit{F}_{#1}}\hspace*{-0.1em}^{(\sslabs)}}
\newcommand{\abspreforgetfullaphotgsi}[1]{\funap{\sabspreforgetfullaphotgsi{#1}}}

\newcommand{\lambdahotg}{$\lambda$\nb-ho-term-graph}
\newcommand{\lambdahotgs}{\lambdahotg{s}}
\newcommand{\alhotg}{{\cal G}}
\newcommand{\alhotgi}[1]{{\cal G}_{#1}}
\newcommand{\alhotgacc}{{\cal G}'}

\newcommand{\lambdaaphotg}{$\lambda$\nb-ap-ho-term-graph}
\newcommand{\lambdaaphotgs}{\lambdaaphotg{s}}
\newcommand{\alaphotg}{{\cal G}}
\newcommand{\alaphotgi}[1]{{\cal G}_{#1}}

\newcommand{\alaphotgiso}{\bs{\cal G}}
\newcommand{\alaphotgacciso}{\bs{\cal G}^{\bs{\prime}}}

\newcommand{\lambdatg}{$\lambda$\nb-term-graph}
\newcommand{\lambdatgs}{\lambdatg{s}}
\newcommand{\altg}{G}
\newcommand{\altgi}{\indap{\altg}}

\newcommand{\altgacc}{G'}

\newcommand{\altgiso}{\textbf{\textit{G}}}

\newcommand{\altgacciso}{\textbf{\textit{G}}^{\bs{\prime}}}

\newcommand{\seag}{\text{\normalfont eag}}
\newcommand{\sfbl}{\text{\normalfont fbl}}

\newcommand{\eagscope}{eager scope}
\newcommand{\fb}{fully back-linked}

\newcommand{\classlhotgsi}{\pbap{\cal H}{\sslabs}}
\newcommand{\classlaphotgsi}[1]{\subap{\cal H}{#1}{\hspace*{-0.1em}}^{(\sslabs)}}
\newcommand{\classeaglhotgsi}{{}^{\seag}\pbap{\cal H}{\sslabs}}
\newcommand{\classeaglaphotgsi}[1]{{}^{\seag}\subap{\cal H}{#1}{\hspace*{-0.1em}}^{(\sslabs)}}
\newcommand{\classlhotgsisoi}{\pbap{\bs{\cal H}}{\bs{\sslabs}}}
\newcommand{\classlaphotgsisoi}[1]{\subap{\bs{\cal H}}{{#1}}{\hspace*{-0.1em}}^{(\bs{\sslabs})}}

\newcommand{\classeaglaphotgsisoi}[1]{{}^{\seag}\subap{\bs{\cal H}}{{#1}}{\hspace*{-0.1em}}^{(\bs{\sslabs})}}

\newcommand{\classltgsi}[1]{\subap{\cal T}{#1}{\hspace*{-0.2em}}^{(\sslabs)}}
\newcommand{\classltgsij}[2]{\subap{\cal T}{#1,#2}{\hspace*{-0.55em}}^{(\sslabs)}\!}

\newcommand{\classeagltgsij}[2]{{}^{\seag}\subap{\cal T}{#1,#2}{\hspace*{-0.55em}}^{(\sslabs)}\!}

\newcommand{\classfblltgsij}[2]{{}^{\sfbl}\subap{\cal T}{#1,#2}{\hspace*{-0.55em}}^{(\sslabs)}\!}

\newcommand{\classltgsisoij}[2]{\subap{\bs{\cal T}}{\hspace*{-0.2em}{#1},{#2}}{\hspace*{-0.55em}}^{(\bs{\sslabs})}\!}

\newcommand{\classeagltgsisoij}[2]{{}^{\seag}\subap{\bs{\cal T}}{\hspace*{-0.2em}{#1},{#2}}{\hspace*{-0.55em}}^{(\bs{\sslabs})}\!}

\newcommand{\classlltgsij}[2]{\subap{\cal T}{#1,#2}{\hspace*{-0.45em}}^{[\hspace*{-1.25pt}(\sslabs)\hspace*{-1.25pt}]}\!}

\newcommand{\classlltgsisoij}[2]{\subap{\bs{\cal T}}{\hspace*{-0.2em}{#1},{#2}}{\hspace*{-0.55em}}^{[\hspace*{-1.25pt}(\bs{\sslabs})\hspace*{-1.25pt}]}\!}

\newcommand{\classtgssiglambdai}{\subap{\cal T}}
\newcommand{\classtgssiglambdaij}[2]{\subap{\cal T}{{#1},{#2}}}

\newcommand{\slhotgstolaphotgsi}{\indap{A}}
\newcommand{\lhotgstolaphotgsi}[1]{\funap{\slhotgstolaphotgsi{#1}}}
\newcommand{\slaphotgstolhotgsi}{\indap{B}}
\newcommand{\laphotgstolhotgsi}[1]{\funap{\slaphotgstolhotgsi{#1}}}
\newcommand{\slhotgsisotolaphotgsisoi}{\indap{\textit{\textbf{A}}}}

\newcommand{\slaphotgsisotolhotgsisoi}{\indap{\textit{\textbf{B}}}}

\newcommand{\slaphotgstoltgsij}[2]{G_{#1,#2}}
\newcommand{\laphotgstoltgsij}[2]{\funap{\slaphotgstoltgsij{#1}{#2}}}
\newcommand{\sltgstolaphotgsij}[2]{{\cal G}_{#1,#2}}
\newcommand{\ltgstolaphotgsij}[2]{\funap{\sltgstolaphotgsij{#1}{#2}}}
\newcommand{\slaphotgsisotoltgsisoij}[2]{\textit{\textbf{G}}_{#1,#2}}
\newcommand{\laphotgsisotoltgsisoij}[2]{\funap{\slaphotgsisotoltgsisoij{#1}{#2}}}
\newcommand{\sltgsisotolaphotgsisoij}[2]{\bs{\cal G}_{#1,#2}}

\newcommand{\sabspre}{P}
\newcommand{\abspre}{\funap{\sabspre}}
\newcommand{\sabsprei}{\indap{P}}
\newcommand{\absprei}[1]{\funap{\sabsprei{#1}}}
\newcommand{\sabspreacc}{P'}
\newcommand{\abspreacc}{\funap{\sabspreacc}}
\newcommand{\absprefix}{ab\-strac\-tion-pre\-fix}

\newcommand{\apre}{p}
\newcommand{\bpre}{q}
\newcommand{\cpre}{r}
\newcommand{\dpre}{s}
\newcommand{\aprei}{\indap{\apre}}
\newcommand{\bprei}{\indap{\bpre}}
\newcommand{\cprei}{\indap{\cpre}}
\newcommand{\dprei}{\indap{\dpre}}
\newcommand{\apreacc}{p'}

\newcommand{\prele}{\le}
\newcommand{\prelt}{<}
\newcommand{\prege}{\ge}

\newcommand{\aclass}{{\cal K}}

\newcommand{\classtgsminover}{\funap{{\text{\normalfont TG}^{-}}}}
\newcommand{\classtgsover}{\funap{{\text{\normalfont TG}}}}

\usetikzlibrary{shapes,positioning,decorations.markings,arrows}

\tikzstyle{->>>} =
  [shorten >= 0.3mm,
   decoration =
     {markings,
      mark=at position -2.4mm with {\arrow{>}},
      mark=at position -1.2mm with {\arrow{>}},
      mark=at position 1 with {\arrow{>}}},
   postaction={decorate}]

\pgfdeclaredecoration{funbisim}{final}{
  \state{final}[width=\pgfdecoratedpathlength]{
    \draw[->]
      (0,\pgfdecorationsegmentamplitude+0.1mm) -- +(\pgfdecoratedpathlength,0);
    \draw[-]
      (0,-\pgfdecorationsegmentamplitude-0.1mm) -- +(\pgfdecoratedpathlength,0);}}
\tikzset{
  funbisim/.style={
    decoration={funbisim, amplitude=0.25ex},
    decorate,
    funbisim options/.style={#1}    
  }}

\pgfdeclaredecoration{bisim}{final}{
  \state{final}[width=\pgfdecoratedpathlength]{
    \draw[<->]
      (0,\pgfdecorationsegmentamplitude+0.1mm) -- +(\pgfdecoratedpathlength,0);
    \draw[-]
      (0,-\pgfdecorationsegmentamplitude-0.1mm) -- +(\pgfdecoratedpathlength,0);}}
\tikzset{
  bisim/.style={
    decoration={bisim, amplitude=0.25ex},
    decorate,
    bisim options/.style={#1}    
  }}

\makeatletter
\pgfdeclareshape{transbox}{
  \inheritsavedanchors[from=rectangle]
  \inheritanchor[from=rectangle]{center}
  \inheritanchor[from=rectangle]{north}
  \inheritanchor[from=rectangle]{south}
  \inheritanchor[from=rectangle]{west}
  \inheritanchor[from=rectangle]{east}
  \inheritbehindbackgroundpath[from=rectangle]
  \inheritbackgroundpath[from=rectangle]
  \inheritbeforebackgroundpath[from=rectangle]
  \inheritbehindforegroundpath[from=rectangle]
  \inheritforegroundpath[from=rectangle]
  \inheritbeforeforegroundpath[from=rectangle]
}
\makeatother

\newtheorem{theorem}{Theorem}[section]
\newtheorem{corollary}[theorem]{Corollary}
\newtheorem{lemma}[theorem]{Lemma}
\newtheorem{proposition}[theorem]{Proposition}

\newtheorem{definition}[theorem]{Definition}
\newtheorem{remark}[theorem]{Remark}
\newtheorem{example}[theorem]{Example}

\newcommand{\qedsym}{\inMath{\Box}}
\newcommand{\qed}{\hspace*{\fill}\qedsym}
\newenvironment{proof}{\noindent{\normalfont{\emph{Proof}.}}}{\qed\medskip}

\newenvironment{proofof}[1]{\noindent{\normalfont{\emph{#1}}.}}{\qed\medskip}

\renewcommand{\emph}[1]{{\em #1}}

\begin{document}
\maketitle

\begin{abstract}
We study various representations for cyclic \lambda-terms as higher-order or as first-order term graphs.
We focus on the relation between `$\lambda$-higher-order term graphs'
(\lambdahotg{s}), which are first-order term graphs endowed with a well-behaved scope function, 
and their representations 
as `\lambdatg{s}', which are plain first-order term graphs with scope-delimiter vertices 
that meet certain scoping requirements. 
Specifically we tackle the question: 
Which class of first-order term graphs
admits a faithful embedding of \lambdahotg{s} in the sense that
(i)~the homomorphism-based sharing-order on \lambdahotg{s} is preserved and reflected,
and 
(ii)~the image of the embedding corresponds closely to a natural class (of \lambdatg{s})
     that is closed under homomorphism?

We systematically examine whether a number of classes of \lambdatg{s} have this property,  
and we find a particular class of \lambdatg{s} that satisfies this criterion.
Term graphs of this class are built from application, abstraction, variable, and scope-delimiter vertices,
and have the characteristic feature that the latter two kinds of vertices have back-links to the corresponding abstraction.

This result puts a handle on the concept of subterm sharing for
higher-order term graphs, both theoretically and algorithmically:
We obtain an easily implementable method for obtaining 
the maximally shared form of \lambdahotg{s}.
Also, we open up the possibility to pull back properties 
from first-order term graphs to \lambdahotg{s}.
In fact we prove this for the property 
of the sharing-order successors of a given term graph
to be a complete lattice with respect to the sharing order.

This report extends the paper \cite{grab:roch:2013:TERMGRAPH}
for the workshop TERMGRAPH~2013.
\end{abstract}

\section{Introduction}\label{sec:intro}
%
Cyclic lambda-terms typically represent infinite \lambda-terms. In this report we study term graph
representations of cyclic \lambda-terms and their respective notions of
homomorphism, or functional bisimulation.

The context in which the results presented in this paper play a central role is
our research on subterm sharing as present in terms of languages such as the
\lambdacalculus\ with \stxtletrec\ \cite{peyt:jone:1987,ario:blom:1997}, 
with recursive definitions \cite{ario:klop:1994}, or languages with $\mu$\nb-recursion~\cite{ario:klop:1996},
and our interest in describing maximal sharing in such settings. 
Specifically we want to obtain concepts and methods as follows:
\begin{itemize}\itemizeprefs
\item an efficient test for term equivalence with respect to \alpha-renaming and unfolding;
\item a notion of `maximal subterm sharing' for terms in the respective language;
\item the efficient computation of the maximally shared form of a term;
\item a sharing (pre-)order on unfolding-equivalent terms.
\end{itemize}\vspace*{-1ex}
Now our approach is to split the work into a part that concerns properties
specific to concrete languages, and into a part that deals with aspects that
are common to most of the languages with constructs for expressing subterm sharing. 
To this end we set out to find classes of term graphs 
that facilitate faithful interpretations of terms in such languages as 
(higher-order, and eventually first-order) term graphs,
and that are `well-behaved' in the sense that maximally shared term graphs do always exist.
In this way the task can be divided into two parts:
an investigation of sharing for term graphs with higher-order features (the aim of this paper),
and a study of language-specific aspects of sharing (the aim of a further paper).

Here we study a variety of classes of term graphs for denoting cyclic \lambdaterms,
term graphs with higher-order features and their first-order `implementations'. 
All higher-order term graphs we consider are built from three kinds of vertices,
which symbolize applications, abstractions, and variable occurrences, respectively. 
They also carry features that describe notions of \emph{scope}, 
which are subject to certain conditions that guarantee the meaningfulness of the term graph
(that a \lambdaterm\ is denoted), and in some cases are crucial to define \emph{binding}. 
The first-order implementations do not have these additional features,
but they may contain scope-delimiter vertices. 

\enlargethispage{2ex}
In particular we study the following three kinds (of classes) of term graphs:
\renewcommand{\descriptionlabel}[1]%
      {\hspace{\labelsep}{\emph{#1}}}
\begin{description}\itemizeprefs
\item[$\lambda$-higher-order-term-graphs {\normalfont (Section~\ref{sec:lambdahotgs})}]
  are extensions of first-order term graphs by adding a scope function that assigns a set of vertices, its scope, to every
  abstraction vertex. There are two variants,
  one with and one without an edge (a \emph{back-link}) from each variable occurrence to
  its corresponding abstraction vertex. The class
  with back-links is related to \emph{higher-order term graphs}
  as defined by Blom in \cite{blom:2001}, 
  and in fact is an adaptation of that concept for the purpose of representing \lambdaterms. 
\item[abstraction-prefix based $\lambda$-higher-order-term-graphs {\normalfont (Section~\ref{sec:lambdaaphotgs})}]
  do not have a scope function but assign, to each vertex $\bvert$,
  an abstraction prefix consisting of a word of abstraction vertices 
  that includes those abstractions for which $\bvert$ is in their scope
  (it actually lists all abstractions for which $\bvert$ is in their `extended scope' \cite{grab:roch:2012}). 
  Abstraction prefixes are aggregations of scope information that is relevant for and locally available at individual vertices. 
\item[$\lambda$-term-graphs with scope delimiters {\normalfont (Section~\ref{sec:ltgs})}]
  are plain first-order term graphs intended to represent higher-order term graphs of the two sorts above,
  and in this way stand for \lambdaterms. 
  Instead of relying upon additional features for describing scopes, 
  they use scope-delimiter vertices to signify the end of scopes. 
  Variable occurrences as well as scoping delimiters may
  or may not have back-links to their corresponding abstraction vertices.
\end{description}
%
%
Each of these classes induces a notion of homomorphism
(functional bisimulation) and bisimulation. 
Homomorphisms increase sharing in term graphs, 
and in this way induce a sharing order.
They preserve the unfolding semantics of term graphs%
  \footnote{While this is well-known for first-order term graphs, it can also be proved
   for the higher-order term graphs considered here.},
and therefore are able to preserve \lambdaterms\ that are denoted by term graphs in the unfolding semantics. 
Term graphs from the classes we consider always represent finite or infinite \lambdaterms,
and in this sense are not `meaningless'.
But this is not shown here. Instead, we lean on motivating examples, intuitions,
and the concept of higher-order term graph from \cite{blom:2001}.  

We 
establish a bijective correspondence between the former two classes, and a correspondence between the latter
two classes that is `almost bijective' 
(bijective up to sharing or unsharing of scope delimiter vertices).
All of these correspondences preserve and reflect the sharing order.
Furthermore, we systematically investigate which specific class of \lambdatg{s}
is closed under homomorphism and renders the mentioned correspondences possible.
We prove (in Section~\ref{sec:not:closed}) that this can only hold for a class 
in which both variable-occurrence and scope-delimiter vertices
have back-links to corresponding abstractions, 
and establish (in Section~\ref{sec:closed}) that the subclass containing only 
\lambdatg{s} with eager application of scope-closure satisfies these properties.
For this class the correspondences allow us:
\begin{itemize}\itemizeprefs
\item 
  to transfer properties known for first-order term graphs, such as the existence of a maximally shared form,
  from \lambdatg{s} to the corresponding classes of higher-order \lambdatg{s};
\item 
  to implement maximal sharing for higher-order \lambdatg{s} (with eager scope
  closure) via bisimulation collapse of the corresponding first-order
  \lambdatg{s} (see algorithm in Section~\ref{sec:conclusion}).
\end{itemize}



This report is an extended version of 
the work-in-progress paper \cite{grab:roch:2013:TERMGRAPH} in the 
proceedings \cite{proc:TERMGRAPH:2013} of the workshop TERMGRAPH~2013.
Here we provide detailed proofs of the main results.
Furthermore, Section~\ref{sec:ltgs} has been extended with the concept
of `\lambdatg\ (with scope delimiters) up to \Sbisimilarity' (see Definition~\ref{def:lambdatg:up:to:S:siglambdaij})
that helps to represent \lambdahotgs\ as equivalence classes (with respect to `\Sbisimilarity') of first-order term graphs.
Finally, Section~\ref{sec:transfer} has been added here, in which the complete lattice structure 
of the sets of sharing-order successors of a given term graph with respect to the sharing order 
is transferred from the first-order \lambdatgs\ with scope delimiters
back to the higher-order \lambdahotgs.

\section{Preliminaries}
  \label{sec:prelims}

By $\nats$ we denote the natural numbers including zero.
For words $w$ over an alphabet $A$ we denote the length of $w$ by $\length{w}$.
For a function $\safun \funin A \to B $ 
we denote by $\dom{\safun}$ the domain, and by $\image{\safun}$ the image of $\safun$;
and for $A_0\subseteq A$ we denote by $\srestrictfunto{\safun}{A_0}$ the restriction of $\safun$ to $A_0$. 

Let $\asig$ be a signature with arity function $\sarity \funin \asig \to \nats$.
A \emph{term graph over $\asig$} (or a \emph{$\asig$-term graph})
is a tuple $\tuple{\verts,\svlab,\svargs,\sroot}$ 
where $\verts$ is a set of \emph{vertices},
$\svlab \funin \verts \to \asig$ the \emph{(vertex) label function},
$\svargs \funin \verts \to \verts^*$ the \emph{argument function} 
  that maps every vertex $\avert$ to the word $\vargs{\avert}$ consisting of the $\arity{\vlab{\avert}}$ successor vertices of $\avert$
  (hence it holds that $\length{\vargs{\avert}} = \arity{\vlab{\avert}}$),
and $\sroot$, the \emph{root}, is a vertex in $\verts$.
Note the fact that term graphs may have infinitely many vertices.
We say that such a term graph is \emph{root-connected}
if every vertex is reachable from the root by a path that arises by repeatedly going from a vertex to one of its successors.
We denote by $\classtgsover{\asig}$ and by $\classtgsminover{\asig}$
the class of all root-connected term graphs over $\asig$, and 
the class of all term graphs over $\asig$, respectively. 
In this notation we have already anticipated the meaning in which we will use the expression `term graph' from now on, namely as follows.

\emph{Note:}\label{note:root-connected}
  By a `term graph' we will mean, from now on, always a root-connected term graph,
  except in a few situations in which we explicitly state otherwise, and will
  refer, for example, to a term graph $\atg\in\classtgsminover{\asig}$ (for some signature $\asig$),
  which then does not need to be root-connected.

Let $\atg$ be a term graph over signature $\asig$. 
As useful notation for picking out any vertex or the $i$-th vertex 
from among the ordered successors of a vertex $\avert$ in $\atg$
we define the (not indexed) edge relation ${\stgsucc} \subseteq \verts\times\verts$,
and for each $i\in\nats$ the indexed edge relation \inMath{{\stgsucci{i}} \subseteq {\verts\times\verts}},
between vertices by stipulating that:
\begin{align*}
  \bvert \tgsucci{i} \bvertacc
    \;\;\funin\,&\Longleftrightarrow\;\;
  \existsst{\bverti{0},\ldots,\bverti{n}\in\verts}
           {\,\vargs{\bvert} = \bverti{0}\ldots\bverti{n}
                \;\logand\;
              \bvertacc = \bverti{i}}
  \\
  \bvert \tgsucc \bvertacc 
    \;\;\funin\,& \Longleftrightarrow\;\;
  \existsst{i\in\nats}{\, \bvert \stgsucci{i} \bvertacc} 
\end{align*}
holds for all $\bvert,\bvertacc\in\verts$.
We write 
$\bvert \tgsuccis{i}{f} \bvertacc$
if 
$\bvert \tgsucci{i} \bvertacc 
   \logand
 \vlab{\bvert} = f$
holds for $\bvert,\bvertacc\in\verts$, $i\in\nats$, $f\in\asig$,
to indicate the label at the source of an edge.
A \emph{path} in $\atg$ is a tuple 
$\tuple{\bverti{0},k_1,\bverti{1},k_2,\bverti{2},k_3,\ldots,k_{n-1},\bverti{n-1},k_{n},\bverti{n}}$
where $\bverti{0},\bverti{1},\bverti{2},\ldots,\bverti{n-1},\bverti{n}\in\verts$ and $n,k_1,k_2,k_3,\ldots,k_{n-1},k_{n}\in\nats$
such that 
$\bverti{0} \tgsucci{k_1} \bverti{1} \tgsucci{k_2} \bverti{2} \tgsucci{k_3} \cdots \tgsucci{k_{n}} \bverti{n}$
holds; paths will usually be denoted in the latter form, using indexed edge relations. 
An \emph{access path} of a vertex $\bvert$ of $\atg$ is
a path that starts at the root of $\atg$, ends in $\bvert$, and does not visit any vertex twice. 
Note that every vertex $\bvert$ has at least one access path: 
since every vertex in a term graph is reachable from the root, 
there is a path $\apath$ from $\sroot$ to $\bvert$;
then an access path of $\bvert$ can be obtained from $\apath$ 
by repeatedly cutting out \emph{cycles}, that is,
parts of the path between different visits to one and the same vertex. 

In the sequel, let $\atgi{1} = \tuple{\vertsi{1},\svlabi{1},\svargsi{1},\srooti{1}}$,
  $\atgi{2} = \tuple{\vertsi{2},\svlabi{2},\svargsi{2},\srooti{2}}$
be term graphs over signature $\asig$. 

A \emph{homomorphism}, also called a \emph{functional bisimulation}, 
from $\atgi{1}$ to $\atgi{2}$ 
is a morphism from the structure
$\tuple{\vertsi{1},\svlabi{1},\svargsi{1},\srooti{1}}$
to the structure
$\tuple{\vertsi{2},\svlabi{2},\svargsi{2},\srooti{2}}$,
that is, a function 
$ \sahom \funin \vertsi{1} \to \vertsi{2}$ 
such that, for all $\avert\in\vertsi{1}$ it holds:
\begin{equation}\label{eq:def:homom}
\left.\qquad
\begin{aligned}
  \ahom{\srooti{1}} 
    & = \srooti{2} 
  & & & & & (\text{roots})
  \\
    \vlabi{1}{\avert}
    & = \vlabi{2}{\ahom{\avert}} 
  & & & & & (\text{labels})  
  \\
  \funap{\sahomext}{\vargsi{1}{\avert}}
    & = \vargsi{2}{\ahom{\avert}}
  & & & & & (\text{arguments})
\end{aligned}
\qquad\right\}
\end{equation}
where $\sahomext$ is the homomorphic extension of $\sahom$ to words over $\vertsi{1}$, that is, to the function 
$\sahomext \funin \vertsi{1}^* \to \vertsi{2}^*$, 
$ \averti{1}\mcdots\averti{n}   \mapsto \ahom{\averti{1}}\mcdots\ahom{\averti{n}} $.  
In this case we write $\atgi{1} \funbisimi{\sahom} \atgi{2}$,
or $\atgi{2} \convfunbisimi{\sahom} \atgi{1}$.
And we write
$\atgi{1} \funbisim \atgi{2}$,
or for that matter $\atgi{2} \convfunbisim \atgi{1}$,
if there is a homomorphism (a functional bisimulation) from $\atgi{1}$ to $\atgi{2}$.

An \emph{isomorphism} between $\atgi{1}$ and $\atgi{2}$ is a bijective homomorphism 
$\saiso \funin \vertsi{1} \to \vertsi{2}$ from $\atgi{1}$ to $\atgi{2}$
(it follows from the homomorphism conditions \eqref{eq:def:homom}
 that then also the inverse function $\saiso^{-1} \funin \vertsi{1} \to \vertsi{2}$ is a homomorphism).
In case that there is an isomorphism between $\atgi{1}$ and $\atgi{2}$,
we write $\atgi{1} \iso \atgi{2}$, and say that $\atgi{1}$ and $\atgi{2}$ are \emph{isomorph}. 
The relation $\siso$ is an equivalence relation on $\classtgsover{\asig}$.
For every term graph $\atg$ over $\asig$ we denote the bisimulation equivalence class 
$\eqcl{\atg}{\siso}$ by (using the boldface letter) $\atgiso$. 

Let $\afunsym\in\asig$. 
An \emph{$\afunsym$\nb-homo\-mor\-phism} between $\atgi{1}$ and $\atgi{2}$ 
is a homomorphism $\sahom$ between $\atgi{1}$ and $\atgi{2}$ that identifies, or `shares',
only vertices with the label $\afunsym$, that is,
$\sahom$ has the property that
$\ahom{\bverti{1}} = \ahom{\bverti{2}} 
   \;\Rightarrow\;
 \vlabi{1}{\bverti{1}} = \vlabi{1}{\bverti{2}} = \afunsym$
holds for all $\bverti{1},\bverti{2}\in\vertsi{1}$.
If $\sahom$ is an $\afunsym$\nb-homo\-mor\-phism between $\atgi{1}$ and $\atgi{2}$,
then we write 
$\atgi{1} \funbisimis{\sahom}{\afunsym} \atgi{2}$ or $\atgi{2} \convfunbisimis{\sahom}{\afunsym} \atgi{1}$,
or dropping $\sahom$,
$\atgi{1} \funbisims{\afunsym} \atgi{2}$ or $\atgi{2} \convfunbisims{\afunsym} \atgi{1}$.


A \emph{bisimulation} between  
$\atgi{1}$ and $\atgi{2}$ 
is a term graph 
$ \atg\in\classtgsminover{\asig} $ 
such that
$ \atgi{1} \convfunbisim \atg \funbisim \atgi{2}$
holds
(note that $\atg$ does not have to be root-connected).
We write
$\atgi{1} \bisim \atgi{2}$
if there is a bisimulation between $\atgi{1}$ and $\atgi{2}$.
In this case we say that $\atgi{1}$ and $\atgi{2}$ are \emph{bisimilar}.

The notion of bisimilarity so defined
stays the same when bisimulations between $\atgi{1}$ and $\atgi{2}$ 
are restricted 
to root-connected term graphs,
or to term graphs 
$\atg = \tuple{\abisim,\svlab,\svargs,\sroot} \in\classtgsminover{\asig}$
(again restriction to $\classtgsover{\asig}$ is possible)
with $\abisim \subseteq \vertsi{1}\times\vertsi{2}$ 
 and $\sroot = \pair{\srooti{1}}{\srooti{2}}$
such that
$ \atgi{1} \convfunbisimi{\sproji{1}} \atg \funbisimi{\sproji{2}} \atgi{2}$
where $\sproji{1}$ and $\sproji{2}$ are projection functions, defined, for $i\in\setexp{1,2}$,
by $ \sproji{i} \funin \vertsi{1}\times\vertsi{2} \to \vertsi{i} $,
$\pair{\averti{1}}{\averti{2}} \mapsto \averti{i}$.
The latter restriction is closely related to the following more commonly used formulation of bisimulation.

Alternatively, bisimulations for term graphs can be defined directly 
as relations on the vertex sets, 
obtaining the same notion of bisimilarity.
In this formulation,
a bisimulation between $\atgi{1}$ and $\atgi{2}$ is
a relation $\abisim \subseteq \vertsi{1}\times\vertsi{2}$ such that
the following conditions hold, for all $\pair{\avert}{\avertacc}\in\abisim$:
\begin{center}
$  
\begin{array}{ccc}
  \pair{\srooti{1}}{\srooti{2}} \in \abisim
    & \hspace*{9ex} & (\text{roots})
  \\[0.75ex]
  \vlabi{1}{\avert} = \vlabi{2}{\avertacc} 
    & & (\text{labels})  
  \\[0.75ex]
  \pair{\vargsi{1}{\avert}}{\vargsi{2}{\avertacc}} \in \abisim^*
    & & (\text{arguments})
\end{array}
$
\end{center}
where $\abisim^* \defdby \descsetexp{\pair{\averti{1}\cdots\averti{k}}{\avertacci{1}\cdots\avertacci{k}}}
                                    {\averti{1},\ldots,\averti{k}\in\vertsi{1},\, 
                                     \avertacci{1},\ldots,\avertacci{k}\in\vertsi{2}
                                     \text{ for $k\in\nats$ such that}
                                     \pair{\averti{i}}{\avertacci{i}}\in\abisim
                                     \text{ for all $1\le i\le k$}} 
       \,$.

Let $\afunsym\in\asig$. 
An \emph{$\afunsym$\nb-bi\-si\-mu\-lation} between $\atgi{1}$ and $\atgi{2}$
is a term graph $\atg\in\classtgsminover{\asig}$ such
that
$ \atgi{1} \convfunbisims{\afunsym} \atg \funbisims{\afunsym} \atgi{2}$
holds.
If there is an $\afunsym$\nb-bi\-si\-mu\-lation $\atg$ between $\atgi{1}$ and $\atgi{2}$,
we say that $\atgi{1}$ and $\atgi{2}$ are \emph{$\afunsym$\nb-bi\-si\-mi\-lar},
and write $\atgi{1} \bisims{\afunsym} \atgi{2}$, 
or even $\atgi{1} \bisimis{\atg}{\afunsym} \atgi{2}$ indicating the bisimulation $\atg$. 

Bisimilarity $\sbisim$ is an equivalence relation on the class $\classtgsover{\asig}$ of term graphs over a signature $\asig$.
The homomorphism (functional bisimulation) relation $\sfunbisim$
is a preorder on term graphs over a given signature $\asig$.
It induces a partial order on the isomorphism equivalence classes of term graphs over $\asig$,
where anti-symmetry is implied by item~(\ref{prop:funbisim:anti:symmetric:item:ii}) of the following proposition.

The following proposition is an easy, but useful, reformulation of the definition
of homomorphism. 

\begin{proposition}\label{prop:hom:tgs}
  Let $\atgi{i} = \tuple{\vertsi{i},\svlabi{i},\svargsi{i},\srooti{i}}$,
  for $i\in\setexp{1,2}$ be term graphs over signature $\asig$.
  Let $\sahom \funin \vertsi{1} \to \vertsi{2}$ be a function. 
  Then $\atgi{1} \funbisimi{\sahom} \atgi{2}$ holds,
  that is, $\sahom$ is a homomorphism (functional bisimulation) between $\atgi{1}$ and $\atgi{2}$,
  if and only if, for all $\bvert,\bverti{1}\in\vertsi{1}$, $\bverti{2}\in\vertsi{2}$, and $k\in\nats$, 
  the following four statements hold:
  \begin{equation*}
  \begin{array}{ccc}
    \ahom{\srooti{1}} = \srooti{2}
      & & \text{\normalfont (roots)}
    \\[0.75ex]
      \vlab{\bvert} = \vlab{\ahom{\bvert}}
      & & \text{\normalfont (labels)}
    \\[0.75ex]
    \ahom{\bverti{1}} = \bverti{2}
                \;\,\logand\,\;
    \forallst{\bvertacci{1}\in\vertsi{1}} 
             {\bverti{1} \tgsucci{k} \bvertacci{1} 
                  \;\,\Rightarrow\,\;
              \existsst{\bvertacci{2}\in\vertsi{2}}
                       {\bverti{2} \tgsucci{k} \bvertacci{2}
                          \;\,\logand\,\;
                        \ahom{\bvertacci{1}} = \bvertacci{2}
                        }      
               }
      & & \text{\normalfont (arguments-forward)}
    \\[0.75ex]
    \ahom{\bverti{1}} = \bverti{2}
      \;\,\logand\,\;                
    \forallst{\bvertacci{2}\in\vertsi{2}}
             {\bverti{2} \tgsucci{k} \bvertacci{2} 
                  \;\,\Rightarrow\,\;
              \existsst{\bvertacci{1}\in\vertsi{1}}
                       {\bverti{1} \tgsucci{k} \bvertacci{1}
                          \;\,\logand\,\;
                        \ahom{\bvertacci{1}} = \bvertacci{2}
                        }      
               } 
      & & \text{\normalfont (arguments-backward)}
  \end{array}                       
  \end{equation*}
\end{proposition}

\begin{proposition}\label{prop:hom:tgs:paths}
  Let $\atgi{i} = \tuple{\vertsi{i},\svlabi{i},\svargsi{i},\srooti{i}}$,
  for $i\in\setexp{1,2}$ be term graphs over signature $\asig$.
  Let $\sahom \funin \vertsi{1} \to \vertsi{2}$ be a homomorphism between $\atgi{1}$ and $\atgi{2}$.
  Then the following statements hold:
  \begin{enumerate}[(i)]
      
    \item{}\label{prop:hom:tgs:paths:item:i}
      Every path 
      $\apath \funin \bverti{0} \tgsucci{k_1} \bverti{1} \tgsucci{k_2} \bverti{2} \tgsucci{k_3} \cdots \tgsucci{k_{n-1}} \bverti{n-1} \tgsucci{k_{n}} \bverti{n}$
      in $\atgi{1}$
      has an image $\ahom{\apath}$ under $\sahom$ in $\atgi{2}$
      in the sense that 
      $\ahom{\apath} \funin 
         \ahom{\bverti{0}} 
           \tgsucci{k_1} 
         \ahom{\bverti{1}} 
           \tgsucci{k_2} 
         \ahom{\bverti{2}}
           \tgsucci{k_3}
             \cdots
           \tgsucci{k_{n-1}} 
         \ahom{\bverti{n-1}} 
           \tgsucci{k_{n}} 
         \ahom{\bverti{n}}$.
    \item{}\label{prop:hom:tgs:paths:item:ii}
      For every $\bverti{0}\in\vertsi{0}$ and $\bvertacci{0}\in\vertsi{1}$ with $\ahom{\bverti{0}} = \bvertacci{0}$
      it holds that every path
      $\apathacc \funin 
         \bvertacci{0} 
           \tgsucci{k_1} 
         \bvertacci{1} 
           \tgsucci{k_1} 
         \bvertacci{2} 
           \tgsucci{k_2}
             \cdots 
           \tgsucci{k_{n-1}}
         \bvertacci{n-1} 
           \tgsucci{k_{n-1}} 
         \bvertacci{n}$
      in $\atgi{2}$
      has a pre-image under $\sahom$ in $\atgi{1}$ that starts in $\bverti{0}$:
      a unique unique path
      $\apath \funin 
         \bverti{0} 
           \tgsucci{k_1} 
         \bverti{1} 
           \tgsucci{k_2} 
         \bverti{2} 
           \tgsucci{k_3}
             \cdots
           \tgsucci{k_{n-1}}   
         \bverti{n-1} 
           \tgsucci{k_{n}} 
         \bverti{n}$
      in $\atgi{1}$
      such that $\apathacc = \ahom{\apath}$ holds in the sense of (\ref{prop:hom:tgs:paths:item:i}).
    \item{}\label{prop:hom:tgs:paths:item:iii}
      $\ahom{\vertsi{1}} = \vertsi{2}$, that is, $\sahom$ is surjective.
  \end{enumerate}
\end{proposition}

\begin{proof}
  Statement~(\ref{prop:hom:tgs:paths:item:i})
  can be shown by induction on the length of $\apath$, 
  using the $\text{(argument-forward)}$ property of $\sahom$ from Proposition~\ref{prop:hom:tgs}.
  Analogously, statement~(\ref{prop:hom:tgs:paths:item:ii})
  can be established by induction on the length of $\apathacc$, 
  using the $\text{(argument-backward)}$ property of $\sahom$ from Proposition~\ref{prop:hom:tgs}.
  Statement~(\ref{prop:hom:tgs:paths:item:iii})
  follows by applying, for given $\bvertacc\in\vertsi{2}$,
  the statement of (\ref{prop:hom:tgs:paths:item:ii}) to an access path $\apathacc$ of $\bvertacc$ in $\atgi{2}$
  (which exists because of our implicit assumption on term graphs to be root-connected, see page~\pageref{note:root-connected}).
\end{proof}

\begin{proposition}\label{prop:funbisim:anti:symmetric}
  Let $\atg$, $\atgi{1}$, and $\atgi{2}$ be term graphs over a signature $\asig$.
  \begin{enumerate}[(i)]
    \item{}\label{prop:funbisim:anti:symmetric:item:i}
      If $\atg \funbisimi{\sahom} \atg$, then $\sahom = \sidfunon{\verts}$, 
      where $\sidfunon{\verts}$ is the identical function on the set $\verts$ of vertices of $\atg$.
    \item{}\label{prop:funbisim:anti:symmetric:item:ii}
      If $\atgi{1} \funbisimi{\sahom} \atgi{2}$ and $\atgi{2} \funbisimi{\sbhom} \atgi{1}$ hold,
      then $\sahom$ and $\sbhom$ are invertible, $\sahom^{-1} = \sbhom$,
      and consequently $\atgi{1} \iso \atgi{2}$.   
  \end{enumerate}
\end{proposition}

\begin{proof}
  For statement~($\ref{prop:funbisim:anti:symmetric:item:i}$),
  suppose that $\atg \funbisimi{\sahom} \atg$ holds for some homomorphism $\sahom$ from $\atg$ to itself.
  The fact that $\ahom{\avert} = \idfunon{\verts}{\avert}$ holds
  for all vertices $\avert$ of $\atg$
  can be established by induction on the length of the shortest access path of $\avert$ in $\atg$.
  Note that we make use here of the root-connectedness of $\atg$ (assumed implicitly, see Section~\ref{sec:prelims})
  in the form of the assumption that every vertex can be reached by an access path.
  In order to show statement~($\ref{prop:funbisim:anti:symmetric:item:ii}$),
  note that
  $\atgi{1} \funbisimi{\sahom} \atgi{2}$ and $\atgi{2} \funbisimi{\sbhom} \atgi{1}$
  entail 
  $\atgi{1} \funbisimi{\scompfuns{\sahom}{\sbhom}} \atgi{1}$ for homomorphisms $\sahom$ and $\sbhom$ between $\atgi{1}$ and $\atgi{2}$.
  From this $\scompfuns{\sahom}{\sbhom} = \sidfunon{\vertsi{1}}$,
  where $\vertsi{1}$ the set of vertices of $\atgi{1}$, follows by ($\ref{prop:funbisim:anti:symmetric:item:i}$).
  Since analogously
  $\scompfuns{\sahom}{\sbhom} = \sidfunon{\vertsi{2}}$ follows,
  where $\vertsi{2}$ is the set of vertices of $\atgi{2}$,
  the further claims follow.
\end{proof}

We will refer to $\sfunbisim$ as the \emph{sharing preorder}, 
and to the induced relation on isomorphism equivalence classes as the \emph{sharing order}.
Note that, different from e.g.\ \cite{terese:2003}, 
we use the order relation $\sfunbisim$ in the same direction as $\le\,$:
if $\atgi{1} \funbisim \atgi{2}$, then $\atgi{2}$ is greater or equal to $\atgi{1}$ 
with respect to the ordering $\sfunbisim$ (indicating that sharing is typically increased from $\atgi{1}$ to $\atgi{2}$). 

Let $\aclass\subseteq\classtgsover{\asig}$ be a class of term graphs over signature~$\asig$. 
For $\atg\in\classtgsover{\asig}$ we will use the notation:
\begin{align}\label{eq:def:eqclin:succsofordin}
  \eqclin{\altgiso}{\sbisimsubscript}{\aclass}
    & \defdby
  \descsetexp{\altgacciso}{\altgacc\in\aclass,\, \altg \bisim \altgacc}  
  &
  \succsofordin{\altgiso}{\,\sfunbisim}{\aclass}
    & \defdby
  \descsetexp{\altgacciso}{\altgacc\in\aclass,\, \altg \funbisim \altgacc}   
\end{align}
for the bisimulation equivalence class $\eqclin{\altgiso}{\sbisimsubscript}{\aclass}$ of $\altgiso$ (the $\siso$\nb-equivalence class of $\altg$)
relativized to (the $\siso$\nb-equivalence classes of term graphs in) $\aclass$,
and respectively, the class $\succsofordin{\altgiso}{\,\sfunbisim}{\aclass}$
of all $\siso$\nb-equivalence classes that are reachable from $\altgiso$ via functional bisimulation
and that contain a term graph in $\aclass$. 
For $\aclass = \classtgsover{\asig}$ we drop the superscript $\aclass$,
and thus write
$\eqcl{\altgiso}{\sbisimsubscript}$ for $\eqclin{\altgiso}{\sbisimsubscript}{\classtgsover{\asig}}$,
and 
$\succsoford{\altgiso}{\,\sfunbisim}$ for $\succsofordin{\altgiso}{\,\sfunbisim}{\classtgsover{\asig}}$. 

A partial ordered set $\pair{\aset}{\le}$ is a \emph{complete lattice} 
if every subset of $\aset$ possesses a \txtlub\ and a \txtglb.

\begin{proposition}
    \label{prop:funbisim:succs:of:tgs:iso:complete:lattice}
  Let $\asig$ be a signature, and $\atg$ be a term graph over $\asig$.
  The bisimulation equivalence class 
  $\eqcl{\atgiso}{\sbisimsubscript}$ 
  of the isomorphism equivalence class $\atgiso$ of $\atg$
  is ordered by homomorphism $\sfunbisim$ 
  such that 
  $\pair{\succsoford{\atgiso}{\,\sfunbisim}}{\sfunbisim}$ is a complete lattice.
\end{proposition}

\begin{remark}\label{rem:prop:funbisim:succs:of:tgs:iso:complete:lattice}\normalfont
  The statement of Proposition~\ref{prop:funbisim:succs:of:tgs:iso:complete:lattice}
  is a restriction to sets of $\sfunbisim$\nb-successors of the statements 
  Theorem~3.19 in \cite{ario:klop:1996}
  and Theorem~13.2.20 in \cite{terese:2003},
  which state the complete lattice property for entire bisimulation equivalence classes (of $\siso$\nb-equivalence classes)
  of term graphs.
\end{remark}


Let $\aclass \subseteq\classtgsover{\asig}$ be a subclass of the term graphs over $\asig$, for a signature $\asig$.
We say that $\aclass$ is \emph{closed under homorphism}, or \emph{closed under functional bisimulation}, 
if $\atg \funbisim \atgacc$  
   for $\atg,\atgacc\in\classtgsover{\asig}$ with $\atg\in\aclass$
   implies $\atgacc\in\aclass$.
We say that $\aclass$ is \emph{closed under bisimulation}
if $\atg \bisim \atgacc$
   for $\atg,\atgacc\in\classtgsover{\asig}$ with $\atg\in\aclass$
   implies $\atgacc\in\aclass$.
Note these concepts are invariant under considering other signatures $\asigacc$
with $\aclass\subseteq\classtgsover{\asigacc}$.

\section{$\lambda$-higher-order-Term-Graphs}
  \label{sec:lambdahotgs}

By $\siglambda$ we designate the signature
$\setexp{ \sslapp, \sslabs }$ 
with $\arity{\sslapp} = 2$, and
     $\arity{\sslabs} = 1$. 
By $\siglambdai{i}$, for $i\in\setexp{0,1}$, we denote the extension
$\siglambda \cup \setexp{ \snlvar }$ of $\siglambda$
where $\arity{\snlvar} = i$.     
The classes of term graphs over $\siglambdai{0}$ and $\siglambdai{1}$
are denoted by $\classtgssiglambdai{0}$ and $\classtgssiglambdai{1}$, respectively. 
   
Let $\atg = \tuple{\verts,\svlab,\svargs,\sroot}$ be a term graph over a signature
extending $\siglambda$ or $\siglambdai{i}$, for $i\in\setexp{0,1}$.   
By $\vertsof{\sslabs}$ we designate the set of 
\emph{abstraction vertices} of $\atg$,
that is, the subset of $\verts$ consisting of all vertices with label $\sslabs$;
more formally,
$\vertsof{\sslabs} \defdby \descsetexp{\avert\in\verts}{\vlab{\avert} = \sslabs}$.
Analogously, the sets $\vertsof{\sslapp}$ and $\vertsof{\snlvar}$ 
of \emph{application vertices} and \emph{variable vertices} of $\atg$
are defined as the sets consisting of all vertices in $\verts$ with label $\sslapp$ or label $\snlvar$, respectively.  
Whether the variable vertices have an outgoing edge depends on the value of $i$.
The intention is to consider two variants of term graphs, one with and one
without variable back-links to their corresponding abstraction.

A `$\lambda$\nb-higher-order-term-graph'
consists of a $\siglambdai{i}$\nb-term-graph together with a scope function 
that maps abstraction vertices to their scopes (`extended scopes' in \cite{grab:roch:2012}), which are subsets of the set of vertices.
%
%

\begin{definition}[\lambdahotg]\label{def:lambdahotg}\normalfont
  Let $i\in\setexp{0,1}$. 
  A \emph{\lambdahotg}
  (short for \emph{$\lambda$\nb-higher-order-term-graph}) 
  \emph{over $\siglambdai{i}$},
  is a five-tuple $\alhotg = \tuple{\verts,\svlab,\svargs,\sroot,\sScope}$
  where $\atgi{\alhotg} = \tuple{\verts,\svlab,\svargs,\sroot}$ is a $\siglambdai{i}$\nb-term-graph,
  called the term graph \emph{underlying} $\alhotg$,
  and $\sScope \funin \vertsof{\sslabs} \to \powersetof{\verts}$
  is the \emph{scope function} of $\alhotg$ 
  (which maps an abstraction vertex $\avert$ to a set of vertices called its \emph{scope})
  that together with $\atgi{\alhotg}$ fulfills
  the following conditions:
  For all $k\in\setexp{0,1}$,
  all vertices $\bvert,\bverti{0},\bverti{1}\in\verts$,
  and all abstraction vertices $\avert,\averti{0},\averti{1}\in\vertsof{\sslabs}$
  it holds:
\begin{align*}
       \;\; & \Rightarrow\;\;
  \sroot\notin\Scopemin{\avert}
  & & \text{(root)} 
  \displaybreak[0]\\
       \;\; & \Rightarrow\;\;
  \avert \in \Scope{\avert} 
  & & \text{(self)}
  \displaybreak[0]\\ 
  \averti{1}\in\Scopemin{\averti{0}}
      \;\; & \Rightarrow\;\;
  \Scope{\averti{1}} \subseteq \Scopemin{\averti{0}} 
  & & \text{(nest)}
  \displaybreak[0]\\
  \bvert \tgsucci{k} \bverti{k}
    \;\logand\;
  \bverti{k}\in\Scopemin{\avert}
       \;\; & \Rightarrow\;\;  
  \bvert\in\Scope{\avert} 
  & & \text{(closed)}
  \displaybreak[0]\\
  \bvert\in\vertsof{\snlvar}
      \;\; & \Rightarrow\;\;
  \existsst{\averti{0}\in\vertsof{\sslabs}}{\bvert\in\Scopemin{\averti{0}}}  
  & & \text{(scope)$_0$} 
  \displaybreak[0]\\[0.5ex]
  \bvert\in\vertsof{\snlvar}
    \;\logand\;
  \bvert \tgsucci{0} \bverti{0}
      \;\; & \Rightarrow\;\;
  \left\{\hspace*{1pt}   
  \begin{aligned}[c]    
    & \bverti{0}\in\vertsof{\sslabs}
         \;\logand\;
    \\[-0.25ex]
    & \;\logand\; 
      \forall \avert\in\vertsof{\sslabs}.
    \\[-0.5ex]
    & \phantom{\;\logand\;} \hspace*{3ex}  
        (\bvert\in\Scope{\avert} \Leftrightarrow \bverti{0}\in\Scope{\avert})
  \end{aligned}  
  \right.
  & & \text{(scope)$_1$} 
\end{align*}%
where $\Scopemin{\avert} \defdby \Scope{\avert}\setminus\setexp{\avert}$.  
Note that if $i=0$, then (scope)$_1$ is trivially true and hence superfluous,
and if $i=1$, then (scope)$_0$ is redundant, because it follows from (scope)$_1$ in this case.  
For $\bvert\in\verts$ and $\avert\in\vertsof{\sslabs}$
we say that $\avert$ is a \emph{binder for} $\bvert$ if $\bvert\in\Scope{\avert}$,
and we designate by
$\binders{\bvert}$  
the set of binders of $\bvert$.

The classes of \lambdahotg{s} over $\siglambdai{0}$ and $\siglambdai{1}$ 
will be denoted by $\classlhotgsi{0}$ and $\classlhotgsi{1}$.
\end{definition}

See Figure~\ref{fig:lambdahotg} for two different \lambdahotg{s} over $\siglambdai{i}$
both of which represent the same term in the \lambdacalculus\ with $\stxtletrec$, namely
$\letrecin{\arecvar = \labs{\avar}{\lapp{(\labs{\bvar}{\lapp{\bvar}{(\lapp{\avar}{\brecvar})}})}{(\labs{\cvar}{\lapp{\brecvar}{\arecvar}})}},\; 
           \brecvar = \labs{\dvar}{\dvar}}
          {\arecvar}$.

\begin{figure}[t]
\graphs{
  \graph{\alhotgi{0}}{running_scopes_eager}
  \graph{\alhotgi{1}}{running_scopes_non_eager}
}
\vspace*{-2ex}
\caption{\label{fig:lambdahotg}
$\alhotgi{0}$ and $\alhotgi{1}$ are \lambdahotg{s} in $\classlhotgsi{i}$ whereby the
dotted back-link edges are present for $i=1$, but absent for
$i=0$. The underlying term graphs of $\alhotgi{0}$ and $\alhotgi{1}$ are
identical but their scope functions (signified by the shaded areas) differ. 
While in $\alhotgi{0}$ scopes are chosen as small as possible, which we refer to as `eager scope closure', 
in $\alhotgi{1}$ some scopes are closed only later in the graph.
}
\end{figure}

The following lemma states some basic properties of the scope function in \lambdahotg{s}.
Most importantly, scopes in \lambdahotg{s} are properly nested, in analogy with
scopes in finite \lambdaterms. 

\begin{lemma}\label{lem:lhotg}
  Let $i\in\setexp{0,1}$, and let $\alhotg = \tuple{\verts,\svlab,\svargs,\sroot,\sScope}$ be a \lambdahotg\ over $\siglambdai{i}$. 
  Then the following statements hold for all $\bvert\in\verts$ and $\avert,\averti{1},\averti{2}\in\vertsof{\sslabs}$:
  \begin{enumerate}[(i)]
    \item{}\label{lem:lhotg:item:i}
      If $\bvert\in\Scope{\avert}$, then $\avert$ is visited on every access path of $\bvert$,
      and all vertices on access paths of $\bvert$ after $\avert$ are in $\Scopemin{\avert}$.
      Hence (since $\atgi{\alhotg}$ is a term graph, every vertex has an access path) $\binders{\bvert}$ is finite.
    \item{}\label{lem:lhotg:item:ii}   
      If $\Scope{\averti{1}}\cap\Scope{\averti{2}} \neq \emptyset$ for $\averti{1} \neq \averti{2}\,$,
      then $\Scope{\averti{1}} \subseteq \Scopemin{\averti{2}}$ or $\Scope{\averti{2}} \subseteq \Scopemin{\averti{1}}$.
      As a consequence,
      if\/~$\Scope{\averti{1}} \cap \Scope{\averti{2}} \neq \emptyset$,
      then 
      $\Scope{\averti{1}} \subsetneqq \Scope{\averti{2}}$
      or $\Scope{\averti{1}} = \Scope{\averti{2}}$
      or $\Scope{\averti{2}} \subsetneqq \Scope{\averti{1}}$.
    \item{}\label{lem:lhotg:item:iii}
      If\/ $\binders{\bvert} \neq \emptyset$, 
      then $\binders{\bvert} = \setexp{\averti{0},\ldots,\averti{n}}$ 
      for $\averti{0},\ldots,\averti{n}\in\vertsof{\sslabs}$
      and $\Scope{\averti{n}} \subsetneqq \Scope{\averti{n-1}} \ldots \subsetneqq \Scope{\averti{0}}$.
  \end{enumerate}
\end{lemma}

\begin{proof}
  Let $i\in\setexp{0,1}$, and let $\alhotg = \tuple{\verts,\svlab,\svargs,\sroot,\sScope}$ 
  be a \lambdahotg\ over $\siglambdai{i}$. 
  
  For showing (\ref{lem:lhotg:item:i}),
  let $\bvert\in\verts$ and $\avert\in\vertsof{\sslabs}$ 
  be such that $\bvert\in\Scope{\avert}$.
  Suppose that 
  $\apath \funin 
     \sroot = \bverti{0} 
       \tgsucci{k_1} 
     \bverti{1} 
       \tgsucci{k_2} 
     \bverti{2} 
       \tgsucci{k_3}
         \cdots 
       \tgsucci{k_{n}} 
     \bverti{n} = \bvert$
  is an access path of $\bvert$. 
  If $\bvert = \avert$, then nothing remains to be shown.
  Otherwise $\bverti{n} = \bvert \in\Scopemin{\avert}$, and, if $n>0$, then by (closed) it follows that $\bverti{n-1}\in\Scope{\avert}$.
  This argument can be repeated to find subsequently smaller $i$
  with $\bverti{i}\in\Scope{\avert}$ and $\bverti{i+1},\ldots,\bverti{n}\in\Scopemin{\avert}$. 
  We can proceed as long as $\bverti{i}\in\Scopemin{\avert}$.
  But since, due to (root), $\bverti{0} = \sroot\notin\Scopemin{\avert}$,  
  eventually we must encounter an $i_0$ such that
  such that $\bverti{i_0+1},\ldots,\bverti{n}\in\Scopemin{\avert}$
  and $\bverti{i_0}\in\Scope{\avert}\setminus\Scopemin{\avert}$.
  This implies $\bverti{i_0} = \avert$, showing that $\avert$ is visited on $\apath$.
  
  For showing (\ref{lem:lhotg:item:ii}),
  let $\bvert\in\verts$ and $\averti{1},\averti{2}\in\vertsof{\sslabs}$,
  $\averti{1} \neq \averti{2}$ 
  be such that
  $\bvert\in\Scope{\averti{1}}\cap\Scope{\averti{2}}$.
  Let $\apath$ be an access path of $\bvert$. 
  Then it follows by (\ref{lem:lhotg:item:i}) 
  that both $\averti{1}$ and $\averti{2}$ are visited on $\apath$,
  and that, depending on whether $\averti{1}$ or $\averti{2}$ is visited first on $\apath$,
  either $\averti{2}\in\Scopemin{\averti{1}}$ or $\averti{1}\in\Scopemin{\averti{2}}$. 
  Then due to (nest) it follows that either 
  $\Scope{\averti{2}}\subseteq\Scopemin{\averti{1}}$ holds or $\Scope{\averti{1}}\subseteq\Scopemin{\averti{2}}$.
  
  Finally, statement (\ref{lem:lhotg:item:iii}) is an easy consequence of statement (\ref{lem:lhotg:item:ii}).
\end{proof}

\begin{remark}\normalfont
  The notion of \lambdahotg\ is an adaptation of the notion of 
  `higher-order term graph' by Blom \cite[Def.\ 3.2.2]{blom:2001}
  for the purpose of representing finite or infinite \lambdaterms\ 
  or cyclic \lambdaterms, that is, terms in the \lambdacalculus\ with $\stxtletrec$.
  In particular, 
  \lambdahotg{s} over $\siglambdai{1}$ correspond closely to
  higher-order term graphs over signature~$\siglambda$. 
  But they differ in the following respects:
  \renewcommand{\descriptionlabel}[1]%
      {\hspace{\labelsep}{\emph{#1{:}}}}
\small
  \begin{description}\itemizeprefs
    \item[Abstractions]
      Higher-order term graphs in \cite{blom:2001} 
      are graph representations of finite or infinite terms 
      in Combinatory Reduction Systems (\CRS{s}).
      They typically contain abstraction vertices with label~$\Box$ 
      that represent \CRS\nb-ab\-strac\-tions.
      In contrast,
      \lambdahotg{s} have abstraction vertices with label~$\sslabs$ that denote
      $\lambda$\nb-abstractions.  
    \item[Signature] 
      Whereas higher-order term graphs in \cite{blom:2001}
      are based on an arbitrary \CRS\nb-signature,
      \lambdahotg{s} over $\siglambdai{1}$ 
      only contain 
      the application symbol $\sslapp$ and the vari\-able-occur\-rence symbol $\snlvar$
      in addition to the abstraction symbol $\sslabs$. 
    \item[Variable back-links and variable occurrence vertices]  
      In the formalization of higher-order term graphs in \cite{blom:2001}
      there are no explicit vertices that represent variable occurrences.
      Instead, variable occurrences are represented by back-link edges to abstraction vertices.
      Actually, in the formalization chosen in \cite[Def.\hspace*{2pt}3.2.1]{blom:2001}, 
      a back-link edge does not directly target the abstraction vertex $\avert$ it refers to, 
      but ends at a special variant vertex $\bar{\avert}$ of $\avert$. 
      (Every such variant abstraction vertex $\bar{\avert}$ could be looked upon as a variable vertex
       that is shared by all edges that represent occurrences of the variable bound by the abstraction vertex $\avert$.) 
      
      In \lambdahotg{s} over $\siglambdai{1}$
      a variable occurrence is represented by a variable-occurrence vertex
      that as outgoing edge has a back-link to the abstraction vertex that binds the occurrence. 
    \item[Conditions on the scope function]
      While the conditions (root), (self), (nest), and (closed) on the scope function in higher-order term graphs
      in \cite[Def.\hspace*{2pt}3.2.2]{blom:2001} correspond directly to the respective conditions in Definition~\ref{def:lambdahotg},
      the difference between the condition (scope) there and (scope)$_1$ in Definition~\ref{def:lambdahotg}
      reflects the difference described in the previous item.
    \item[Free variables]      
      Whereas the higher-order term graphs in \cite{blom:2001} cater for the
      presence of free variables, free variables have been excluded from the basic format of \lambdahotg{s}. 
  \end{description}
\end{remark}

\begin{definition}[homomorphism, bisimulation]\label{def:homom:lambdahotg}\normalfont
  Let $i\in\setexp{0,1}$. 
  Let $\alhotgi{1}$ and $\alhotgi{2}$ be \lambdahotg{s} over $\siglambdai{i}$ with
  $\alhotgi{k} = \tuple{\vertsi{k},\svlabi{k},\svargsi{k},\srooti{k},\sScopei{k}}$
  for $k\in\setexp{1,2}$. 

  A \emph{homomorphism}, also called a \emph{functional bisimulation}, 
  from $\alhotgi{1}$ to $\alhotgi{2}$ 
  is a morphism from the structure
  $\tuple{\vertsi{1},\svlabi{1},\svargsi{1},\srooti{1},\sScopei{1}}$
  to the structure
  $\tuple{\vertsi{2},\svlabi{2},\svargsi{2},\srooti{2},\sScopei{2}}$,
  that is, a function 
  $ \sahom \funin \vertsi{1} \to \vertsi{2}$ 
  such that, for all $\avert\in\vertsi{1}$ 
  the conditions (labels), (arguments), and (roots) in 
  in \eqref{eq:def:homom} are satisfied, and additionally,
  for all $\avert\in\vertsiof{1}{\sslabs}$:
\begin{equation}\label{eq:def:homom:lambdahotg}
\begin{aligned}
  \funap{\sahomextext}{\Scopei{1}{\avert}}
    & = \Scopei{2}{\ahom{\avert}}
  & & & & & (\text{scope functions})
\end{aligned}
\end{equation}
where $\sahomextext$ is the homomorphic extension of $\sahom$ to sets over $\vertsi{1}$, that is, to the function 
$\bar{\sahom} \funin \powersetof{\vertsi{1}} \to \powersetof{\vertsi{2}}$, 
$ A \mapsto \descsetexp{\ahom{a}}{a\in A} $.  
If there exists a homomorphism (a functional bisimulation) $\sahom$ from $\alhotgi{1}$ to $\alhotgi{2}$,
then we write $\alhotgi{1} \funbisimi{\sahom} \alhotgi{2}$
or $\alhotgi{2} \convfunbisimi{\sahom} \alhotgi{1}$,
or, dropping $\sahom$ as subscript, 
$\alhotgi{1} \funbisim \alhotgi{2}$ or $\alhotgi{2} \convfunbisim \alhotgi{1}$.

An \emph{isomorphism} from $\alhotgi{1}$ to $\alhotgi{2}$ is a
homomorphism from $\alhotgi{1}$ to $\alhotgi{2}$ that, as a function from $\vertsi{1}$ to $\vertsi{2}$, is bijective. 
If there exists an isomorphism $\saiso \funin \vertsi{1} \to \vertsi{2}$ from $\alhotgi{1}$ to $\alhotgi{2}$,
then we say that $\alhotgi{1}$ and $\alhotgi{2}$ are \emph{isomorphic},
and write $\alhotgi{1} \isoi{\saiso} \alhotgi{2}$
or, dropping $\saiso$ as subscript, $\alhotgi{1} \iso \alhotgi{2}$.
The property of existence of an isomorphism between two \lambdahotg{s}\ forms an equivalence relation.
We denote by $\classlhotgsisoi{0}$ and $\classlhotgsisoi{1}$
the isomorphism equivalence classes of \lambdahotg{s} over $\siglambdai{0}$ and $\siglambdai{1}$, respectively. 

A \emph{bisimulation} between  
$\alhotgi{1}$ and $\alhotgi{2}$ 
is a (\lambdahotg\nb-like) structure
$ \alhotg = \tuple{\abisim,\svlab,\svargs,\sroot,\sScope}$ 
where $\tuple{\abisim,\svlab,\svargs,\sroot}\in\classtgsminover{\asig}$,
$\abisim \subseteq \vertsi{1}\times\vertsi{2}$ 
 and $\sroot = \pair{\srooti{1}}{\srooti{2}}$
such that
$ \alhotgi{1} \convfunbisimi{\sproji{1}} \alhotg \funbisimi{\sproji{2}} \alhotgi{2}$
where $\sproji{1}$ and $\sproji{2}$ are projection functions, defined, for $i\in\setexp{1,2}$,
by $ \sproji{i} \funin \vertsi{1}\times\vertsi{2} \to \vertsi{i} $,
$\pair{\averti{1}}{\averti{2}} \mapsto \averti{i}$.
If there exists a bisimulation $\abisim$ between $\alhotgi{1}$ and $\alhotgi{2}$,
then we write $\alhotgi{1} \bisimi{\abisim} \alhotgi{2}$, 
or just $\alhotgi{1} \bisim \alhotgi{2}$. 
\end{definition}

\section{Abstraction-prefix based $\lambda$-h.o.-term-graphs}
  \label{sec:lambdaaphotgs}

By an `abstraction-prefix (function) based $\lambda$\nb-higher-order-term-graph' we will
mean a term-graph over $\siglambdai{i}$ for $i\in\setexp{0,1}$ that is endowed
with a correct abstraction prefix function. Such a function $\sabspre$ maps abstraction vertices
$\bvert$ to words $\abspre{\bvert}$ consisting of all those abstraction vertices that have $\bvert$ in their scope,
in the order of their nesting from outermost to innermost abstraction vertex.
(Note again that `scope' here corresponds to `extended scope' in the terminology of \cite{grab:roch:2012}, see \ldots).
If $\abspre{\bvert} = \averti{1}\mcdots\averti{n}$,
then $\averti{1},\ldots,\averti{n}$ are the abstraction vertices that have $\bvert$ in their scope,
with $\averti{1}$ the outermost and $\averti{n}$ the innermost such abstraction vertex.

So the conceptual difference between the scope functions of \lambdahotgs\ defined in the previous section,
and \absprefix\ functions of the \lambdaaphotgs\ defined below
is the following:
A scope function $\sScope$ associates with every abstraction vertex $\avert$ 
the information on its scope $\Scope{\avert}$ and makes it available at $\avert$.
By contrast, an \absprefix\ function $\sabspre$ gathers
all scope information that is relevant to a vertex $\bvert$
(in the sense that it contains all abstraction vertices in whose scope $\bvert$ is)
and makes it available at $\bvert$ in the form $\abspre{\bvert}$. 
The fact that \absprefix\ functions make relevant scope information locally available at all vertices
leads to simpler correctness conditions.


\begin{definition}[correct abstraction-prefix function for $\siglambdai{i}$\nb-term-graphs]%
    \label{def:abspre:function:siglambdai}\normalfont
  Let $ \atg = \tuple{\verts,\svlab,\svargs,\sroot}$ be, for an $i\in\setexp{0,1}$,
  a $\siglambdai{i}$\nb-term-graph.
  
  A function $\sabspre \funin \verts \to \verts^*$
  from vertices of $\atg$ to words of vertices 
  is called an \emph{abstraction-prefix function} for $\atg$.
  Such a function is called \emph{correct} if
  for all $\bvert,\bverti{0},\bverti{1}\in\verts$ and $k\in\setexp{0,1}$:
  \begin{align*}
      \;\; & \Rightarrow \;\;
    \abspre{\sroot} = \emptyword  
    &
    (\text{root})
    \\
    \bvert\in\vertsof{\sslabs} 
      \;\logand\;
    \bvert \tgsucci{0} \bverti{0}
      \;\; & \Rightarrow \;\;
    \abspre{\bverti{0}} \prele \abspre{\bvert} \bvert
    & 
    (\sslabs)
    \\
    \bvert\in\vertsof{\sslapp}
      \;\logand\;
    \bvert \tgsucci{k} \bverti{k}
      \;\; & \Rightarrow \;\;
    \abspre{\bverti{k}} \prele \abspre{\bvert} 
    & 
    (\sslapp)
    \\
    \bvert\in\vertsof{\snlvar}
      \;\; & \Rightarrow \;\;
    \abspre{\bvert} \neq \emptyword  
    &
    (\snlvar)_{0}
    \\
    \bvert\in\vertsof{\snlvar}
      \;\logand\;
    \bvert \tgsucci{0} \bverti{0}
      \;\; & \Rightarrow \;\;
    \bverti{0}\in\vertsof{\sslabs}
      \;\logand\;
    \abspre{\bverti{0}}\bverti{0} = \abspre{\bvert}    
    & 
    (\snlvar)_{1}  
  \end{align*}
  Note that analogously as in Definition~\ref{def:lambdahotg},
  if $i=0$, then $(\snlvar)_1$ is trivially true and hence superfluous,
  and if $i=1$, then $(\snlvar)_0$ is redundant, because it follows from $(\snlvar)_1$ in this case. 

  We say that $\atg$ \emph{admits} a correct abstraction-prefix function if
  such a function exists for $\atg$.    
\end{definition}  

\begin{definition}[\lambdaaphotg]\label{def:aplambdahotg}\normalfont
  Let $i\in\setexp{0,1}$.
  A \emph{\lambdaaphotg} 
  (short for \emph{ab\-strac\-tion-pre\-fix based $\lambda$\nb-higher-order-term-graph})
  over signature $\siglambdai{i}$
  is a five-tuple $\alhotg = \tuple{\verts,\svlab,\svargs,\sroot,\sabspre}$
  where $\atgi{\alhotg} = \tuple{\verts,\svlab,\svargs,\sroot}$ is a $\siglambdai{i}$\nb-term-graph,
  called the term graph \emph{underlying} $\alhotg$,
  and $\sabspre$ is a correct \absprefix\ function for $\atgi{\alhotg}$.
  The classes of \lambdaaphotg{s} over $\siglambdai{i}$ 
  will be denoted by $\classlaphotgsi{i}$.
\end{definition}

See Figure~\ref{fig:lambdaaphotg} for two examples, which correspond, as we will see, to the 
\lambdahotg{s} in Figure~\ref{fig:lambdahotg}.

\begin{figure}[t]
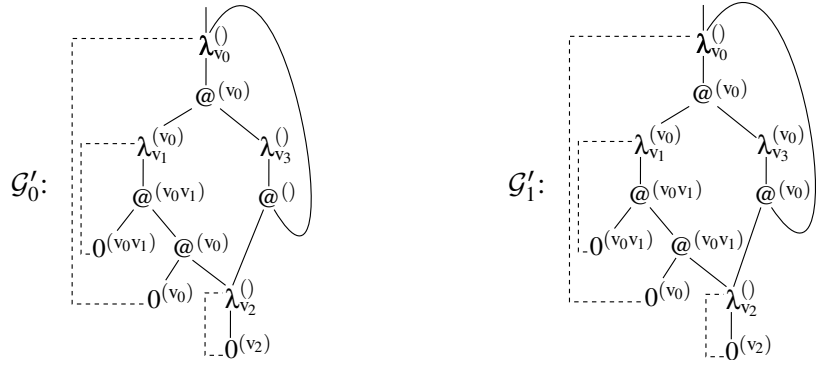

\graphs{
  \graph{\alaphotgi{0}'}{running_prefixed_ho_tg_eager}
  \graph{\alaphotgi{1}'}{running_prefixed_ho_tg_non_eager}
}
\vspace*{-2ex}
\caption{\label{fig:lambdaaphotg}
The \protect\lambdaaphotg{s} corresponding to the \protect\lambdahotg{s} in
Fig~\ref{fig:lambdahotg}. The subscripts of abstraction vertices indicate their names.
The super-scripts of vertices indicate their \absprefix{es}. A precise
formulation of this correspondence is given in
Example~\ref{ex:corr:lhotgs:laphotgs}.
}
\end{figure}

The following lemma states some basic properties of the scope function in \lambdaaphotg{s}.

\begin{lemma}\label{lem:laphotgs}
  Let $i\in\setexp{0,1}$ and let $\alaphotg = \tuple{\verts,\svlab,\svargs,\sroot,\sabspre}$ be a \lambdaaphotg\ over $\siglambdai{i}$. 
  Then the following statements hold:
  \begin{enumerate}[(i)]\itemizeprefs
    \item{}\label{lem:laphotg:item:i}
      Suppose that, for some $\avert,\bvert\in\verts$, $\avert$ occurs in $\abspre{\bvert}$.
      Then $\avert\in\vertsof{\sslabs}$, occurs in $\abspre{\bvert}$ only once,  
      and every access path of $\bvert$ passes through $\avert$, but does not end there, and thus $\bvert\neq\avert$.
      Furthermore it holds: $\stringcon{\abspre{\avert}}{\avert} \prele \abspre{\bvert}$.
      And conversely, if $\abspre{\bvert} = \stringcon{\apre}{\stringcon{\avert}{\bpre}}$ for some $\apre,\bpre\in\verts^*$, 
      then $\abspre{\avert} = \apre$.
    \item{}\label{lem:laphotg:item:ii}
      Vertices in abstraction prefixes are abstraction vertices, and hence
      $\sabspre$ is of the form $\sabspre \funin \verts\to(\vertsof{\sslabs})^*$.  
    \item{}\label{lem:laphotg:item:iii}
      For all $\avert\in\vertsof{\sslabs}$ it holds: $\avert\notin\abspre{\avert}$. 
    \item{}\label{lem:laphotg:item:iv}
      While access paths might end in vertices in $\vertsof{\snlvar}$,
      they pass only through vertices in $\vertsof{\sslabs} \cup \vertsof{\sslapp}$.
  \end{enumerate}
\end{lemma}

\begin{proof}
  Let $i\in\setexp{0,1}$ and let $\alaphotg = \tuple{\verts,\svlab,\svargs,\sroot,\sabspre}$ be a \lambdaaphotg\ over $\siglambdai{i}$. 
  
  For showing (\ref{lem:laphotg:item:i}), let $\avert,\bvert\in\verts$ be such that $\avert$ occurs in $\abspre{\bvert}$. 
  Suppose further that $\apath$
  is an access path of $\bvert$. 
  Note that when walking through $\apath$ the abstraction prefix starts out empty (due to (root)),
  and is expanded only in steps from vertices $\avertacc\in\vertsof{\sslabs}$ 
  (due to ($\sslabs$), ($\sslapp$), and ($\snlvar$)$_1$)
  in which just $\avertacc$ is added to the prefix on the right (due to ($\sslabs$)).
  Since $\avert$ occurs in $\abspre{\bvert}$, it follows that $\avert\in\vertsof{\sslabs}$, 
  that $\avert$ must be visited on $\apath$, 
  and that $\apath$ continues after the visit to $\avert$.
  That $\apath$ is an access path also entails that $\avert$ is not visited again on $\apath$,
  hence that $\bvert\neq\avert$ and that $\avert$ occurs only once in $\abspre{\bvert}$,
  and that $\stringcon{\abspre{\avert}}{\avert}$, the abstraction prefix of the successor vertex of $\avert$ on $\apath$,
  is a prefix of the abstraction prefix of every vertex that is visited on $\apath$ after $\avert$. 
  
  Statements~(\ref{lem:laphotg:item:ii}) and (\ref{lem:laphotg:item:iii}) follow directly from statement~(\ref{lem:laphotg:item:i}).
  
  For showing (\ref{lem:laphotg:item:iv}), consider an access path 
  $\apath \funin \sroot = \bverti{0} \tgsucc \cdots \tgsucc \bverti{n}$ 
  that leads to a vertex $\bverti{n}\in\vertsof{\snlvar}$.
  If $i=0$, then there is no path that extends $\apath$ properly beyond $\bverti{n}$. 
  So suppose $i=1$, 
  and let $\bverti{n+1}\in\verts$ be such that $\bverti{n} \tgsucci{0} \bverti{n+1}$. 
  Then ($\snlvar$)$_1$ implies that  
  $\abspre{\bverti{n}} = \stringcon{\abspre{\bverti{n+1}}}{\bverti{n+1}}$,
  from which it follows by (\ref{lem:laphotg:item:i})  
  that 
  $\bverti{n+1}$ is visited already on $\apath$.
  Hence $\apath$ does not extend to a longer path that is again an access path.
\end{proof}

\begin{definition}[homomorphism, bisimulation]\label{def:homom:aplambdahotg}\normalfont
  Let $i\in\setexp{0,1}$.
  Let $\alaphotgi{1}$ and $\alaphotgi{2}$ 
  be \lambdaaphotg{s} over $\siglambdai{i}$ with
  $\alaphotgi{k} = \tuple{\vertsi{k},\svlabi{k},\svargsi{k},\srooti{k},\sabsprei{k}}$
  for $k\in\setexp{1,2}$.

  A \emph{homomorphism}, also called a \emph{functional bisimulation}, 
  from $\alaphotgi{1}$ to $\alaphotgi{2}$ 
  is a morphism from the structure
  $\tuple{\vertsi{1},\svlabi{1},\svargsi{1},\srooti{1},\sabsprei{1}}$
  to the structure
  $\tuple{\vertsi{2},\svlabi{2},\svargsi{2},\srooti{2},\sabsprei{2}}$,
  that is, a function 
  $ \sahom \funin \vertsi{1} \to \vertsi{2}$ 
  such that, for all $\avert\in\vertsi{1}$ 
  the conditions (labels), (arguments), and (roots) in 
  in \eqref{eq:def:homom} are satisfied, and additionally,
  for all $\bvert\in\vertsi{1}$:
\begin{equation}\label{eq:def:homom:aplambdahotg}
\begin{aligned}
  \funap{\bar{\sahom}}{\absprei{1}{\bvert}}
    & = \absprei{2}{\ahom{\bvert}}
  & & & & & (\text{abstraction-prefix functions})
\end{aligned}
\end{equation}
where ${\bar{\sahom}}$ is the homomorphic extension of $\sahom$ to $\vertsi{1}$.
In this case we write $\alaphotgi{1} \funbisimi{\sahom} \alaphotgi{2}$,
or $\alaphotgi{2} \convfunbisimi{\sahom} \alaphotgi{1}$.
And we write
$\alaphotgi{1} \funbisim \alaphotgi{2}$,
or for that matter $\alaphotgi{2} \convfunbisim \alaphotgi{1}$,
if there is a homomorphism (a functional bisimulation) from $\alaphotgi{1}$ to $\alaphotgi{2}$.
An \emph{isomorphism} from $\alaphotgi{1}$ to $\alaphotgi{2}$ is a
homomorphism from $\alaphotgi{1}$ to $\alaphotgi{2}$ that, as a function from $\vertsi{1}$ to $\vertsi{2}$, is bijective. 
If there exists an isomorphism $\saiso \funin \vertsi{1} \to \vertsi{2}$ from $\alaphotgi{1}$ to $\alaphotgi{2}$,
then we say that $\alaphotgi{1}$ and $\alaphotgi{2}$ are \emph{isomorphic},
and write $\alaphotgi{1} \isoi{\saiso} \alaphotgi{2}$
or, dropping $\saiso$ as subscript, $\alaphotgi{1} \iso \alaphotgi{2}$.
The property of existence of an isomorphism between two \lambdahotg{s}\ forms an equivalence relation.
We denote by $\classlaphotgsisoi{0}$ and $\classlaphotgsisoi{1}$
the isomorphism equivalence classes of \lambdahotg{s} over $\siglambdai{0}$ and $\siglambdai{1}$, respectively. 

A \emph{bisimulation} between  
$\alaphotgi{1}$ and $\alaphotgi{2}$ 
is a term graph 
$ \alaphotg = \tuple{\abisim,\svlab,\svargs,\sroot,\sabspre}$ over $\asig$
with $\abisim \subseteq \vertsi{1}\times\vertsi{2}$ 
 and $\sroot = \pair{\srooti{1}}{\srooti{2}}$
such that
$ \alaphotgi{1} \convfunbisimi{\sproji{1}} \alaphotg \funbisimi{\sproji{2}} \alaphotgi{2}$
where $\sproji{1}$ and $\sproji{2}$ are projection functions, defined, for $i\in\setexp{1,2}$,
by $ \sproji{i} \funin \vertsi{1}\times\vertsi{2} \to \vertsi{i} $,
$\pair{\averti{1}}{\averti{2}} \mapsto \averti{i}$.
If there exists a bisimulation $\abisim$ between $\alhotgi{1}$ and $\alhotgi{2}$,
then we write $\alhotgi{1} \bisimi{\abisim} \alhotgi{2}$, 
or just $\alhotgi{1} \bisim \alhotgi{2}$. 
\end{definition}

The following proposition defines mappings between \lambdahotg{s} and
\lambdaaphotg{s} by which we establish a bijective correspondence between the
two classes. For both directions the underlying \lambdatg\ remains unchanged.
$\slhotgstolaphotgsi{i}$ derives an abstraction-prefix function $\sabspre$ from
a scope function by assigning to each vertex a word of its binders in the
correct nesting order.
$\slaphotgstolhotgsi{i}$ defines its scope function $\sScope$ by
assigning to each \lambda-vertex $\avert$ the set of vertices that have
$\avert$ in their prefix (along with $\avert$ since a vertex never has itself
in its abstraction prefix).

\begin{proposition}\label{prop:mappings:lhotgs:laphotgs}
  For each $i\in\setexp{0,1}$, the mappings $\slhotgstolaphotgsi{i}$ and $\slaphotgstolhotgsi{i}$ 
  are well-defined between the class of \lambdahotg{s} over $\siglambdai{i}$ 
  and the class of \lambdaaphotg{s} over $\siglambdai{i}$:
\begin{align}
& 
\left.
\begin{aligned}
  &
  \slhotgstolaphotgsi{i} \funin \classlhotgsi{i} \to \classlaphotgsi{i},
    \;\;
  \alhotg = \tuple{\verts,\svlab,\svargs,\sroot,\sScope} 
    \mapsto \lhotgstolaphotgsi{i}{\alhotg} \defdby \tuple{\verts,\svlab,\svargs,\sroot,\sabspre} 
  \\
  & \hspace*{22ex}
  \text{where }
  \sabspre \funin \verts\to\verts^*, \;\,
  \bvert \mapsto \averti{0}\mcdots\averti{n} \;
    \text{ if } \begin{aligned}[t]
                  & \binders{\bvert} \setminus \setexp{\bvert} = \setexp{\averti{0},\ldots,\averti{n}} \text{ and}
                  \\[-0.5ex]
                  & \Scope{\averti{n}} \subsetneqq \Scope{\averti{n-1}} \ldots \subsetneqq \Scope{\averti{0}}
                 \end{aligned}  
\end{aligned}
\right\} \label{eq1:prop:mappings:lhotgs:laphotgs}
\\[0.75ex]
&
\left.
\begin{aligned}
  &
  \slaphotgstolhotgsi{i} \funin \classlaphotgsi{i} \to \classlhotgsi{i},
    \;\;
  \alhotg = \tuple{\verts,\svlab,\svargs,\sroot,\sabspre} 
    \mapsto \lhotgstolaphotgsi{i}{\alhotg} \defdby \tuple{\verts,\svlab,\svargs,\sroot,\sScope} 
  \\
  & \hspace*{22ex}
  \text{where }
  \sScope \funin \vertsof{\sslabs}\to\powersetof{\verts}, \;\,
    \avert \mapsto \descsetexp{\bvert\in\verts}{\avert\text{ occurs in }\abspre{\bvert}} \cup \setexp{\avert}
\end{aligned}
\hspace*{1.7ex}\right\} \label{eq2:prop:mappings:lhotgs:laphotgs}
\end{align}
On the isomorphism equivalence classes $\classlhotgsisoi{i}$ of $\classlhotgsi{i}$ ,
and $\classlaphotgsisoi{i}$ of $\classlaphotgsi{i}$,
the functions $\slhotgstolaphotgsi{i}$ and $\slaphotgstolhotgsi{i}$ induce the functions
$\slhotgsisotolaphotgsisoi{i} \funin \classlhotgsisoi{i} \to \classlaphotgsisoi{i}$,
$\eqcl{\alhotg}{\siso} \mapsto \eqcl{\lhotgstolaphotgsi{i}{\alhotg}}{\siso}$
and
$\slaphotgsisotolhotgsisoi{i} \funin \classlaphotgsisoi{i} \to \classlhotgsisoi{i}$,
$\eqcl{\alaphotg}{\siso} \mapsto \eqcl{\laphotgstolhotgsi{i}{\alaphotg}}{\siso}$.
\end{proposition}

\begin{theorem}[correspondence of \lambdahotg{s} with \lambdaaphotg{s}]\label{thm:corr:lhotgs:laphotgs}
  For each $i\in\setexp{0,1}$ it holds that
  the mappings $\slhotgstolaphotgsi{i}$ in \eqref{eq1:prop:mappings:lhotgs:laphotgs}
  and $\slaphotgstolhotgsi{i}$ in \eqref{eq2:prop:mappings:lhotgs:laphotgs}
  are each other's inverse;
  thus they define a bijective correspondence between  
  the class of \lambdahotg{s} over $\siglambdai{i}$ 
  and the class of \lambdaaphotg{s} over~$\siglambdai{i}$.
  Furthermore, they preserve and reflect the sharing orders on $\classlhotgsi{i}$ and on $\classlaphotgsi{i}$:
  \begin{align*}
    (\forall \alhotgi{1},\alhotgi{2}\in\classlhotgsi{i}) \;\; \hspace*{-18ex} &&
      \alhotgi{1} \funbisim \alhotgi{2}
        \;\; &\Longleftrightarrow \;\;
      \lhotgstolaphotgsi{i}{\alhotgi{1}} \funbisim \lhotgstolaphotgsi{i}{\alhotgi{1}} 
    \\  
    (\forall \alaphotgi{1},\alaphotgi{2}\in\classlaphotgsi{i}) \;\; \hspace*{-18ex} &&  
       \laphotgstolhotgsi{i}{\alaphotgi{1}} \funbisim \laphotgstolhotgsi{i}{\alaphotgi{1}}
         \;\; &\Longleftrightarrow \;\;
       \alhotgi{1} \funbisim \alhotgi{2}
  \end{align*}
\end{theorem}

\begin{example}\label{ex:corr:lhotgs:laphotgs}\normalfont
The \lambdahotg{s} in Figure~\ref{fig:lambdahotg} correspond to the \lambdaaphotg{s} in Figure~\ref{fig:lambdaaphotg}
via the mappings $\slhotgstolaphotgsi{i}$ and $\slaphotgstolhotgsi{i}$ as follows: 
$
~~~~~~ \lhotgstolaphotgsi{i}{\alhotgi{0}} = \alaphotgi{0}' \, ,
~~~~~~ \lhotgstolaphotgsi{i}{\alhotgi{1}} = \alaphotgi{1}' \, ,
~~~~~~ \laphotgstolhotgsi{i}{\alaphotgi{0}} = \alhotgi{0}' \, ,
~~~~~~ \lhotgstolaphotgsi{i}{\alaphotgi{1}} = \alhotgi{1}' \, 
$.
\end{example}


For \lambdahotg{s}\ over the signature $\siglambdai{0}$ (that is,
without variable back-links) essential binding information is lost when looking
only at the underlying term graph, to the extent that \lambdaterms\ cannot be
unambiguously represented anymore. For instance the \lambdahotg{s} that
represent the \lambdaterms\ $\labs{\avar\bvar}{\lapp{\avar}{\bvar}}$ and
$\labs{\avar\bvar}{\lapp{\avar}{\avar}}$ have the same
underlying term graph. The same holds for \lambdaaphotg{s}.


This is not the case for \lambdahotg{s} (\lambdaaphotg{s}) over
$\siglambdai{1}$, because the abstraction vertex to which a variable-occurrence
vertex belongs is uniquely identified by the back-link. 
This is the reason why the following notion is only defined for the signature
$\siglambdai{1}$.

\begin{definition}[\lambdatg\ over $\siglambdai{1}$]\label{def:classltgs}\normalfont
  A term graph $\atg$ over $\siglambdai{1}$ is called a \emph{\lambdatg} over $\siglambdai{1}$
  if $\atg$ admits a correct ab\-strac\-tion-pre\-fix function. 
  By $\classltgsi{1}$ we denote the class of \lambdatg{s} over $\siglambdai{1}$.
\end{definition}
 
In the rest of this section we examine, and then dismiss, a naive approach to
implementing functional bisimulation on \lambdahotg{s} or \lambdaaphotg{s},
which is to apply the homomorphism on the underlying term graph, hoping that
this application would simply extend to the \lambdahotg\ (\lambdaaphotg)
without further ado. We demonstrate that this approach fails, concluding that a
faithful first-order implementation of functional bisimulation must not be
negligent of the scoping information.

\begin{definition}[scope- and \absprefix-forgetful mappings]\label{def:forgetful}\normalfont
  Let $i\in\setexp{0,1}$. 
  The \emph{scope-forgetful mapping}~$\sscopeforgetfullhotgsi{i}$ 
  and the \emph{\absprefix-for\-get\-ful mapping}~$\sabspreforgetfullaphotgsi{1}$
  map \lambdahotg{s} in $\classltgsi{i}$, and respectively, \lambdahotg{s} in $\classltgsi{i}$
  to their underlying term graphs: 
  \begin{center}
    $\sscopeforgetfullhotgsi{i} \funin \classlhotgsi{i} \to \classtgssiglambdai{i}\, , \;\; 
       \tuple{\verts,\svlab,\svargs,\sroot,\sScope} \mapsto \tuple{\verts,\svlab,\svargs,\sroot}$
    \\[0.5ex]
    $\sabspreforgetfullaphotgsi{i} \funin \classlaphotgsi{i} \to \classtgssiglambdai{i}\, , \;\; 
       \tuple{\verts,\svlab,\svargs,\sroot,\sabspre} \mapsto \tuple{\verts,\svlab,\svargs,\sroot}$
  \end{center}  
\end{definition}

\begin{definition}\normalfont
  Let $\alhotg$ be a \lambdahotg\ over $\siglambdai{i}$ for $i\in\setexp{0,1}$ 
  with underlying term graph $\scopeforgetfullhotgsi{i}{\alhotg}$.
  And suppose that $\scopeforgetfullhotgsi{i}{\alhotg} \funbisimi{\sahom} \atgacc$ holds
  for a term graph $\atgacc$ over $\siglambdai{i}$
  and a functional bisimulation $\sahom$.
  We say that $\sahom$ \emph{extends to a functional bisimulation on} $\alhotg$
  if $\atgacc$ can be endowed with a scope function to obtain a \lambdahotg\ $\alhotgacc$ 
  with $\scopeforgetfullhotgsi{i}{\alhotgacc} = \atgacc$ 
  and such that it holds $\alhotg \funbisimi{\sahom} \alhotgacc$. 
  
  We say that a class $\aclass$ of \lambdahotg{s} is 
  \emph{closed under functional bisimulations on the underlying term graphs}
  if for every $\alhotg\in\aclass$
    and for every homomorphism $\sahom$ on the term graph underlying $\alhotg$ 
    that witnesses $\scopeforgetfullhotgsi{i}{\alhotg} \funbisimi{\sahom} \atgacc$ for a term graph $\atgacc$
    there exists $\alhotgacc\in\aclass$ with $\scopeforgetfullhotgsi{i}{\alhotgacc} = \atgacc$
    such that $\alhotg \funbisimi{\sahom} \alhotgacc$ holds, that is,
    $\sahom$ is also a homorphism between $\alhotg$ and $\alhotgacc$.
    
  
  These notions are also extended, by analogous stipulations,
  to \lambdaaphotg{s} over $\siglambdai{i}$ for $i\in\setexp{0,1}$ and their underlying term graphs. 
\end{definition}

%
%

\begin{proposition}\label{prop:no:extension:funbisim:siglambda1}
  Neither the class $\classlhotgsi{1}$ of \lambdahotg{s} nor the class $\classlaphotgsi{1}$ of \lambdaaphotg{s} 
  is closed under functional bisimulations on the underlying term graphs.

\end{proposition}

\begin{proof}
  In view of Theorem~\ref{thm:corr:lhotgs:laphotgs} it suffices to show the statement for $\classlhotgsi{1}$.
  We show that not every functional bisimulation on the term graph underlying 
  a \lambdahotg\ over $\siglambdai{1}$ 
  extends to a functional bisimulation on the higher-order term graphs.
  Consider the following term graphs $\atgi{0}$ and $\atgi{1}$ over $\siglambdai{1}\hspace*{1pt}$
  (at first, please ignore the scope shading):
  \graphs
  {
          \vcentered{$\atgi{1}$:~~} \vcentered{\includegraphics[scale=0.80]{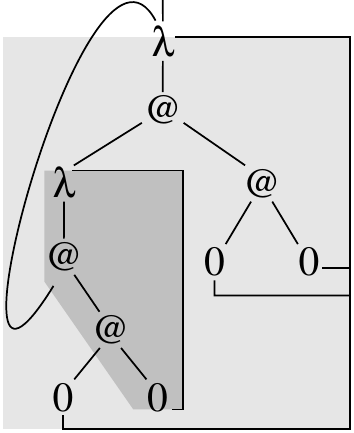}}\hfill
          \graph{\atgi{0}}{no_extension_funbisim_siglambda1_g0}
  }%
  There is an obvious homomorphism $\sahom$ that witnesses 
  \mbox{$\atgi{1} \funbisimi{\sahom} \atgi{0}$}.
  Both of these term graphs extend to \lambdahotg{s} by suitable scope functions
  (one possibility per term graph is indicated by the scope shadings above; $\atgi{1}$ actually admits two possibilities). 
  However, $\sahom$ does not extend to any of the \lambdahotg{s} $\alhotgi{1}$ and $\alhotgi{0}$ 
  that extend $\atgi{1}$ and $\atgi{0}$, respectively. 
\end{proof}

The next proposition is merely a reformulation of Proposition~\ref{prop:no:extension:funbisim:siglambda1}.

\begin{proposition}\label{prop:forgetful}
  The scope-forgetful mapping $\sscopeforgetfullhotgsi{1}$ on $\classlhotgsi{1}$ 
  and the \absprefix-for\-get\-ful mapping $\sabspreforgetfullaphotgsi{1}$ on $\classlaphotgsi{1}$ 
  preserve, but do not reflect, the sharing orders on these classes. In particular:
  \begin{align*}
    (\forall \alaphotgi{1},\alaphotgi{2}\in\classlaphotgsi{1}) \;\; \hspace*{-14ex} &&
       \alaphotgi{1} \funbisim \alaphotgi{2}
         \;\: & \Longrightarrow\;\:
       \abspreforgetfullaphotgsi{1}{\alaphotgi{1}} \funbisim \abspreforgetfullaphotgsi{1}{\alaphotgi{2}}
    \\[-0.25ex]
    (\exists \alaphotgi{1},\alaphotgi{2}\in\classlaphotgsi{1}) \;\; \hspace*{-14ex} &&
       \alaphotgi{1} \not\funbisim \alaphotgi{2}
         \;\: & \logand \;\:
       \abspreforgetfullaphotgsi{1}{\alaphotgi{1}} \funbisim \abspreforgetfullaphotgsi{1}{\alaphotgi{2}}
  \end{align*}
  
\end{proposition}

As a consequence of this proposition it is not possible to faithfully implement functional
bisimulation on \lambdahotg{s} and \lambdaaphotg{s} by only considering the underlying term graphs,
and in doing so neglecting the scoping information          
from the scope function, or respectively, from the abstraction prefix function.%
         \footnote{
           In the case of $\siglambdai{1}$ implicit information about possible scopes is being kept,
           due to the presence of back-links from variable occurrence vertices to abstraction vertices.
           But this is not enough for reflecting the sharing order under the forgetful mappings.}
In order to yet be able to implement functional bisimulation of
\lambdahotg{s} and \lambdaaphotg{s} in a first-order setting, in the next
section we introduce a class of first-order term graphs that accounts for
scoping by means of scope delimiter vertices.


\section{$\lambda$-Term-Graphs with Scope Delimiters}
  \label{sec:ltgs}

For all $i\in\setexp{0,1}$ and $j\in\setexp{1,2}$ we define the extensions
$\siglambdaij{i}{j} \defdby \siglambda \cup \setexp{\snlvar, \snlvarsucc}$ 
of the signature $\siglambda$
where $\arity{\snlvar} = i$ and $\arity{\snlvarsucc} = j\,$,
and we denote the class of term graphs over signature $\siglambdaij{i}{j}$ by $\classtgssiglambdaij{i}{j}$.

Let $\atg$ be a term graph with vertex set $\verts$ over a signature extending $\siglambdaij{i}{j}$ 
for $i\in\setexp{0,1}$ and $j\in\setexp{1,2}$.
We denote by $\vertsof{\snlvarsucc}$ the subset of $\verts$ consisting of all vertices
with label $\snlvarsucc$, which are called the \emph{(scope) delimiter vertices} of $\atg$.
Delimiter vertices signify the end of an `extended scope' \cite{grab:roch:2012}. 
They are analogous to occurrences of function symbols $\snlvarsucc$
in representations of \lambdaterms\ in a nameless 
de-Bruijn index \cite{de-bruijn:72} form in which Dedekind numerals
based on $\snlvar$ and the successor function symbol $\snlvarsucc$ are used 
(this form is due to Hendriks and van~Oostrom, see also \cite{oost:looi:zwit:2004},
 and is related to their end-of-scope symbol $\adbmal$ \cite{hend:oost:2003}).

Analogously as for the classes $\classlhotgsi{i}$ and $\classlaphotgsi{i}$,
the index $i$ will determine whether in correctly formed \lambdatg{s} (defined below)
variable vertices have back-links to the corresponding abstraction. 
Here additionally scope-delimiter vertices have such back-links (if $j=2$) or not (if $j=1$). 

\begin{definition}[correct abstraction-prefix function for $\siglambdaij{i}{j}$-term-graphs]%
    \label{def:abspre:function:siglambdaij}\normalfont
  Let $\atg = \tuple{\verts,\svlab,\svargs,\sroot}$ be a $\siglambdaij{i}{j}$\nb-term-graph
  for an $i\in\setexp{0,1}$ and an $j\in\setexp{1,2}$.
  
  A function $\sabspre \funin \verts \to \verts^*$
  from vertices of $\atg$ to words of vertices 
  is called an \emph{abstraction-prefix function} for $\atg$.
  Such a function is called \emph{correct} if
  for all $\bvert,\bverti{0},\bverti{1}\in\verts$ and $k\in\setexp{0,1}$ it holds:
  \begin{align*}
      \;\; & \Rightarrow \;\;
    \abspre{\sroot} = \emptyword  
    &
    (\text{root})
    \displaybreak[0]\\
    \bvert\in\vertsof{\sslabs} 
      \;\logand\;
    \bvert \tgsucci{0} \bvertbp{0}{}
      \;\; & \Rightarrow \;\;
    \abspre{\bvertbp{0}{}} = \abspre{\bvert} \bvert 
    & 
    (\sslabs)
  \displaybreak[0]\\
    \bvert\in\vertsof{\sslapp}
      \;\logand\;
    \bvert \tgsucci{k} \bvertbp{k}{}
      \;\; & \Rightarrow \;\;
    \abspre{\bvertbp{k}{}} = \abspre{\bvert} 
    & 
    (\sslapp)
  \displaybreak[0]\\  
    \bvert\in\vertsof{\snlvar}
      \phantom{\;\;\logand\;\bvert \tgsucci{0} \bvertacc\;}
        \;\; & \Rightarrow \;\;
     \abspre{\bvert} \neq \emptyword   
    & 
    (\snlvar)_0
  \displaybreak[0]\\
    \bvert\in\vertsof{\snlvar}
      \;\logand\;
    \bvert \tgsucci{0} \bvertbp{0}{}
      \;\; & \Rightarrow \;\;
    \begin{aligned}[c]  
      & 
        \bvertbp{0}{}\in\vertsof{\sslabs}
      \;\logand\;
      \abspre{\bvertbp{0}{}} \bvertbp{0}{} = \abspre{\bvert} 
    \end{aligned}
    & 
    (\snlvar)_1 
    \displaybreak[0]\\
    \bvert\in\vertsof{\snlvarsucc}
      \;\logand\;
    \bvert \tgsucci{0} \bvertbp{0}{}
        \;\; & \Rightarrow \;\;
    \abspre{\bvertbp{0}{}} \avert
      =
    \abspre{\bvert}
      \;\;   
    \text{for some $\avert\in\verts$} \hspace*{1ex}
    &
    (\snlvarsucc)_1 
  \displaybreak[0]\\
    \bvert\in\vertsof{\snlvarsucc}
      \;\logand\;
    \bvert \tgsucci{1} \bvertbp{1}{}
        \;\; & \Rightarrow \;\;
    \bvertbp{1}{} \in \vertsof{\sslabs}
      \;\logand\;        
    \abspre{\bvertbp{1}{}} \bvertbp{1}{}
      =
    \abspre{\bvert}
    &
    (\snlvarsucc)_2  
  \end{align*}%
  Note that analogously as in Definition~\ref{def:lambdahotg} and in Definition~\ref{def:aplambdahotg},
  if $i=0$, then $(\snlvar)_1$ is trivially true and hence superfluous,
  and if $i=1$, then $(\snlvar)_0$ is redundant, because it follows from $(\snlvar)_1$ in this case. 
  Additionally, if $j=1$, then $(\snlvarsucc)_2$ is trivially true and therefore superfluous.         
\end{definition}

\begin{definition}[{\lambdatg\ over $\siglambdaij{i}{j}$}]%
    \label{def:lambdatg:siglambdaij}\normalfont
  Let $i\in\setexp{0,1}$ and $j\in\setexp{1,2}$.
  A \emph{\lambdatg\ (with scope-de\-li\-mi\-ters)} over $\siglambdaij{i}{j}$
  is a $\siglambdaij{i}{j}$\nb-term-graph
  that admits a correct ab\-strac\-tion-pre\-fix function. 
  The class of \lambdatg{s} over $\siglambdaij{i}{j}$ is denoted by $\classltgsij{i}{j}$,
  and the class of isomorphism equivalence classes of \lambdatg{s} over $\siglambdaij{i}{j}$ by $\classltgsisoij{i}{j}$.
\end{definition}

See Figure~\ref{fig:lambdatg} for examples,
that, as we will see, correspond to the ho-term-graphs in Figure~\ref{fig:lambdahotg} and in Figure~\ref{fig:lambdaaphotg}.

\begin{figure}[t]
\graphs{
  \graph{\atgi{0}}{running_snodes_eager}
  \graph{\atgi{1}}{running_snodes_non_eager}
}
\vspace*{-2ex}
\caption{\label{fig:lambdatg}%
The \protect\lambdatg{s} corresponding to the \protect\lambdaaphotg{s} from
Fig~\ref{fig:lambdaaphotg} and the \protect\lambdahotg{s} from Fig~\ref{fig:lambdahotg}.
A precise formulation of this correspondence is given in
Example~\ref{ex:corr:laphotgs:ltgs}.
}
\end{figure}

\vspace{0.5ex}
The following lemma states some basic properties of \lambdatg{s}.

\begin{lemma}\label{lem:ltg}
  Let $i\in\setexp{0,1}$, and $j\in\setexp{1,2}$.
  Let $\altg = \tuple{\verts,\svlab,\svargs,\sroot}$ be a \lambdatg\ over $\siglambdaij{i}{j}$,
  and let $\sabspre$ be a correct \absprefix\ function for $\altg$.
  Then the statements (\ref{lem:laphotg:item:i})--(\ref{lem:laphotg:item:iii}) in Lemma~\ref{lem:laphotgs} hold
  as items (i)--(iii) of this lemma, 
  and additionally: 
  \begin{enumerate}[(i)]\setcounter{enumi}{3}\itemizeprefs
    \item{}\label{lem:ltg:item:iv}
      When stepping through $\altg$ along a path,
      the \absprefix\ function $\sabspre$ behaves like a stack data structure, in the sense that
      in a step $\bvert \tgsucci{i} \bvertacc$ 
                                               between vertices $\bvert$ and $\bvertacc$ it holds:
      \begin{itemize}
        \item 
          if $\bvert\in\vertsof{\sslabs}$, 
          then 
          $\abspre{\bvertacc} = \abspre{\bvert} \bvert$,
          that is,
          $\abspre{\bvertacc}$ is obtained by adding
          $\bvert$ to $\abspre{\bvert}$ on the right
          ($\bvert$ is `pushed on the stack'); 
        \item
          if $\bvert\in\vertsof{\snlvar}\cup\vertsof{\snlvarsucc}$, 
          then 
          $\abspre{\bvertacc} \avert = \abspre{\bvert}$ for some $\avert\in\vertsof{\sslabs}$,
          that is,
          $\abspre{\bvertacc}$ is obtained by removing
          the leftmost vertex $\avert$ from $\abspre{\bvert}$
          ($\avert$ is `popped from the stack');
        \item
          if $\bvert\in\vertsof{\sslapp}$,
          then
          $\abspre{\bvertacc} = \abspre{\bvert}$,
          that is, the value of the \absprefix\ function stays unchanged
          (the stack stays unchanged). 
      \end{itemize}       
    \item\label{lem:ltg:item:v}
      Access paths may end in vertices in $\vertsof{\snlvar}$,
      but only pass through vertices in $\vertsof{\sslabs} \cup \vertsof{\sslapp} \cup \vertsof{\snlvarsucc}$,
      and depart from vertices in $\vertsof{\snlvarsucc}$ only via indexed edges $\stgsuccis{0}{\snlvarsucc}$.
    \item\label{lem:ltg:item:vi}
      There exists precisely one correct ab\-strac\-tion-pre\-fix function on $\altg$.
  \end{enumerate}
\end{lemma}

\begin{proof}
  That also here statements (\ref{lem:laphotg:item:i})--(\ref{lem:laphotg:item:iii}) in Lemma~\ref{lem:laphotgs} hold,
  and that statement~(\ref{lem:ltg:item:v}) holds,
  can be shown analogously as in the proof of the respective items of Lemma~\ref{lem:laphotgs}. 
  Statement~\ref{lem:laphotg:item:iv} is easy to check from the definition (see Definition~\ref{def:abspre:function:siglambdaij} 
  of a correct \absprefix\ function for a \lambdatg).
  For (\ref{lem:ltg:item:vi}) it suffices to observe
  that if $\sabspre$ is a correct \absprefix\ function for $\altg$,
  then, for all $\bvert\in\verts$, the value $\abspre{\bvert}$ of $\sabspre$ at $\bvert$
  can be computed by choosing an arbitrary access path $\apath$ from $\sroot$ to $\bvert$
  and using the conditions $(\sslabs)$, $(\sslapp)$, and $(\snlvarsucc)_0$  
  to determine in a stepwise manner the values of $\sabspre$ at the vertices that are visited on $\apath$.
  Hereby note that in every transition along an edge on $\apath$ the length of the abstraction prefix only changes by at most 1. 
\end{proof}

Lemma~\ref{lem:ltg}, (\ref{lem:ltg:item:vi}) allows us now
to speak of \emph{the \absprefix\ function of} a \lambdatg. 

Although the requirement of the existence of a correct \absprefix\ function
restricts their possible forms, 
\lambdatgs\ are first-order term graphs,
and as such the definitions of homomorphism, isomorphism, and bisimulation
for first-order term graphs from Section~\ref{sec:prelims} apply to them. 


The question arises whether indeed, and if so then how, a homomorphism $\sahom$ between 
\lambdatgs\ $\altgi{1}$ and $\altgi{2}$ relates their \absprefix\ functions $\sabsprei{1}$ and $\sabsprei{2}$.
It turns out, see Proposition~\ref{prop:hom:image:absprefix:function:ltgs} below,
that $\sabsprei{2}$ is the `homomorphic image' under $\sahom$ of $\sabsprei{1}$
in the sense as defined by the following definition.

\begin{definition}\label{def:hom:image:absprefix:function:ltgs}\normalfont
   Let $\atgi{1} = \tuple{\vertsi{1},\svlabi{1},\svargsi{1},\srooti{1}}$ 
   and $\atgi{2} = \tuple{\vertsi{2},\svlabi{2},\svargsi{2},\srooti{2}}$ 
   be term graphs over $\siglambdaij{i}{j}\,$, for some $i\in\setexp{0,1}$ and $j\in\setexp{1,2}$,
   and let $\sabsprei{1}$ and $\sabsprei{2}$ be \absprefix\ functions
   (not required to be correct) for $\atgi{1}$ and $\atgi{2}$, respectively. 
   Furthermore, let $\sahom \funin \vertsi{1} \to \vertsi{2}$ be a homomorphism from $\altgi{1}$ to $\altgi{2}$. 
   We say that $\sabsprei{2}$ is the \emph{homomorphic image of} $\sabsprei{1}$ under $\sahom$
   if it holds: 
  \begin{equation}\label{eq:def:hom:image:absprefix:function:ltgs}
    \scompfuns{\bar{\sahom}}{\sabsprei{1}}
      =
    \scompfuns{\sabsprei{2}}{\sahom}
  \end{equation}(i.e.,
  $\funap{\bar{\sahom}}{\absprei{1}{\bvert}}
      = 
   \absprei{2}{\ahom{\bvert}}$
  for all $\bvert\in\vertsi{1}$),
  where $\bar{\sahom}$ is the homomorphic extension of $\sahom$ to words over $\vertsi{1}$.
\end{definition}

Note that for the \lambdaaphotgs\ defined in the previous section
it was the case by definition
that every homomorphism maps the \absprefix\ function of the source \lambdaaphotg\ 
to the \absprefix\ function of the target \lambdaaphotg. 
So if $\sahom$ is a homomorphism between \lambdaaphotgs\
$ \alaphotgi{1} = \tuple{\vertsi{1},\svlabi{1},\svargsi{1},\srooti{1},\sabsprei{1}}$
and
$ \alaphotgi{2} = \tuple{\vertsi{2},\svlabi{2},\svargsi{2},\srooti{2},\sabsprei{2}}$,
then \eqref{eq:def:hom:image:absprefix:function:ltgs} 
was required to hold (see \eqref{eq:def:homom:aplambdahotg} in Definition~\ref{def:homom:aplambdahotg}). 
This condition was a consequence of the fact that homomorphisms between \lambdaaphotgs\
were defined as morphisms between \lambdaaphotgs\ when viewed as algebraical structures,
because \absprefix\ functions are parts of \lambdaaphotgs\ that have to be respected by morphisms. 

Yet for the \lambdatgs\ defined here, 
\absprefix\ functions do not make part of their underlying formalization as first-order term graphs. 
The existence of an \absprefix\ function is 
a mathematical property that separates \lambdatgs\
from among the first-order term graphs (over one of the signatures $\siglambdaij{i}{j}$). 
Homomorphisms between \lambdatgs\ are defined as morphisms between the
structures that underlie their formalization, first-order term graphs, and therefore these homomorphisms
do not by definition respect \absprefix\ functions. 

However, it turns out that, as a consequence of the specific form of the correctness conditions
for \lambdatgs, homomorphisms between \lambdatgs\ do in fact respect their (unique, see above) \absprefix\ functions.
This fact is stated formally by the following proposition.
The indicated reason will become clear in the proof of Lemma~\ref{lem:hom:image:absprefix:function:ltgs}
which we use to establish this proposition.

\begin{proposition}\label{prop:hom:image:absprefix:function:ltgs}
  Let $\atgi{1}$ 
  and $\atgi{2}$ 
  be \lambdatgs,
  and let $\sabsprei{1}$ and $\sabsprei{2}$ be their (correct) \absprefix\ functions, respectively. 
  Suppose that $\sahom$ is a homomorphism from $\atgi{1}$ to $\atgi{2}$.
  Then $\sabsprei{2}$ is the homomorphic image of $\sabsprei{1}$ under $\sahom$.
\end{proposition}

This proposition follows from the
statement in item~(\ref{lem:hom:image:absprefix:function:ltgs:item:ii}) of the lemma below. 
Item~(\ref{lem:hom:image:absprefix:function:ltgs:item:i}) of the lemma states
that the correctness of \absprefix\ functions is preserved and reflected by homomorphisms on term graphs over $\siglambdaij{i}{j}\,$.

\begin{lemma}\label{lem:hom:image:absprefix:function:ltgs}
  Let $\atgi{1} = \tuple{\vertsi{1},\svlabi{1},\svargsi{1},\srooti{1}}$
  and $\atgi{2} = \tuple{\vertsi{2},\svlabi{2},\svargsi{2},\srooti{2}}$
  be term graphs over $\siglambdaij{i}{j}$ for some $i\in\setexp{0,1}$ and $j\in\setexp{1,2}$.
  Let $\sahom \funin \vertsi{1} \to \vertsi{2}$ be a homomorphism from $\atgi{1}$ to $\atgi{2}$. 
  Furthermore, let $\sabsprei{1}$ and $\sabsprei{2}$ be their \absprefix\ functions 
  (not required to be correct) for $\atgi{1}$ and $\atgi{2}$, respectively.
  Then the following two statements hold:
  \begin{enumerate}[(i)]
    \item{}\label{lem:hom:image:absprefix:function:ltgs:item:ii}
      If $\sabsprei{1}$ and $\sabsprei{2}$ are correct for $\atgi{1}$ and $\atgi{2}$, respectively,
      then $\sabsprei{2}$ is the homomorphic image of $\sabsprei{1}$.
    \item{}\label{lem:hom:image:absprefix:function:ltgs:item:i}
      If $\sabsprei{2}$ is the homomorphic image of $\sabsprei{1}$,
      then $\sabsprei{1}$ is correct for $\atgi{1}$ if and only if $\sabsprei{2}$ is correct for $\atgi{2}$.
  \end{enumerate}
 
\end{lemma}

%
%

\begin{proof}
  Let $\atgi{1}$, $\atgi{2}$, $\sahom$, $\sabsprei{1}$, and $\sabsprei{2}$ as assumed in the lemma.
  
  For showing statement~(\ref{lem:hom:image:absprefix:function:ltgs:item:ii}),
  we assume that the \absprefix\ functions $\sabsprei{1}$ and $\sabsprei{2}$ are correct for $\atgi{1}$ and $\atgi{2}$, respectively. 
  We establish that $\sabsprei{2}$ is the homomorphic image under $\sahom$ of $\sabsprei{1}$ by
  proving that
  \begin{equation}\label{eq2:prf:lem:hom:image:absprefix:function:ltgs}
    \forallst{\bvert\in\vertsi{1}}
             {\;
              \compfuns{\bar{\sahom}}{\sabsprei{1}}{\bvert} = \compfuns{\sabsprei{2}}{\sahom}{\bvert}}
  \end{equation}
  holds, by induction on the length of an access path
  $\apath 
     \funin
   \srooti{1} = \bverti{0}
     \tgsucc
   \bverti{1}
     \tgsucc
   \ldots
     \tgsucc
   \bverti{n} = \bvert$        
   of $\bvert$ in $\atgi{1}$.
  
  If $\length{\apath} = 0$, then $\bvert = \srooti{1}$, and it follows:
  $\funap{\bar{\sahom}}{\absprei{1}{\bvert}}
     =
   \funap{\bar{\sahom}}{\absprei{1}{\srooti{1}}}
     =
   \funap{\bar{\sahom}}{\emptyword}
     =
   \emptyword
     =
   \absprei{2}{\srooti{2}}
     =
   \absprei{2}{\ahom{\srooti{1}}}$            
  by using the correctness conditions ($\text{root}$) for $\sabsprei{1}$ and $\sabsprei{2}$,
  and the condition ($\text{root}$) for $\sahom$ to be a homomorphism.
  
  If $\length{\apath} = n+1$, 
  then $\apath$ is of the form 
  $\apath 
     \funin
   \srooti{1} = \bverti{0}
     \tgsucc
   \bverti{1}
     \tgsucc
   \ldots
     \tgsucc
   \bverti{n}
     \tgsucci{i} 
   \bverti{n+1} = \bvert$ 
   for some $i\in\setexp{0,1}$.
   We have to show that 
   $\compfuns{\bar{\sahom}}{\sabsprei{1}}{\bvert} 
      = 
    \compfuns{\sabsprei{2}}{\sahom}{\bvert}$
   holds, and in doing so we may use  
   $\funap{\bar{\sahom}}{\absprei{1}{\bverti{n}}} 
      = 
    \absprei{2}{\ahom{\bverti{n}}}$  
   as this follows by the induction hypothesis.
   We will do so by distinguishing the by Lemma~\ref{lem:ltg},~(\ref{lem:ltg:item:v}) only possible three cases 
   of different labels for the vertex $\bverti{n}$, namely $\sslabs$, $\sslapp$, and $\snlvarsucc$.   
   In each of these cases we will also use that
   $\ahom{\bverti{n}} 
      \tgsuccacci{i}
    \ahom{\bverti{n+1}} = \ahom{\bvert}$,
   where $\stgsuccacc$ the directed-edge relation in $\atgi{2}$, 
   and
   $\vlabi{1}{\bverti{n}}
      = 
    \vlabi{2}{\ahom{\bverti{n}}}$
   holds, 
   which follow since $\sahom$ is a homomorphism.  
   
   Suppose that $\bverti{n}\in\vertsiof{1}{\sslabs}$.
   Then also $\ahom{\bverti{n}}\in\vertsiof{2}{\sslabs}$ holds, since $\sahom$ is a homomorphism.
   From the correctness conditions ($\sslabs$) for $\bverti{n}$ in $\atgi{1}$
   and for $\ahom{\bverti{n}}$ in $\atgi{2}$ it follows that 
   $\absprei{1}{\bvert} = \stringcon{\absprei{1}{\bverti{n}}}{\bverti{n}}$
   and
   $\absprei{2}{\ahom{\bvert}} = \stringcon{\absprei{2}{\ahom{\bverti{n}}}}{\ahom{\bverti{n}}}$.
   From this we now obtain
   $\funap{\bar{\sahom}}{\absprei{1}{\bvert}}
      =
    \funap{\bar{\sahom}}{\stringcon{\absprei{1}{\bverti{n}}}{\bverti{n}}}
      =
    \stringcon{\funap{\bar{\sahom}}{\absprei{1}{\bverti{n}}}}{\ahom{\bverti{n}}}
      =    
    \stringcon{\absprei{2}{\funap{\bar{\sahom}}{{\bverti{n}}}}}{\ahom{\bverti{n}}}
      =
    \absprei{2}{\ahom{\bvert}}$
   by using the induction hypothesis for $\bverti{n}$. 
   
   Suppose that $\bverti{n}\in\vertsiof{1}{\sslapp}$.
   Then $\ahom{\bverti{n}}\in\vertsiof{2}{\sslapp}$ follows as $\sahom$ is a homomorphism.
   From the correctness conditions ($\sslapp$) for $\bverti{n}$ in $\atgi{1}$
   and for $\ahom{\bverti{n}}$ in $\atgi{2}$
   it follows that  
   $\absprei{1}{\bvert} = \absprei{1}{\bverti{n}}$
   and
   $\absprei{2}{\ahom{\bvert}} = \absprei{1}{\ahom{\bverti{n}}}$.
   Then we obtain
   $\funap{\bar{\sahom}}{\absprei{1}{\bvert}}
      =
    \funap{\bar{\sahom}}{\absprei{1}{\bverti{n}}} 
      =
    \absprei{2}{\ahom{\bverti{n}}}
      =
    \absprei{2}{\ahom{\bvert}}$
   by using the induction hypothesis for $\bverti{n}$. 
   
   Finally, suppose that $\bverti{n}\in\vertsiof{1}{\snlvarsucc}$.
   Then also $\ahom{\bverti{n}}\in\vertsiof{2}{\snlvarsucc}$ follows.
   In this case we have $i=0$, and hence the final step in $\apath$ is of the form
   $\bverti{n} 
      \tgsucci{0}
    \bverti{n+1} = \bvert$,
   which also entails
   $\ahom{\bverti{n}} 
      \tgsuccacci{0}
    \ahom{\bverti{n+1}} = \ahom{\bvert}$.
   From the correctness conditions ($\snlvarsucc$) for $\bverti{n}$ in $\atgi{1}$
   we obtain that there is a vertex $\avert\in\vertsi{1}$ such that
   $\absprei{1}{\bverti{n}}
      =
    \stringcon{\absprei{1}{\bvert}}{\avert}$ holds.
   Then by the induction hypothesis for $\bverti{n}$ it follows:
   $\absprei{2}{\ahom{\bverti{n}}}
      =
    \funap{\bar{\sahom}}{\absprei{1}{\bverti{n}}}
      =
    \funap{\bar{\sahom}}{\stringcon{\absprei{1}{\bvert}}{\avert}}  
      =
    \stringcon{\funap{\bar{\sahom}}{\absprei{1}{\bvert}}}{\ahom{\avert}}$.
   From this 
   $\absprei{2}{\ahom{\bvert}}
      =
    \funap{\bar{\sahom}}{\absprei{1}{\bvert}}$
   follows, 
   due to
   the correctness condition ($\snlvarsucc$) for $\ahom{\bverti{n}}$.

  In this way we have carried out the proof by induction of \eqref{eq2:prf:lem:hom:image:absprefix:function:ltgs},
  and thereby have shown that $\sabsprei{2}$ is the homomorphic image of $\sabsprei{1}$.

  \vspace{0.5ex}  
  For showing statement (\ref{lem:hom:image:absprefix:function:ltgs:item:i}),
  we suppose that $\sabsprei{2}$ is the homomorphic image of $\sabsprei{1}$. 
  
  We first argue for the direction from left to right in the equivalence stated in (\ref{lem:hom:image:absprefix:function:ltgs:item:i});
  the argument for the converse direction will be similar. 
  So we assume that $\sabsprei{1}$ is correct for $\atgi{1}$, and we show that $\sabsprei{2}$ is correct for $\atgi{2}$. 
  The correctness of $\sabsprei{2}$ for $\atgi{2}$ can be established
  by recognizing that the correctness conditions from Definition~\ref{def:abspre:function:siglambdaij}
  carry over from $\atgi{1}$ to $\atgi{2}$ via the homomorphism $\sahom$ 
  due to the property \eqref{eq:def:hom:image:absprefix:function:ltgs}
  of $\sabsprei{2}$ being the homomorphic image of $\sabsprei{1}$.
  
  For the two conditions that do not involve transitions this is easy to show: 
  The condition ($\text{root}$) for $\atgi{2}$
  follows from $\ahom{\srooti{1}} = \srooti{2}$ (the condition ($\text{root}$) for $\sahom$ to be a homomorphism)
  and \eqref{eq:def:hom:image:absprefix:function:ltgs}.
  And the condition $(\snlvar)_0$ follows similarly
  by using that every pre-image under $\sahom$ of a variable vertex in $\atgacc$
  must be a variable vertex in $\altg$ since $\sahom$ is a homomorphism.
 
%
%

  As an example for transferring the other correctness conditions from $\atgi{1}$ to $\atgi{2}$, 
  we treat the case of the condition $(\snlvarsucc)_1$.
  For this, let 
  $\bvertacc,\bvertacci{0}\in\vertsi{2}$
  such that
  $\bvertacc\in\vertsiof{2}{\snlvarsucc}$
  with
  $\bvertacc \tgsuccacci{0} \bvertacci{0}$. 
  Then as a consequence of the `arguments-forward' condition for $\sahom$ as in Proposition~\ref{prop:hom:tgs}
  there exist 
  $\bvert,\bverti{0}\in\vertsi{1}$
  with
  $\bvert\in\vertsiof{1}{\snlvarsucc}$,
  $\bvert \tgsucci{0} \bverti{0}$,
  and with
  $\ahom{\bvert} = \bvertacc$,
  $\ahom{\bverti{0}} = \bvertacci{0}$.
  Since $(\snlvarsucc)_1$ is satisfied for $\atgi{1}$,
  there exists a vertex $\avert\in\vertsi{1}$
  such that
  $\absprei{1}{\bvertbp{0}{}} \avert
      =
    \absprei{1}{\bvert}$.
  From this and by~\eqref{eq:def:hom:image:absprefix:function:ltgs} we obtain
  $\stringcon{\absprei{2}{\bvertacci{0}}}{\ahom{\avert}}
     =
   \stringcon{\absprei{2}{\ahom{\bverti{0}}}}{\ahom{\avert}}  
     =
   \stringcon{\ahomext{\absprei{1}{\bverti{0}}}}{\ahom{\avert}}  
     =
   \ahomext{\stringcon{\absprei{1}{\bverti{0}}}{\avert}}
     =
   \ahomext{\absprei{1}{\bvert}}
     =
   \absprei{2}{\ahom{\bvert}}  
     =
   \absprei{2}{\bvertacc}$.
  This shows the existence of 
  $\avertacc\in\vertsi{2}$
  (let namely $\avertacc \defdby \ahom{\avert}$)
  with 
  $\stringcon{\absprei{2}{\bvertacci{0}}}{\avertacc}
     =
   \absprei{2}{\bvertacc}$.
  In this way we have established the correctness condition $(\snlvarsucc)_1$ for $\atgi{2}$.  
  
  The direction from left to right in the equivalence in (\ref{lem:hom:image:absprefix:function:ltgs:item:i})
  can be established analogously 
  by recognizing that the correctness conditions from Definition~\ref{def:abspre:function:siglambdaij}
  carry over also from $\atgi{2}$ to $\atgi{1}$ via the homomorphism $\sahom$ 
  due to the homomorphic image property \eqref{eq:def:hom:image:absprefix:function:ltgs}.
  The arguments for the individual correctness conditions are analogous to the ones used above,
  but they depend on using the `arguments-forward' property in Proposition~\ref{prop:hom:tgs}
  of the homomorphism $\sahom$. 
\end{proof}

Now we define a precise relationship between \lambdatg{s} defined above and \lambdaaphotg{s} as defined in the previous section
via translation mappings between these classes:
\begin{description}\itemizeprefs
\item[The mapping $\slaphotgstoltgsij{i}{j}$ {\normalfont (see Proposition~\ref{prop:mappings:laphotgs:to:ltgs})}] 
  produces a \lambdatg\ for any given \lambdaaphotg\ by adding to the
  original set of vertices a number of delimiter vertices at the appropriate
  places. That is, at every position where the abstraction prefix decreases by $n$
  elements, $n$ $\snlvarsucc$-vertices are inserted. In the image, the original
  abstraction prefix is retained as part of the vertices. This can be considered
  intermittent information used for the purpose of defining the edges of the image.
\item[The mapping $\sltgstolaphotgsij{i}{j}$ {\normalfont (see Proposition~\ref{prop:mappings:ltgs:to:laphotgs})}] 
  back to \lambdaaphotg{s} is
  simpler because it only has to erase the $\snlvarsucc$-vertices, and add the correct abstraction prefix
  that exists for the \lambdatg\ to be translated.
\end{description}
 

\begin{proposition}\label{prop:mappings:laphotgs:to:ltgs}
  Let $i\in\setexp{0,1}$ and $j\in\setexp{1,2}$.
  The mapping $\slaphotgstoltgsij{i}{j}$ defined below 
  is well-defined between the class of \lambdatg{s} over $\siglambdaij{i}{j}$
  and the class of \lambdaaphotg{s} over $\siglambdai{i}$:
\begin{align*}
&    
\begin{aligned}
  &
  \slaphotgstoltgsij{i}{j} \funin \classlaphotgsi{i} \to \classltgsij{i}{j},
    \;\;
  \alhotg = \tuple{\verts,\svlab,\svargs,\sroot,\sabspre} 
    \mapsto \laphotgstoltgsij{i}{j}{\alhotg} \defdby \tuple{\vertsacc,\svlabacc,\svargsacc,\srootacc} 
\end{aligned}
\end{align*}
where:
\begin{equation*}
\begin{gathered}  
  \vertsacc \defdby
    \descsetexp{\pair{\bvert}{\abspre{\bvert}}}{\bvert\in\verts}
      \cup
    \descsetexpBig{\tuple{\bvert,k,\bvertacc,\apre}}{\bvert,\bvertacc\in\verts,\,
                                              \bvert \tgsucci{k} \bvertacc,\,
                                              \begin{aligned}[c]
                                                & \vertsof{\bvert} = \sslabs
                                                    \:\logand\:
                                                  \abspre{\bvertacc} \prelt \apre \prele \stringcon{\abspre{\bvert}}{\bvert}
                                                \\[-0.75ex]
                                                & \logor\:
                                                  \vertsof{\bvert} = \sslapp
                                                    \:\logand\:
                                                  \abspre{\bvertacc} \prelt \apre \prele \abspre{\bvert}
                                              \end{aligned}} 
  \\
  \srootacc \defdby \pair{\sroot}{\emptyword} 
  \hspace*{10ex} 
  \begin{gathered}[t]
  \svlabacc \funin \vertsacc \to \siglambdaij{i}{j} \, , \;
    \begin{aligned}[t]
      \pair{\bvert}{\abspre{\bvert}} & \mapsto \vlabacc{\pair{\bvert}{\abspre{\bvert}}} \defdby \vlab{\bvert}
      \\[-0.5ex]
      \tuple{\bvert,k,\bvertacc,\apre} & \mapsto \vlabacc{\bvert,k,\bvertacc,\apre} \defdby \snlvarsucc
    \end{aligned}
  \end{gathered}  
\end{gathered}    
\end{equation*}
and $\svargsacc \funin \vertsacc \to (\vertsacc)^*$ is defined such that 
for the induced indexed successor relation $\stgsuccacci{(\cdot)}$ it holds: 
\begin{align*}
  &
  \bvert \tgsucci{k} \bverti{k}
    \:\logand\:
  \nodels{\bvert}{k} = 0
    \;\;\Longrightarrow\;\;
  \pair{\bvert}{\abspre{\bvert}} \tgsuccacci{k} \pair{\bverti{k}}{\abspre{\bverti{k}}} 
  \displaybreak[0]\\
  &
  \bvert \tgsucci{0} \bverti{0}
    \:\logand\:  
  \nodels{\bvert}{0} > 0
    \:\logand\:
  \vlab{\bvert} = \sslabs
    \:\logand\:
  \abspre{\bvert} = \stringcon{\stringcon{\abspre{\bverti{0}}}{\avert}}{\apre} 
    \\[-0.5ex]
    & \hspace*{6ex} 
    \;\;\Longrightarrow\;\;
      \pair{\bvert}{\abspre{\bvert}} \stgsuccacci{0} \tuple{\bvert,0,\bverti{0},\stringcon{\abspre{\bvert}}{\bvert}}
        \:\logand\:
      \tuple{\bvert,0,\bverti{0},\abspre{\bverti{0}}\avert} \stgsuccacci{0} \pair{\bverti{0}}{\abspre{\bverti{0}}}  
  \displaybreak[0]\\
  &   
  \bvert \tgsucci{k} \bverti{k} 
    \:\logand\:  
  \nodels{\bvert}{k} > 0
    \:\logand\:
  \vlab{\bvert} = \sslapp
    \:\logand\:
  \abspre{\bvert} = \stringcon{\stringcon{\abspre{\bverti{k}}}{\avert}}{\apre} 
    \\[-0.5ex]
    & \hspace*{6ex} 
    \;\;\Longrightarrow\;\;
      \pair{\bvert}{\abspre{\bvert}} \stgsuccacci{k} \tuple{\bvert,k,\bverti{k},\abspre{\bvert}}
        \:\logand\:
      \tuple{\bvert,k,\bverti{k},\abspre{\bverti{k}}\avert} \stgsuccacci{0} \pair{\bverti{k}}{\abspre{\bverti{k}}}
  \displaybreak[0]\\ 
  &  
  \bvert \tgsucci{k} \bverti{k}
    \:\logand\:  
  \nodels{\bvert}{k} > 0
    \:\logand\:
  \tuple{\bvert,k,\bverti{k},\stringcon{\apre}{\avert}}, \,
  \tuple{\bvert,k,\bverti{k},\apre} \in\vertsacc
    \;\;\Longrightarrow\;\;
  \tuple{\bvert,k,\bverti{k},\stringcon{\apre}{\avert}} \tgsuccacci{0} \tuple{\bvert,k,\bverti{k},\apre}
  \\
  &
  \bvert \tgsucci{k} \bverti{k}
    \:\logand\:  
  \nodels{\bvert}{k} > 0
    \:\logand\:
  \tuple{\bvert,k,\bverti{k},\stringcon{\apre}{\avert}} \in\vertsacc
    \:\logand\:
  j = 2  
    \;\;\Longrightarrow\;\;
  \tuple{\bvert,k,\bverti{k},\stringcon{\apre}{\avert}} \tgsuccacci{1} \pair{\bverti{k}}{\abspre{\bverti{k}}}
\end{align*}
for all $\bvert,\bverti{0},\bverti{1},\avert\in\verts$, $k\in\setexp{0,1}$, $\apre\in\verts^*$,
and where the function $\snodels$ is defined as:
\begin{equation*}
  \nodels{\bvert}{k} 
    \defdby
  \begin{cases}
    \length{\abspre{\bvert}} - \length{\abspre{\bvertacc}} 
      & \text{if } \bvert\in\vertsof{\sslapp} \logand \bvert \stgsucci{k} \bvertacc 
    \\[-0.5ex]
    \length{\abspre{\bvert}} + 1 - \length{\abspre{\bvertacc}} 
      & \text{if } \bvert\in\vertsof{\sslabs} \logand \bvert \stgsucci{k} \bvertacc
    \\[-0.5ex]
    0 & \text{otherwise}   
  \end{cases}  
\end{equation*}
$\slaphotgstoltgsij{i}{j}$ induces the function~
$\slaphotgsisotoltgsisoij{i}{j} \funin \classlaphotgsisoi{i} \to \classltgsisoij{i}{j}$,
$\alaphotgiso = \eqcl{\alaphotg}{\siso} \mapsto \eqcl{\laphotgstoltgsij{i}{j}{\alaphotg}}{\siso}$ 
on the isomorphism equivalence classes.
\end{proposition}

\begin{proposition}\label{prop:mappings:ltgs:to:laphotgs}
  Let $i\in\setexp{0,1}$ and $j\in\setexp{1,2}$.
  The mapping $\sltgstolaphotgsij{i}{j}$ defined below 
  is well-defined between the class of \lambdatg{s} over $\siglambdaij{i}{j}$ 
  and the class of \lambdaaphotg{s} over $\siglambdai{i}$:
\begin{align*}
& 
\begin{aligned}
  &
  \sltgstolaphotgsij{i}{j} \funin \classltgsij{i}{j} \to \classlaphotgsi{i},
    \;\;
  \altg = \tuple{\verts,\svlab,\svargs,\sroot} 
    \mapsto \ltgstolaphotgsij{i}{j}{\altg} \defdby \tuple{\vertsacc,\svlabacc,\svargsacc,\sroot',\sabspre'} 
  \\
  & \hspace*{15ex}
  \text{where }
  \begin{aligned}[t]
    & \vertsacc \defdby \vertsof{\sslabs} \cup \vertsof{{\sslapp}} \cup \vertsof{\snlvar}, \:
      \svlabacc \defdby \srestrictfunto{\svlab}{\vertsacc}, \:
      \sroot' \defdby \sroot, \:
      \\[-0.5ex]
    & \svargsacc \funin \vertsacc \to (\vertsacc)^*
       \text{ so that for the induced indexed succ.\ relation $\stgsuccacci{(\cdot)}$:} 
      \\[-0.5ex]
    & \hspace*{6ex}
        \averti{0} \tgsuccacci{k} \averti{1} 
          \;\funin\,\Leftrightarrow\;
        \averti{0} \binrelcomp{\stgsucci{k}}{\stgsuccisstar{0}{\snlvarsucc}} \averti{1} 
        \hspace*{3ex}  
        \text{(for all $\averti{0},\averti{1}\in\vertsacc$, $k\in\setexp{0,1}$)} \; 
      \\[-0.5ex]
    & \sabspre' \defdby \srestrictfunto{\sabspre}{\vertsacc}
      \text{ for the correct \absprefix\ function $\sabspre$ for $\altg$.} \,   
  \end{aligned}
\end{aligned}
\end{align*}
$\sltgstolaphotgsij{i}{j}$ induces the function~
$\sltgsisotolaphotgsisoij{i}{j} \funin \classltgsisoij{i}{j} \to \classlaphotgsisoi{i}$,
$\altgiso = \eqcl{\altg}{\siso} \mapsto \eqcl{\ltgstolaphotgsij{i}{j}{\altg}}{\siso}$
on the isomorphism equivalence classes.
\end{proposition}

The mappings $\sltgstolaphotgsij{i}{j}$ and $\slaphotgstoltgsij{i}{j}$
now define a correspondence between the class $\classlaphotgsi{i}$ of \lambdaaphotgs\
and the class $\classltgsij{i}{j}$ of \lambdatgs\ in then sense made precise by the following theorem.

\begin{theorem}[correspondence between \lambdaaphotg{s} with \lambdatg{s}]\label{thm:corr:laphotgs:ltgs}
  Let $i\in\setexp{0,1}$ and $j\in\setexp{1,2}$.
  The mappings $\sltgstolaphotgsij{i}{j}$ from Proposition~\ref{prop:mappings:ltgs:to:laphotgs} 
  and $\slaphotgstoltgsij{i}{j}$ from Proposition~\ref{prop:mappings:laphotgs:to:ltgs}
  define a correspondence 
  between the class of \lambdatg{s} over $\siglambdaij{i}{j}$ and the class of \lambdaaphotg{s} over $\siglambdai{i}$ 
  with the following properties:
\begin{enumerate}[(i)]\itemizeprefs
  \item{}\label{thm:corr:laphotgs:ltgs:item:i}
    $\scompfuns{\sltgstolaphotgsij{i}{j}}{\slaphotgstoltgsij{i}{j}} = \sidfunon{\classlaphotgsi{i}}$.
  \item{}\label{thm:corr:laphotgs:ltgs:item:ii} 
    For all $\altg\in\classltgsij{i}{j}\,$: \mbox{ } 
    $(\compfuns{\slaphotgstoltgsij{i}{j}}{\sltgstolaphotgsij{i}{j})}{\altg}
           \Sfunbisim
         \altg$.
  \item{}\label{thm:corr:laphotgs:ltgs:item:iii} 
    $\sltgstolaphotgsij{i}{j}$  and $\slaphotgstoltgsij{i}{j}$ 
    preserve and reflect the sharing orders on $\classlaphotgsi{i}$ and on $\classltgsij{i}{j}$:
    \vspace*{-0.25ex}
  \begin{align*}
    (\forall \alaphotgi{1},\alaphotgi{2}\in\classlaphotgsi{i}) \;\; \hspace*{-14ex} &&
       \alaphotgi{1} \funbisim \alaphotgi{2}
       \;\: & \Longleftrightarrow\;\:
     \laphotgstoltgsij{i}{j}{\alaphotgi{1}} \funbisim \laphotgstoltgsij{i}{j}{\alaphotgi{2}}
    \\  
    (\forall \altgi{1},\altgi{2}\in\classltgsij{i}{j}) \;\; \hspace*{-14ex} && 
       \ltgstolaphotgsij{i}{j}{\altgi{1}} \funbisim \ltgstolaphotgsij{i}{j}{\altgi{2}}
          \;\: & \Longleftrightarrow\;\:
       \altgi{1} \funbisim \altgi{2}
  \end{align*}
\end{enumerate}
Furthermore, 
statements analogous to
(\ref{thm:corr:laphotgs:ltgs:item:i}), (\ref{thm:corr:laphotgs:ltgs:item:ii}), and (\ref{thm:corr:laphotgs:ltgs:item:iii})
hold for the 
correspondences $\slaphotgsisotoltgsisoij{i}{j}$ and $\sltgsisotolaphotgsisoij{i}{j}$, 
induced by $\slaphotgstoltgsij{i}{j}$ and $\sltgstolaphotgsij{i}{j}$,
between the classes $\classlaphotgsisoi{i}$ and $\classltgsisoij{i}{j}$ 
of isomorphism equivalence classes of graphs in $\classlaphotgsisoi{i}$ and $\classltgsisoij{i}{j}$, respectively.  
\end{theorem}

\begin{example}\label{ex:corr:laphotgs:ltgs}\normalfont
The \lambdaaphotg{s} in Figure~\ref{fig:lambdaaphotg} correspond to the \lambdaaphotg{s} in Figure~\ref{fig:lambdatg}
via the mappings $\slaphotgstoltgsij{i}{j}$ and $\sltgstolaphotgsij{i}{j}$ as follows: 
$
~~ \laphotgstoltgsij{i}{j}{\alaphotgi{0}} = \altgi{0}' \, ,
~~~ \laphotgstoltgsij{i}{j}{\alaphotgi{1}} = \altgi{1}' \, ,
~~~ \ltgstolaphotgsij{i}{j}{\altgi{0}} = \alaphotgi{0}' \, ,
~~~ \ltgstolaphotgsij{i}{j}{\altgi{1}} = \alaphotgi{1}' \, 
$.
\end{example}


The next example indicates that the
correspondence mappings $\sltgstolaphotgsij{i}{j}$ from \lambdatgs\ to \lambdaaphotgs\ are not injective,
and hence not bijective, neither. 

\begin{example}\normalfont
  Consider the \lambdaaphotg~$\alaphotg\in\classlaphotgsi{0}$,
  and the \lambdatgs~$\altg,\altgacc\in\classltgsij{0}{1}$:
\graphs{
  \graph{\altg}{rem_corr_laphotgs_ltgs_nobij_snodes_shared}
  \graph{\alaphotg}{rem_corr_laphotgs_ltgs_nobij_aphotg}
  \graph{\altgacc}{rem_corr_laphotgs_ltgs_nobij_snodes_unshared}
}
Then it holds: $\sltgstolaphotgsij01(\atg) = \alhotg = \sltgstolaphotgsij01(\atgacc)$.
Furthermore note that $\altgacc \funbisims{\snlvarsucc} \altg$, and consequently $\altg \bisimS \altgacc$.
\end{example}

Of the \lambdatgs~$\atg$ and $\atgacc$ in the example above
we can say that they differ only in their `degree of $\snlvarsucc$\nb-sha\-ring',
with $\atg$ exhibiting a higher degree of $\snlvarsucc$\nb-sharing than $\atgacc$. 
Since the mapping $\sltgstolaphotgsij{0}{1}$ ignores $\snlvarsucc$\nb-vertices and their sharing,
the $\snlvarsucc$\nb-sha\-ring present and its degree are not reflected 
in the corresponding \lambdaaphotg. 
This observation for the example at hand can suggest the more general fact that is expressed by the following proposition:
the correspondence mappings $\sltgstolaphotgsij{i}{j}$
from \lambdatgs\ to \lambdaaphotgs\ do not distinguish between \Sbisimilar\ \lambdatgs.

\begin{proposition}\label{prop:funbisimS:bisimS:ltgstolaphotgs}
  Let $i\in\setexp{0,1}$ and $j \in\setexp{1, 2}$.
  The mapping $\sltgstolaphotgsij{i}{j}$ in Proposition~\ref{prop:mappings:ltgs:to:laphotgs}
  maps two \lambdatg{s} that are $\snlvarsucc$\nb-homomorphic or $\snlvarsucc$\nb-bi\-si\-mi\-lar to the same \lambdaaphotg.
  That is: 
  \begin{align*}
    \atgi{1} \funbisims{\snlvarsucc} \atgi{2}
      \;\; & \Longrightarrow\;\;
    \ltgstolaphotgsij{i}{j}{\atgi{1}} \iso \ltgstolaphotgsij{i}{j}{\atgi{2}} 
    &
    \atgi{1} \bisims{\snlvarsucc} \atgi{2}
      \;\; & \Longrightarrow\;\;
    \ltgstolaphotgsij{i}{j}{\atgi{1}} \iso \ltgstolaphotgsij{i}{j}{\atgi{2}}
  \end{align*}
  holds for all \lambdatg{s}~$\atgi{1},\atgi{2}$ over $\siglambdaij{i}{j}$.
  The two statements are equivalent.
\end{proposition}

What is more: the mappings $\slaphotgstoltgsij{i}{j}$ and $\sltgstolaphotgsij{i}{j}$ 
induce a bijective correspondence between the class of \lambdaaphotgs\ over $\siglambdai{i}$ 
and the result of factoring out \Sbisimilarity\ 
from the class of \lambdatgs\ over $\siglambdaij{i}{j}$,
after going over from these classes of term graphs
to respective classes of isomorphism equivalence classes of term graphs.

\begin{theorem}\label{thm:classlaphotgsiso:iso:factor:classltgsiso:bisimS}
  For all $i\in\setexp{0,1}$ and $j\in\setexp{1,2}$,
  the following mappings are each other's inverse:
  \begin{align*}
    \factorset{\slaphotgsisotoltgsisoij{i}{j}}{\sbisimSsubscript} 
      \funin 
    \classlaphotgsisoi{i} & \to \factorset{\classltgsisoij{i}{j}}{\sbisimSsubscript}
    &
    \factorset{\sltgsisotolaphotgsisoij{i}{j}}{\sbisimSsubscript} 
      \funin \factorset{\classltgsisoij{i}{j}}{\sbisimSsubscript} & \to \classlaphotgsisoi{i}
    \\
    \eqcl{\alhotg}{\siso} 
      & \mapsto
    \eqcl{\eqcl{\laphotgstoltgsij{i}{j}{\alhotg}}{\siso}}{\sbisimSsubscript}  
    &
    \eqcl{\eqcl{\altg}{\siso}}{\sbisimSsubscript}
      & \mapsto
    \eqcl{\ltgstolaphotgsij{i}{j}{\altg}}{\sbisimSsubscript}  
  \end{align*}   
  and preserve and reflect the sharing order $\sfunbisim$ on $\classlaphotgsisoi{i}$, 
  and the induced sharing order $\sfunbisim$ on $\factorset{\classltgsisoij{i}{j}}{\sbisimSsubscript}\,$, respectively.
  Hence an isomorphism is obtained:
  $ \pair{\classlaphotgsisoi{i}}{\sfunbisim}
      \iso 
    \pair{\factorset{\classltgsisoij{i}{j}}{\sbisimSsubscript}}
         {\sfunbisim} $.
\end{theorem}

The statement of this theorem could be paraphrased as follows:
\lambdaaphotgs\ can be viewed as $\sbisims{\snlvarsucc}$\nb-equi\-va\-lence classes
of \lambdatgs.
And, if it is made clear that this only holds for isomorphism equivalence classes,
then this wording of the statement is actually quite precise.
Yet it obfuscates a detail.
Namely, that 
the $\sbisims{\snlvarsucc}$\nb-equi\-va\-lence class of a \lambdatg~$\altg$
does not consist of all term graphs that are $\snlvarsucc$\nb-bi\-si\-mi\-lar with $\altg$,
but only of all \lambdatgs\ with that property.
The reason is that the class of \lambdatgs\ is not closed under $\sfunbisims{\snlvarsucc}$, 
and hence neither is it closed under $\sbisims{\snlvarsucc}$. 

\begin{example}\label{ex:classltgs:not:closed:under:bisimS}\normalfont
  Consider the term graphs over $\siglambdaij{1}{1}$ below: $\altgi{1}$ (left), and $\altgi{0}$ (right).  
  \begin{equation*}
    \begin{aligned}[c] 
      \figsmall{ltgs_without_sbacklinks_g1}
    \end{aligned} 
       \hspace*{5ex}\sfunbisimS\hspace*{2.5ex}
    \begin{aligned}[c]
      \figsmall{ltgs_without_sbacklinks_g0}
    \end{aligned}
  \end{equation*}
  Yet while $\altgi{1}$ is a \lambdatg, $\altgi{0}$ is not a \lambdatg: 
  the shared $\snlvarsucc$\nb-ver\-tex in $\altgi{0}$ should close the scopes of two different abstractions in $\altgi{0}$,
  but this is not possible in \lambdatgs\ due to the requirement of the existence of an \absprefix\ function. 
  Note that in a \lambdatg~$\altg$, a delimiter vertex $\bvert$ closes the scope
  of just the leftmost vertex (which is an abstraction vertex) in $\abspre{\bvert}$, where $\sabspre$ is the (uniquely existing) \absprefix\ function for $\altg$.
  
  Hence the class $\classltgsij{0}{1}$ of \lambdatgs\ over $\siglambdaij{0}{1}$
  is not closed under $\sfunbisimS$, nor under $\sbisimS$.
\end{example}

In connection with this observation, we will now define, as an aside but in order to round off the relationship between \lambdatgs\ and \lambdaaphotgs, 
a variant notion of \lambdatg\ that leads to a class of term graphs that \emph{is} closed under $\sbisims{\snlvarsucc}$. 
The motivation consists in a fact and a consideration:
\begin{itemize}
  \item
    The classes of \lambdatgs\ are not closed under, and hence are sensitive to,
    $\snlvarsucc$\nb-bisimilarity.
  \item 
    The degree of $\snlvarsucc$\nb-sha\-ring in the \lambdatg\ corresponding to a \lambdahotg\
    should not really matter for the representation as term graph.
    Recall here that $\snlvarsucc$\nb-sha\-ring is ignored by 
    the correspondence mappings $\sltgstolaphotgsij{i}{j}$ defined in Proposition~\ref{prop:funbisimS:bisimS:ltgstolaphotgs}.
\end{itemize}     
The idea underlying the variant notion of \lambdatg\ introduced below
is to consider correct \absprefix\ functions for term graphs over $\siglambdaij{i}{j}$ 
that are only defined on non-delimiter vertices.
Yet the correctness conditions for such \absprefix\ functions 
still record the scope closure effect that individual delimiter vertices have
on the \absprefix{es} of the other vertices. 
As a consequence, a delimiter vertex is able to close more than one scope.
For many term graphs over one of the signatures $\siglambdaij{i}{j}$
this facilitates degrees of $\snlvarsucc$\nb-sharing that are cannot take place in \lambdatgs.

\begin{definition}[correct \absprefix\ function on non-delimiter vertices
                   f.\ $\siglambdaij{i}{j}$-term-graphs]
    \label{def:abspre:function:siglambdaij:plus}\normalfont\mbox{}\\[-0.35ex]
  Let $\atg = \tuple{\verts,\svlab,\svargs,\sroot}$
  be a $\siglambdaij{i}{j}$\nb-term-graph
  with $i\in\setexp{0,1}$ and $j\in\setexp{1,2}$.
  A function $\sabspre \funin \verts\setminus\vertsof{\snlvarsucc} \to \verts^*$ 
  is called an \emph{\absprefix\ function} on the non-de\-li\-mi\-ter vertices for $\atg$.
  Such a function is called \emph{correct} if
  for all $\bvert,\bvertacc,\bverti{0}\in\verts$, $k\in\setexp{0,1}$, and $n\in\nats$ it holds:
  \begin{align*}
      \;\; & \Rightarrow \;\;
    \abspre{\sroot} = \emptyword  
    &&
    (\text{root})
    \displaybreak[0]\\
    \bvert\in\vertsof{\sslabs}
      \;\logand\;
    \bvert \stgsucci{0} \bvertacc  
      \;\logand\;
    \bvertacc\notin\vertsof{\snlvarsucc}
      \;\; & \Rightarrow \;\;
    \abspre{\bvert} = \abspre{\bvertacc} \bvert  
    &&
    (\sslabs)^{(0)}
    \\     
    \bvert\in\vertsof{\sslabs}
      \;\logand\;
    \bvert \comprewrels{\stgsucci{0}}{\stgsuccins{0}{n+1}{\snlvarsucc}} \bvertacc  
      \;\logand\;
    \bvertacc\notin\vertsof{\snlvarsucc}
      \;\; & \Rightarrow \;\;
    \existsst{\averti{1},\ldots,\averti{n}\in\verts}
             {\abspre{\bvert} = \abspre{\bvertacc} \averti{n}\cdots\averti{1}}
    \hspace*{1em}         
    && 
    (\sslabs)^{(n+1)}         
    \\    
    \bvert\in\vertsof{\sslapp}
      \;\logand\;
    \bvert \comprewrels{\stgsucci{k}}{\stgsuccins{0}{n}{\snlvarsucc}} \bvertacc  
      \;\logand\;
    \bvertacc\notin\vertsof{\snlvarsucc}
      \;\; & \Rightarrow \;\;
    \existsst{\bverti{1},\ldots,\bverti{n}\in\verts}
             {\abspre{\bvert} = \abspre{\bvertacc} \averti{n}\cdots\averti{1}}
    && 
    (\sslapp)
  \displaybreak[0]\\  
    \bvert\in\vertsof{\snlvar}
        \;\; & \Rightarrow \;\;
     \abspre{\bvert} \neq \emptyword   
    && 
    (\snlvar)_0
  \displaybreak[0]\\
    \bvert\in\vertsof{\snlvar}
      \;\logand\;
    \bvert \tgsucci{0} \bvertbp{0}{}
      \;\; & \Rightarrow \;\;
    \begin{aligned}[c]  
      & 
        \bvertbp{0}{}\in\vertsof{\sslabs}
      \;\logand\;
      \abspre{\bvertbp{0}{}} \bvertbp{0}{} = \abspre{\bvert} 
    \end{aligned}
    &&
    (\snlvar)_1 
  \end{align*}%
 Note that analogously as in Definition~\ref{def:lambdahotg} and in Definition~\ref{def:aplambdahotg},
  if $i=0$, then $(\snlvar)_1$ is trivially true and hence superfluous,
  and if $i=1$, then $(\snlvar)_0$ is redundant, because it follows from $(\snlvar)_1$ in this case. 
  Additionally, if $j=1$, then $(\snlvarsucc)_2$ is trivially true and therefore superfluous. 
\end{definition}

\begin{definition}[{\lambdatg\ over $\siglambdaij{i}{j}$} up to \Sbisimilarity]%
    \label{def:lambdatg:up:to:S:siglambdaij}\normalfont
  Let $i\in\setexp{0,1}$ and $j\in\setexp{1,2}$.
  A \emph{\lambdatg\ (with scope-de\-li\-mi\-ters)} over $\siglambdaij{i}{j}$ \emph{up to \Sbisimilarity} 
  is a $\siglambdaij{i}{j}$\nb-term-graph
  that admits a correct \absprefix\ function on non-delimiter vertices. 
  The class of \lambdatg{s} over $\siglambdaij{i}{j}$ up to \Sbisimilarity\
  is denoted by $\classlltgsij{i}{j}$,
  and the class of isomorphism equivalence classes of \lambdatg{s} over $\siglambdaij{i}{j}$ 
  up to \Sbisimilarity\ by $\classlltgsisoij{i}{j}$.
\end{definition}  
  
We first note that every \lambdatg~$\altg$ is also a \lambdatg\ up to \Sbisimilarity,
since restricting the \absprefix\ function of $\altg$ to its non-delimiter vertices
yields a correct \absprefix\ function on non-delimiter vertices in the sense of Definition~\ref{def:abspre:function:siglambdaij:plus}  
for $\altg$.    
  
\begin{proposition}
  Every \lambdatg\ over $\siglambdaij{i}{j}\,$, where $i\in\setexp{0,1}$ and $j\in\setexp{1,2}$,
  is also a \lambdatg\ over $\siglambdaij{i}{j}$ up to \Sbisimilarity. 
\end{proposition}  
  
\begin{example}\label{ex:ltg:up:to:Sbisimilarity}\normalfont
  While term graph $\altgi{0}$ over $\siglambdaij{0}{1}$ in Ex.~\ref{ex:classltgs:not:closed:under:bisimS} 
  is not a \lambdatg,
  it \emph{is} a \lambdatg\ up to \Sbisimilarity.
\end{example}

The next proposition justifies the name for this concept of \lambdatg.
  
\begin{proposition}\label{prop:classlltgs}
  For all $i\in\setexp{0,1}$ and $j\in\setexp{1,2}$, it holds that: $\,$
  \begin{enumerate}[(i)]
    \item{}\label{prop:classlltgs:item:i}
      $\classlltgsij{i}{j}$ is closed under $\sbisimS$ (as subclass of\/ the class $\classtgssiglambdaij{i}{j}$ of term graphs over $\siglambdaij{i}{j}$).
    \item{}\label{prop:classlltgs:item:ii}
      A term graph $\altg$ over $\siglambdaij{i}{j}$ 
        is a \lambdatg\ up to \Sbisimilarity\
          if and only if 
        it is \Sbisimilar\ to a \lambdatg\
        (that is, if $\altg \bisimS \altgacc$ for a \lambdatg~$\altgacc$).
     Consequently it holds: 
     \begin{center} 
       $\classlltgsij{i}{j} 
           \; = \;
        \descsetexp{\atg\in\classtgsover{\siglambdaij{i}{j}}}{\existsst{\atgacc\in{\classltgsij{i}{j}}}{\atg \bisimS \atgacc}}$
     \end{center} 
  \end{enumerate}       
\end{proposition}


Finally we obtain a statement analogous to Theorem~\ref{thm:classlaphotgsiso:iso:factor:classltgsiso:bisimS},
which now can be paraphrased as follows:
\lambdaaphotgs\ can be viewed as $\sbisims{\snlvarsucc}$\nb-equi\-va\-lence classes
of \lambdatgs\ up to \Sbisimilarity, when isomorphism equivalence classes are considered.
The subtle difference with the statement of Theorem~\ref{thm:classlaphotgsiso:iso:factor:classltgsiso:bisimS}
is that, due to Proposition~\ref{prop:classlltgs},~(\ref{prop:classlltgs:item:i}),
the $\sbisims{\snlvarsucc}$\nb-equi\-va\-lence class of a \lambdatg~$\altg$ up to \Sbisimilarity\
consists of all term graphs that are \Sbisimilar\ with $\altg$, and not just of all \lambdatgs\ that are \Sbisimilar\ with $\altg$.

\begin{theorem}\label{thm:classlaphotgsiso:iso:factor:classlltgsiso:bisimS}
  For all $i\in\setexp{0,1}$ and $j\in\setexp{1,2}$,
  the following mappings are each other's inverse:
  \begin{align*}
    \factorset{\slaphotgsisotoltgsisoij{i}{j}}{\sbisimSsubscript} 
      \funin 
    \classlaphotgsisoi{i} & \to \factorset{\classlltgsisoij{i}{j}}{\sbisimSsubscript}
    &
    \factorset{\sltgsisotolaphotgsisoij{i}{j}}{\sbisimSsubscript} 
      \funin \factorset{\classlltgsisoij{i}{j}}{\sbisimSsubscript} & \to \classlaphotgsisoi{i}
    \\
    \eqcl{\alhotg}{\siso} 
      & \mapsto
    \eqcl{\eqcl{\laphotgstoltgsij{i}{j}{\alhotg}}{\siso}}{\sbisimSsubscript}  
    &
    \eqcl{\eqcl{\altg}{\siso}}{\sbisimSsubscript}
      & \mapsto
    \eqcl{\ltgstolaphotgsij{i}{j}{\altg}}{\sbisimSsubscript}  
  \end{align*}   
%
  Furthermore, they preserve and reflect the sharing order $\sfunbisim$ on $\classlaphotgsisoi{i}$, 
  and the induced sharing order $\sfunbisim$ on $\factorset{\classltgsisoij{i}{j}}{\sbisimSsubscript}\,$, respectively.
  Hence an isomorphism is obtained:
  $ \pair{\classlaphotgsisoi{i}}{\sfunbisim}
      \iso 
    \pair{\factorset{\classlltgsisoij{i}{j}}{\sbisimSsubscript}}
         {\sfunbisim} $.  
\end{theorem}


\section{Not closed under bisimulation and functional bisimulation}
  \label{sec:not:closed}

In this section we collect all negative results concerning closedness under bisimulation
and functional bisimulation for the classes of \lambdatg{s} as introduced in the previous section.

\begin{proposition}\label{prop:ltgs:not:closed:under:bisim}
  None of the classes 
  $\classltgsi{1}$ and\/ $\classltgsij{i}{j}$, for\/ $i\in\setexp{0,1}$ and $j\in\setexp{1,2}$, of \lambdatg{s}
  are closed under bisimulation.
\end{proposition}

This proposition is an immediate consequence of the next one, which can be viewed 
as a refinement, because it formulates non-closedness of classes of \lambdatg{s}
under specializations of bisimulation,
namely for functional bisimulation (under which some classes are not closed), 
and for converse functional bisimulation (under which none of the classes considered here is closed).

\begin{proposition}\label{prop:ltgs:not:closed:under:funbisim:convfunbisim}
  None of the classes $\classltgsij{{i}}{j}$ of \lambdatg{s} for\/ $i\in\setexp{0,1}$ and $j\in\setexp{1,2}$
  are closed under functional bisimulation, or under converse functional bisimulation.
  Additionally, the class $\classltgsi{{1}}$ of \lambdatg{s} is not closed under converse functional bisimulation.
  Broken down in items:
  \begin{enumerate}[(i)]\itemizeprefs
    \item{}\label{prop:ltgs:not:closed:under:funbisim:convfunbisim:item:i}
      None of the classes 
      $\classltgsij{{0}}{j}$ for\/ $j\in\setexp{1,2}$
      of \lambdatg{s} are closed under functional bisimulation $\sfunbisim$, or under
      converse functional bisimulation $\sconvfunbisim\,$.   
    \item{}\label{prop:ltgs:not:closed:under:funbisim:convfunbisim:item:ii}
      None of the classes 
      $\classltgsi{{1}}$ and\/ $\classltgsij{{1}}{j}$ for\/ $j\in\setexp{1,2}$
      of \lambdatg{s} are closed under converse functional bisimulation $\sconvfunbisim\,$.  
%
%
    \item{}\label{prop:ltgs:not:closed:under:funbisim:convfunbisim:item:iii} 
      The class $\classltgsij{{1}}{{1}}$ of \lambdatg{s} 
      is not closed under functional bisimulation $\sfunbisim\,$. 
    \item{}\label{prop:ltgs:not:closed:under:funbisim:convfunbisim:item:iv}  
      The class $\classltgsij{{1}}{{2}}$ of \lambdatg{s} 
      is not closed under functional bisimulation $\sfunbisim\,$. 
  \end{enumerate}
\end{proposition}

\begin{proof}
  For showing (\ref{prop:ltgs:not:closed:under:funbisim:convfunbisim:item:i}),
  let $\bsig$ be one of the signatures 
                                       $\siglambdaij{0}{j}$. 
  Consider the following term graphs over $\bsig$:
  \graphs
  {
      \graph{\atgi{2}}{lambdatgs_not_closed_under_funbisim_convfunbisim_item_i_g2}
      \graph{\atgi{1}}{lambdatgs_not_closed_under_funbisim_convfunbisim_item_i_g1}
      \graph{\atgi{0}}{lambdatgs_not_closed_under_funbisim_convfunbisim_item_i_g0}
  }%
  Note that $\atgi{2}$ represents the syntax tree of the nameless de-Bruijn-index notation
  $\lapp{(\sslabs\snlvar)}{(\sslabs\snlvar)}$
  for the \lambdaterm\ $\lapp{(\labs{\avar}{\avar})}{(\labs{\avar}{\avar})}$. 
  Then it holds: $\atgi{2} \funbisim \atgi{1} \funbisim \atgi{0}$.
  But while $\atgi{2}$ and $\atgi{0}$ admit correct ab\-strac\-tion-pre\-fix functions over $\bsig$ (nestedness of the implicitly defined scopes, here shaded),
  and consequently are \lambdatg{s} over $\bsig$, this is not the case for $\atgi{1}$  (overlapping scopes).
  Hence the class of \lambdatg{s} over $\bsig$ is closed neither under functional bisimulation nor under converse functional bisimulation.
  
  For showing (\ref{prop:ltgs:not:closed:under:funbisim:convfunbisim:item:ii}),
  let $\bsig$ be one of the signatures $\siglambdai{1}$ and $\siglambdaij{1}{j}$.
  Consider the term graphs over $\bsig$:
  \graphs
  {
          \graph{\atgi{1}'}{lambdatgs_not_closed_under_funbisim_convfunbisim_item_ii_g1}
          \graph{\atgi{0}'}{lambdatgs_not_closed_under_funbisim_convfunbisim_item_ii_g0}
  }
  Then it holds: $\atgi{1}' \funbisim \atgi{0}'$.
  But while $\atgi{0}'$ admits a correct ab\-strac\-tion-pre\-fix function,
  and therefore is a \lambdatg, over $\bsig$, this is not the case for $\atgi{1}'$ (due to overlapping scopes).
  Hence the class of \lambdatg{s} over $\bsig$ is not closed under converse functional bisimulation.
%

  For showing (\ref{prop:ltgs:not:closed:under:funbisim:convfunbisim:item:iii}),
  consider the following term graphs over $\siglambdaij{1}{1}$:
  \graphs
  {
          \graph{\atgi{1}''}{lambdatgs_not_closed_under_funbisim_convfunbisim_item_iii_g1}
          \graph{\atgi{0}''}{lambdatgs_not_closed_under_funbisim_convfunbisim_item_iii_g0}
  }
  Then it holds that $\atgi{1}'' \funbisim \atgi{0}''$.         
  However,        
  while $\atgi{1}''$ admits a correct abstraction-prefix function,
  and hence is a \lambdatg\ over $\siglambdaij{1}{1}$,
  this is not the case for $\atgi{0}''$ (due to overlapping scopes).
  Therefore the class of \lambdatg{s} over $\siglambdaij{1}{1}$ is not closed under functional bisimulation.

  For showing (\ref{prop:ltgs:not:closed:under:funbisim:convfunbisim:item:iv}),
  consider the following term graphs over $\siglambdaij{1}{2}$: 
  \begin{equation}\label{eq:counterexample:classltgsij12}
  \mbox{%
  \graphs{
          \graph{\atg'''_{1}}{lambdatgs_over_siglambda12_not_closed_under_funcbisim_g1}
          \hspace*{13ex}
          \graph{\atg'''_{0}}{lambdatgs_over_siglambda12_not_closed_under_funcbisim_g0}
          }}
  \end{equation}        
  Then it holds that $\atg'''_{1} \funbisim \atg'''_{0}$.         
  However,        
  while $\atg'''_{1}$ admits a correct abstraction-prefix function,
  and hence is a \lambdatg\ over $\siglambdaij{1}{2}$,
  this is not the case for $\atg'''_{0}$ (overlapping scopes).
  Therefore the class of \lambdatg{s} over $\siglambdaij{1}{2}$ is not closed under functional bisimulation.
  The scopes defined implicitly by these graphs are larger than necessary: they
  do not exhibit `eager scope closure', see Section~\ref{sec:closed}.%
\end{proof}

As an easy consequence of Proposition~\ref{prop:ltgs:not:closed:under:bisim},
and of Proposition~\ref{prop:ltgs:not:closed:under:funbisim:convfunbisim},
 (\ref{prop:ltgs:not:closed:under:funbisim:convfunbisim:item:i})
 and (\ref{prop:ltgs:not:closed:under:funbisim:convfunbisim:item:ii}),
 together with the examples used in the proof, 
we obtain the following two propositions. 

\begin{proposition}\label{prop:not:closed:under:bisim:lambdahotgs:aplambdahotgs:siglambdai}
  Let $i\in\setexp{0,1}$. 
  None of the classes $\classlhotgsi{i}$ of \lambdahotg{s}, 
  or $\classlaphotgsi{i}$ of \lambdaaphotg{s} 
  are closed under bisimulations on the underlying term graphs. 
\end{proposition}

\begin{proposition}\label{prop:not:closed:under:funbisim:convfunbisim:lambdahotgs:aplambdahotgs:siglambdai}
  The following statements hold:
  \begin{enumerate}[(i)]\itemizeprefs
    \item{}\label{prop:not:closed:under:funbisim:convfunbisim:lambdahotgs:aplambdahotgs:siglambdai:item:i}
      Neither the class $\classlhotgsi{0}$ 
      nor the class $\classlaphotgsi{0}$ 
      is closed under functional bisimulations, or under converse functional bisimulations, on the underlying term graphs.
    \item{}\label{prop:not:closed:under:funbisim:convfunbisim:lambdahotgs:aplambdahotgs:siglambdai:item:ii} 
      Neither the class $\classlhotgsi{1}$ of \lambdahotg{s} 
      nor the class $\classlaphotgsi{1}$ of \lambdaaphotg{s} 
      is closed under converse functional bisimulations on underlying term graphs. 
  \end{enumerate}  
\end{proposition}
Note that Proposition~\ref{prop:not:closed:under:funbisim:convfunbisim:lambdahotgs:aplambdahotgs:siglambdai},
          (\ref{prop:not:closed:under:funbisim:convfunbisim:lambdahotgs:aplambdahotgs:siglambdai:item:i})
is a strengthening of the statement of Proposition~\ref{prop:no:extension:funbisim:siglambda1} earlier.

\section{Closed under functional bisimulation}
  \label{sec:closed}

The negative results gathered in the last section might seem to show our enterprise in a quite poor state: 
For the classes of \lambdatg{s} we introduced, 
Proposition~\ref{prop:ltgs:not:closed:under:funbisim:convfunbisim}
only leaves open the possibility that the class $\classltgsi{1}$ is closed under functional bisimulation.
Actually, $\classltgsi{1}$ is closed (we do not prove this here), 
but that does not help us any further, because the correspondences in
Theorem~\ref{thm:corr:laphotgs:ltgs} do not apply to this class, and worse still,
Proposition~\ref{prop:forgetful} rules out simple correspondences for $\classltgsi{1}$.
So in this case we are left  
without the satisfying correspondences to \lambdahotg{s} and \lambdaaphotg{s}
that yet exist for the other classes of \lambdatg{s}, but which in their turn are not closed under functional bisimulation. 

But in this section we establish that the class $\classltgsij{1}{2}$ is very useful after all: its restriction 
to term graphs with eager application of scope closure is in fact closed under functional bisimulation. 

The reason for the non-closedness of $\classltgsij{1}{2}$ under functional bisimulation consists in the fact 
that \lambdatg{s} over $\siglambdaij{1}{2}$ do not necessarily exhibit `eager scope closure': 
For example in the term graph $\atgi{1}'''$ 
from the proof of Proposition~\ref{prop:ltgs:not:closed:under:funbisim:convfunbisim}, (\ref{prop:ltgs:not:closed:under:funbisim:convfunbisim:item:iv}),
the scopes of the two topmost abstractions are not closed on the paths to variable occurrences belonging to the bottommost abstractions,
although both scopes could have been closed immediately before the bottommost abstractions. 
When this is actually done, and the following variation $\tilde{\atg}_{1}$ of $\atgi{1}$ with eager scope closure is obtained,
then the problem disappears:
\begin{equation}\label{example:classeagltgsij12}
  \mbox{%
  \graphs{
          \graph{\tilde{\atg}_{1}}{lambdatgs_over_siglambda12_closed_under_funcbisim_eager_backlinks_g1}
          \hspace*{13ex}
          \graph{\tilde{\atg}_{0}}{lambdatgs_over_siglambda12_closed_under_funcbisim_eager_backlinks_g0}
          }}
\end{equation}
Its bisimulation collapse $\tilde{\atg}_{0}$ has again a correct \absprefix\ function and hence is a \lambdatg.


For \lambdatg{s} over $\siglambdaij{1}{1}$ and $\siglambdaij{1}{2}$ we define 
the properties of being `eager scope' and `fully back-linked' in the definition below.

\begin{definition}[eager scope and fully back-linked \lambdatg{s}]%
    \label{def:eager-scope:fully:back-linked}\normalfont
  Let $\atg = \tuple{\verts,\svlab,\svargs,\sroot}$ be a \lambdatg\ over $\siglambdaij{1}{j}$ for $j\in\setexp{1,2}$.
  Let $\sabspre \funin \verts \to \verts^*$ be the \absprefix\ function of $\atg$.
  
  We say that $\atg$ is an \emph{eager scope \lambdatg\ over $\siglambdaij{1}{j}$},
  or that $\atg$ \emph{has the eager-scope property} 
  if:
  \begin{equation}\label{eq:def:eager:scope}
    \left.
    \begin{aligned}
    \forallstzero{\bvert,\avert\in\verts}\, 
      \forallstzero{\apre\in\verts^*\!}
                    &
                    \abspre{\bvert} = \stringcon{\apre}{\avert}
                      \;\:\slogand\:\;
                    \bvert\notin\vertsof{\snlvarsucc}
                   \;\,\Rightarrow\,\;
                    \begin{aligned}[t]
                      & 
                      \existsstzero{n\in\nats}\existsstzero{\bverti{0},\ldots,\bverti{n}\in\verts} 
                      \\[-0.25ex]
                      & \hspace*{4ex}
                        \bvert =\bverti{0} \tgsucc \bverti{1} \tgsucc \ldots \tgsucc \bverti{n} \tgsucci{0} \avert
                      \\[-0.25ex]
                      & \hspace*{4ex}
                          \;\:\slogand\:\;
                        \bverti{n} \in\vertsof{\snlvar}
                          \;\:\slogand\:\;
                        \forallst{1\le i\le n-1}
                                 {\stringcon{\apre}{\avert} 
                                                            \le \abspre{\bverti{i}}} \punc{,}
                    \end{aligned}
    \end{aligned}          
    \hspace*{-1ex}\right\}          
  \end{equation}
  holds, or in words: if 
  for every non-delimiter vertex $\bvert$ in $\atg$ 
  with a non-empty abstraction-prefix $\abspre{\bvert}$ ending with $\avert$
  there exists a path from $\bvert$ to $\avert$ in $\atg$ via vertices with abstraction-prefixes that extend $\abspre{\bvert}$
  (a path in the scope of $\avert$,
   in view of the translation into \lambdahotg{s} defined in Proposition~\ref{prop:mappings:lhotgs:laphotgs}),
  and finally a variable-occurrence vertex before reaching $\avert$. 
  By $\classeagltgsij{1}{j}$ we denote the subclass of $\classltgsij{1}{j}$ 
  that consists of all eager-scope \lambdatg{s}. 
  
  We say that $\atg$ is \emph{fully back-linked} if it holds:
  \begin{equation}\label{eq:def:fully:back-linked}
    \left.
    \begin{aligned}
    \forallstzero{\bvert,\avert\in\verts}\, 
      \forallstzero{\apre\in\verts^*\!}
                    \abspre{\bvert} = \stringcon{\apre}{\avert}
                      \;\;\Rightarrow\;\;
                      & 
                      \existsstzero{n\in\nats}\existsstzero{\bverti{0},\ldots,\bverti{n}\in\verts} 
                      \\[-0.25ex]
                      & \hspace*{4ex}
                        \bvert = \bverti{0} \tgsucc \bverti{1} \tgsucc \ldots \tgsucc \bverti{n} \tgsucc \avert
                      \\[-0.25ex]
                      & \hspace*{4ex}
                          \;\:\slogand\:\;
                        \forallst{0\le i\le n}
                                     {\stringcon{\apre}{\avert} \le \abspre{\bverti{i}}}
                      \punc{,}
    \end{aligned}          
    \;\;\;\;\right\}          
  \end{equation}
  that is, if for all vertices $\bvert$ of $\atg$
  with non-empty abstraction prefix $\abspre{\bvert}$ that ends with $\avert$,
  there exists a path in $\atg$ from $\bvert$ to $\avert$ 
  via vertices with abstraction-prefixes that extend $\abspre{\bvert}$
  (hence, in view of the translation into \lambdahotg{s} defined in Proposition~\ref{prop:mappings:lhotgs:laphotgs}, 
   via vertices in the scope of $\avert$). 
  By $\classfblltgsij{1}{j}$ we denote the subclass of $\classltgsij{1}{j}$ 
  that consists of all fully back-linked \lambdatg{s}.
\end{definition}

For \lambdatg{s} with back-bindings for both variable vertices and delimiter-vertices
the property of being fully back-linked can be viewed as a generalization of the property of being an eager-scope \lambdatg.
This is stated by the following proposition.

\begin{proposition}\label{prop:eager:scope:fully:back-linked}
  Every eager-scope \lambdatg\ over $\siglambdaij{1}{2}$ is fully back-linked. 
\end{proposition}

\begin{proof}
  Let $\atg$ be an \lambdatg\ over $\siglambdaij{1}{2}$ with set $\verts$ of vertices.
  Let $\sabspre$ be the \absprefix\ function of $\atg$.
  Suppose that $\atg$ is an eager-scope \lambdatg. 
  This means that the condition \eqref{eq:def:eager:scope} holds for $\atg$. 
  Now note that, for all $\bvert,\avert\in\verts$ and $\apre\in\verts^*$, 
  if $\abspre{\bvert} = \stringcon{\apre}{\avert}$ and $\bvert\in\vertsof{\snlvarsucc}$ 
  holds, then the correctness condition $(\snlvarsucc)_2$ on the \absprefix\ function $\sabspre$
  entails $\bvert \tgsucci{1} \avert$.
  It follows that $\atg$ also satisfies the condition \eqref{eq:def:fully:back-linked},
  and therefore, that $\atg$ is fully back-linked.
\end{proof}

\begin{remark}[generalization of `eager scope' to \lambdatg{s} over $\siglambdaij{0}{1}$ and $\siglambdaij{0}{2}$]%
    \label{rem:eag-scope:lambdatgs:siglambdaij:0}\normalfont
  The intuition for the property of a \lambdatg\ $\atg$ to be `fully back-linked' is that,
  for every vertex $\bvert$ of $\atg$ with a non-empty \absprefix\ $\abspre{\bvert} = \stringcon{\apre}{\avert}$,
  it is possible to get back to the final abstraction vertex $\avert$ 
  in the \absprefix\ of $\bvert$
  by a directed path via vertices in the scope of $\avert$
  and via a last edge that is a back-binding from a variable or a delimiter vertex to the abstraction vertex $\avert$. 
  Therefore the presence of either sort of back-binding in \lambdatg{s} is crucial for this concept.
  Indeed,
  at least back-bindings for variables have to be present
  so that the property to be fully back-linked can make sense for a \lambdatg:
  all \lambdatg{s} with variable vertices but without variable back-links
  (for example, consider a representation of the \lambdaterm\ $\labs{\avar}{\avar}$ by such a term graph) 
  are not fully back-linked. 
  
  The situation is different for the property of being an eager-scope \lambdatg,
  because its intuition is that,
  for every non-delimiter vertex $\bvert$ of $\atg$ with a non-empty \absprefix\ $\abspre{\bvert} = \stringcon{\apre}{\avert}$,
  a variable vertex that represents a variable occurrence $\bverti{n}$ bound by the abstraction in $\avert$ 
  (and hence with $\abspre{\bverti{n}} = \stringcon{\apre}{\avert}$)
  can be reached from $\bvert$ via a directed path in $\atg$.
  While the presence of back-bindings was used for the definition of this property
  in \eqref{eq:def:eager:scope}, the assumption that all variable vertices have back-bindings
  is not essential here. 
  In fact the condition \eqref{eq:def:eager:scope} can be generalized 
  to apply also to \lambdatg{s} over $\siglambdaij{0}{1}$ and $\siglambdaij{0}{2}$, 
  by expressing the intuition just sketched, namely as follows:
  \begin{equation}\label{eq:def:eager:scope:without:var:backlinks}
    \left.
    \begin{aligned}
    \forallstzero{\bvert,\avert\in\verts}\, 
      \forallstzero{\apre\in\verts^*\!}
                    &
                    \abspre{\bvert} = \stringcon{\apre}{\avert}
                      \;\:\slogand\:\;
                    \bvert\notin\vertsof{\snlvarsucc}
                    \\
                    & \hspace*{3ex}
                   \Rightarrow\;\;
                   \existsstzero{n\in\nats}\existsstzero{\bverti{0},\ldots,\bverti{n}\in\verts} \;
                    \begin{aligned}[t]
                      & \bvert =\bverti{0} \tgsucc \bverti{1} \tgsucc \ldots \tgsucc \bverti{n} 
                      \\[-0.25ex]
                      & 
                          \;\:\slogand\:\;
                        \bverti{n} \in\vertsof{\snlvar}
                          \;\:\slogand\:\; 
                        \abspre{\bverti{n}} 
                                            = \stringcon{\apre}{\avert}
                      \\[-0.25ex]   
                      & 
                          \;\:\slogand\:\;
                        \forallst{0\le i\le n}
                                 {\stringcon{\apre}{\avert} = \abspre{\bvert} \le \abspre{\bverti{i}}} \punc{.}
                    \end{aligned}
    \end{aligned}          
    \;\;\;\;\right\}          
  \end{equation}
  For \lambdatg{s} over $\siglambdaij{1}{1}$ and $\siglambdaij{1}{2}$
  the conditions \eqref{eq:def:eager:scope} and \eqref{eq:def:eager:scope:without:var:backlinks} coincide:
  For the implication 
  \mbox{$\eqref{eq:def:eager:scope:without:var:backlinks} \Rightarrow \eqref{eq:def:eager:scope}$}
  note that
  the correctness condition $(\snlvar)_1$ on the \absprefix\ functions
  of these \lambdatg{s} 
  yields that the statements
  $\bverti{n} \in\vertsof{\snlvar}$
  and 
  $\abspre{\bverti{n}} = \abspre{\bvert} = \stringcon{\apre}{\avert}$
  imply that
  $\bverti{n} \stgsucci{0} \avert$.
  For the implication
  $\text{\eqref{eq:def:eager:scope}} \Rightarrow \text{\eqref{eq:def:eager:scope:without:var:backlinks}}$
  observe that
  $\bverti{n} \in\vertsof{\snlvar}$, $\bverti{n} \stgsucci{0} \avert$,
  and $\abspre{\bverti{n}} \prele \stringcon{\apre}{\avert}$
  implies 
  $\abspre{\bverti{n}} = \abspre{\bvert} = \stringcon{\apre}{\avert}$
  in view of the condition $(\snlvar)_1$ and Lemma~\ref{lem:ltg}. 
  
\end{remark}

\begin{proposition}\label{prop:hom:ltgs:preserve:reflect:eagscope:fb}
  Homomorphisms between \lambdatg{s} in $\classltgsij{1}{j}$, where $j\in\setexp{1,2}$, 
  preserve and reflect the properties `eager scope' and `fully back-linked'. 
\end{proposition}

\begin{proof}
  We only show preservation under homomorphic image of the property `\eagscope' for \lambdatgs,
  since preservation of the property `\fb' can be shown analogously and involves less technicalities. 
  Also, reflection of these properties under homomorphism can be demonstrated similarly.
  
  Let $j\in\setexp{1,2}$. 
  Let $\altgi{1} = \tuple{\vertsi{1},\svlabi{1},\svargsi{1},\srooti{1}}$
  and $\altgi{2} = \tuple{\vertsi{2},\svlabi{2},\svargsi{2},\srooti{2}}$
  be \lambdatgs\ over $\siglambdaij{1}{j}$ 
  with (correct) \absprefix\ functions $\sabsprei{1}$ and $\sabsprei{2}$, respectively. 
  Let $\sahom \funin \vertsi{1} \to \vertsi{2}$ be a homomorphism from $\altgi{1}$ to $\altgi{2}$,
  and suppose that $\altgi{1}$ is \eagscope.
  We show that also $\altgi{2}$ is an \eagscope\ \lambdatg.
  
  For this, 
  let $\bvertacc\in\vertsi{2}$ 
  such that 
  $\bvertacc\notin\vertsof{\snlvarsucc}$
  and $\absprei{2}{\bvertacc} = \stringcon{\apreacc}{\avertacc}$ for some $\avertacc\in\vertsi{2}$ and $\apreacc\in(\vertsi{2})^*$.
  Since $\sahom$ is surjective
  by Proposition~\ref{prop:hom:tgs:paths},~(\ref{prop:hom:tgs:paths:item:iii}), 
  there exists a vertex $\bvert\in\vertsi{1}$ such that $\ahom{\bvert} = \bvertacc$.
  Now note that, due to Proposition~\ref{prop:hom:image:absprefix:function:ltgs},
  $\sabsprei{2}$ is the homomorphic image of $\sabsprei{1}$.
  It follows that $\funap{\bar{\sahom}}{\absprei{1}{\bvert}} = \absprei{2}{\ahom{\bvert}} = \absprei{2}{\bvertacc} = \stringcon{\apreacc}{\avertacc}$.
  Hence there exist $\avert\in\vertsi{1}$ and $\apre\in(\vertsi{1})^*$ 
  such that $\absprei{1}{\bvert} = \stringcon{\apre}{\avert}$ and $\funap{\bar{\sahom}}{\apre} = \apreacc$ and $\ahom{\avert} = \avertacc$.
  Since $\atgi{1}$ is an \eagscope\ \lambdatg, there exists a path in $\altgi{1}$ of the form:
  \begin{equation*}
    \apath \;\funin\;
      \bvert 
        =
      \bverti{0}
        \tgsucc
      \bverti{1}
        \tgsucc
      \ldots
        \tgsucc
      \bverti{n}
        \tgsucci{0}
      \avert          
  \end{equation*}
  such that $\bverti{n}\in\vertsiof{1}{\snlvar}$ 
                 and $\absprei{1}{\bverti{i}} \prege \stringcon{\apre}{\avert}$
                 for all $i\in\setexp{0,1,\ldots,n}$.
  As $\sahom$ is a homomorphism, it follows from Proposition~\ref{prop:hom:tgs:paths},~(\ref{prop:hom:tgs:paths:item:i})
  that $\apath$ has an image $\ahom{\apath}$ in $\altgi{2}$ of the form:
  \begin{equation*}
    \ahom{\apath} \;\funin\;
      \bvertacc
        =
      \ahom{\bvert} 
        =
      \ahom{\bverti{0}}
        \tgsuccacc
      \ahom{\bverti{1}}
        \tgsuccacc
      \ldots
        \tgsuccacc
      \ahom{\bverti{n}}
        \tgsuccacci{0}
      \ahom{\avert} = \avertacc          
  \end{equation*}               
  where $\stgsuccacc$ is the directed-edge relation in $\altgi{2}$.
  Using again that $\sahom$ is a homomorphism, it follows that $\ahom{\bverti{n}}\in\vertsiof{2}{\snlvar}$.
  Due to the fact that $\sabsprei{2}$ is the homomorphic image of $\sabsprei{1}$
  it follows that
  for all $i\in\setexp{0,1,\ldots,n}$ it holds:
  $\absprei{2}{\ahom{\bverti{i}}}
     = 
   \funap{\bar{\sahom}}{\absprei{1}{\bverti{i}}}
     \prege
   \funap{\bar{\sahom}}{\stringcon{\apre}{\avert}}
     = 
   \stringcon{\funap{\bar{\sahom}}{\apre}}{\ahom{\avert}}
     =
   \stringcon{\apreacc}{\avertacc}$.
  Hence we have shown that for
  $\bvertacci{i} \defdby \ahom{\bverti{i}}\in\vertsi{2}$
  with $i\in\setexp{0,1,\ldots,n}$
  it holds that
  $\bvertacc 
     =
   \bvertacci{0}
     \tgsuccacc
   \bvertacci{1}
     \tgsuccacc
   \ldots
     \tgsuccacc
   \bvertacci{n}
     \tgsuccacci{0}
   \avertacc$
  such that $\bvertacci{n}\in\vertsiof{2}{\snlvar}$
  and 
  $\absprei{2}{\bvertacci{i}} \prege \stringcon{\apreacc}{\bvertacc}$
  for all $i\in\setexp{0,1,\ldots,n}$.
  In this way we have shown that also $\altgi{2}$ has the `\eagscope' property. 
  
  Note that for showing that the `\eagscope' property is reflected by a homomorphism $\sahom$ between \lambdatgs\ $\altgi{1}$ and $\altgi{2}$
  the fact that paths in $\altgi{2}$ have pre-images under $\sahom$ in $\altg$,
  see Proposition~\ref{prop:hom:tgs:paths},~(\ref{prop:hom:tgs:paths:item:ii}),
  can be used. 
\end{proof}

The following proposition states that the defining conditions for
a \lambdatg\ to be \eagscope\ or \fb\ 
are equivalent to respective generalized conditions.
The generalized condition for eager-scopedness requires that
for every non-delimiter vertex $\bvert$ with $\sabspre{\bvert} = \stringcon{\apre}{\stringcon{\avert}{\bpre}}$  
there exists a path from $\bvert$ to $\avert$ within the scope of $\avert$
that only transits variable-vertex back-links, but not delimiter-vertex back-links.
In the generalized condition for fully-backlinkedness the conclusion, 
in which the path may then also proceed via delimiter-vertex back-links,
holds for all vertices.

\begin{proposition}\label{prop:eager:scope:fully:back-linked:pumped}
  Let $\atg$ be a \lambdatg\ over $\siglambdaij{1}{2}$,
  and let $\sabspre$ be its \absprefix\ function.
  \begin{enumerate}[(i)]
    \item{}\label{prop:eager:scope:fully:back-linked:pumped:item:eager:scope}
      $\atg$ is an eager-scope \lambdatg\ if and only if:  
      \begin{equation}\label{eq:eager:scope:pumped}
        \left.
        \begin{aligned}
          \forallstzero{\bvert,\avert\in\verts}\, 
          \forallstzero{\apre,\bpre\in\verts^*\!}
                        \abspre{\bvert} = \stringcon{\apre}{\stringcon{\avert}{\bpre}} 
                        &
                          \;\:\slogand\:\;  
                        \bvert\notin\vertsof{\snlvarsucc}
                        \\
                        &
                      \begin{aligned}[t]  
                        \;\;\Rightarrow\;\;
                        & 
                        \existsstzero{n\in\nats}\existsstzero{\bverti{0},\ldots,\bverti{n}\in\verts} 
                        \\[-0.25ex]
                        & \hspace*{3ex}
                          \bvert = \bverti{0} \tgsucc \bverti{1} \tgsucc \ldots \tgsucc \bverti{n} \tgsucc \avert
                        \\[-0.25ex]
                        & \hspace*{3ex}
                          \bverti{n} \in\vertsof{\snlvar}
                            \;\:\slogand\:\;
                          \forallst{0\le i\le n}
                                   {\stringcon{\apre}{\avert} \le \abspre{\bverti{i}}} 
                          \punc{,}          
                      \end{aligned}           
        \end{aligned}          
        \;\;\;\;\right\}          
      \end{equation}
    \item{}\label{prop:eager:scope:fully:back-linked:pumped:item:fully:back-linked}
      The \lambdatg~$\atg$ is fully back-linked if and only if: 
      \begin{equation}\label{eq:fully:back-linked:pumped}
        \left.
        \begin{aligned}
          \forallstzero{\bvert,\avert\in\verts}\, 
          \forallstzero{\apre,\bpre\in\verts^*\!}
                        \abspre{\bvert} = \stringcon{\apre}{\stringcon{\avert}{\bpre}}
                        \;\;\Rightarrow\;\;
                        & 
                        \existsstzero{n\in\nats}\existsstzero{\bverti{0},\ldots,\bverti{n}\in\verts} 
                        \\[-0.25ex]
                        & \hspace*{3ex}
                          \bvert = \bverti{0} \tgsucc \bverti{1} \tgsucc \ldots \tgsucc \bverti{n} \tgsucc \avert
                        \\[-0.25ex]
                        & \hspace*{3ex}
                            \;\:\slogand\:\;
                          \forallst{0\le i\le n}
                                       {\stringcon{\apre}{\avert} \le \abspre{\bverti{i}}}
                      \punc{,}
        \end{aligned}          
        \;\;\;\;\right\}          
      \end{equation}
  \end{enumerate}
\end{proposition}

\begin{proof}
  We will only prove the statement~(\ref{prop:eager:scope:fully:back-linked:pumped:item:fully:back-linked}),
  since the proof of statement~(\ref{prop:eager:scope:fully:back-linked:pumped:item:eager:scope})
  is an easy adaptation of the proof of the former that is given below. 
  
  For showing statement~(\ref{prop:eager:scope:fully:back-linked:pumped:item:fully:back-linked}),
  let $\atg$ be an arbitrary \lambdatg\ over $\siglambdaij{1}{j}$ with $j\in\setexp{1,2}$.
  
  If \eqref{eq:fully:back-linked:pumped} holds for $\atg$, then 
  so does its special case \eqref{eq:def:fully:back-linked},
  and it follows that $\atg$ is fully back-linked.
  
  For showing the converse we assume that $\atg$ is fully back-linked.
  Then \eqref{eq:def:fully:back-linked} holds.
  We show \eqref{eq:fully:back-linked:pumped},
  which is a statement with universal quantification over $\bpre\in\verts^*$,
  by induction on the length $\length{\bpre}$ of $\bpre$.
  
  In the base case we have $\bpre = \emptyword$, and the statement to be establish follows 
  immediately from \eqref{eq:def:fully:back-linked}.
  
  For carrying out the induction step,
  let $\apre,\bpre\in\verts^*$ and $\bvert,\avert\in\verts$ be such that
  $\abspre{\bvert} = \stringcon{\apre}{\stringcon{\avert}{\bpre}}$ 
  with $\bpre = \stringcon{\avert'}{\bprei{0}}$ for $\avert'\in\verts$ and $\bprei{0}\in\verts^*$.  
  Due to $\length{\bprei{0}} < \length{\bpre}$,
  the induction hypothesis can be applied  
  to $\abspre{\bvert} = \stringcon{(\stringcon{\apre}{\avert})}{\stringcon{\avert'}{\bprei{0}}}$,
  yielding a path
  $\bvert \pathto \bvertacc \tgsucc \avertacc$, for some $\bvertacc\in\verts$,
  that in its part between $\bvert$ and $\bvertacc$ visits vertices with $\stringcon{\apre}{\stringcon{\avert}{\avertacc}}$ 
  as prefix of their abstraction prefixes.
  Due to $\stringcon{\apre}{\stringcon{\avert}{\avertacc}} \prele \abspre{\bvertacc}$ 
  it follows that 
  $\abspre{\avertacc} = \stringcon{\apre}{\avert}$
  by Lemma~\ref{lem:ltg}, (i)
  (which entails that  
   $\stringcon{\apre}{\stringcon{\avert}{\avertacc}} = \abspre{\bvertacc}$ 
   and that $\bvertacc \tgsucc \avertacc$ must be a back-link).
  From this
  the condition \eqref{eq:def:fully:back-linked}
  yields a path $\avert' \pathto \bvert'' \tgsucc \avert$, for some $\bvert''\in\verts$,
  that in its part between $\avert'$ and $\bvert''$ visits vertices with $\stringcon{\apre}{\avert}$ 
  as prefix of their abstraction prefixes.
  Combining these two paths, we obtain a path
  $\bvert \pathto \bvert' \tgsucc \avert' \pathto \bvert'' \tgsucc \avert$
  that until $\bvert''$ visits only vertices with $\stringcon{\apre}{\avert}$ as prefix of their abstraction prefixes.
  This establishes the statement to show for the induction step. 
\end{proof}



As the main lemma for the preservation result of \eagscope\ and \fb\ \lambdatgs\ under homomorphisms,
Theorem~\ref{thm:preserve:ltgs} below, we will show the following statement:
for vertices $\bverti{1}$ and $\bverti{2}$ in a \lambdatg~$\atg$ with \absprefix\ function $\sabspre$
that have the same image under a homomorphism~$\sahom$
also the images under $\sahom$ of their abstraction prefixes $\abspre{\bverti{1}}$ and $\abspre{\bverti{2}}$ are the same.

\begin{lemma}\label{lem:preserve:ltgs}
  Let $\atg\in\classltgsij{1}{2}$ 
  be a fully back-linked \lambdatg\  
  with vertex set $\verts$. Let $\sabspre$ be its \absprefix\ function. 
  Let $\atgacc\in\classtgssiglambdaij{1}{2}$ be a term graph over $\siglambdaij{1}{2}$
  such that $\atg \funbisimi{\sahom} \atgacc$ for a functional bisimulation~$\sahom$.   
  Then it holds:
  \begin{equation}\label{eq:lem:preserve:ltgs}
    \forallst{\bverti{1},\bverti{2}\in\verts}
             {\; 
              \ahom{\bverti{1}} = \ahom{\bverti{2}}
                \;\;\Rightarrow\;\;
              \ahomext{\abspre{\bverti{1}}} = \ahomext{\abspre{\bverti{2}}}} \punc{,}
  \end{equation}
  where $\sahomext$ is the homomorphic extension of $\sahom$ to words over $\verts$.
\end{lemma}

In order to establish this lemma we will, as a stepping stone, prove the following technical lemma.
In its proof the generalized condition
for fully back-linkedness from Proposition~\ref{prop:eager:scope:fully:back-linked:pumped}
is used to vertex-wise relate the \absprefix{es} on corresponding paths in a \lambdatg\ that
depart from vertices $\bverti{1}$ and $\bverti{2}$ with the same image under a homomorphism.

\begin{lemma}\label{lem:preserve:paths:ltgs}
  Let $\atg = \tuple{\verts,\svlab,\svargs,\sroot}$ be a \lambdatg\ over $\siglambdaij{1}{j}\,$, for $j\in\setexp{1,2}$,
  i.e.\ $\atg\in\classltgsij{1}{j}$.
  Let $\sabspre$ be its \absprefix\ function. 
  Let $\atgacc = \tuple{\vertsacc,\svlabacc,\svargsacc,\srootacc}$ 
  be a term graph over $\siglambdaij{1}{j}$ (i.e.\ $\atgacc\in\classtgssiglambdaij{1}{j}\,$, thus not assumed to be a \lambdatg)
  such that $\atg \funbisimi{\sahom} \atgacc$ for a functional bisimulation~$\sahom$. 
  
  Let $\bverti{1},\bverti{2}\in\verts$ be such that $\ahom{\bverti{1}} = \ahom{\bverti{2}}$.  
  
  Suppose that
  $\abspre{\bverti{1}} = \stringcon{\aprei{1}}{\stringcon{\averti{1}}{\bprei{1}}}$
  for $\averti{1}\in\verts$, and $\aprei{1},\bprei{1}\in\verts^*$,
  and that $\apathi{1}$ is a path from $\bverti{1}$ in $\atg$ of the form
  $\apathi{1} \funin 
                 \bverti{1} = \bverti{1,0}
                   \tgsucci{k_1} 
                 \bverti{1,1} 
                   \tgsucci{k_2} 
                 \bverti{1,2} 
                   \tgsucci{k_3} 
                     \cdots
                   \tgsucci{k_{n-1}}   
                 \bverti{1,n-1} 
                   \tgsucci{k_n}
                 \bverti{1,n}
                   =  
                 \averti{1}$,
  with some $k_0,k_1,\ldots,k_n\in\setexp{0,1}$,               
  is a path in $\atg$               
  such that furthermore 
  $\abspre{\bverti{1,n-1}} = \stringcon{\aprei{1}}{\averti{1}}$
  holds.
  
  Then there are
  $\averti{2}\in\verts$ and $\aprei{2},\bprei{2}\in\verts^*$
  with $\length{\bprei{2}} = \length{\bprei{1}}$ 
  such that
  $\abspre{\bverti{2}} = \stringcon{\aprei{2}}{\stringcon{\averti{2}}{\bprei{2}}}$,
  and 
  a path $\apathi{2}$ in $\atg$ from $\bverti{2}$ of the form
  $\apathi{2} \funin 
                 \bverti{2} = \bverti{2,0}
                   \tgsucci{k_1} 
                 \bverti{2,1} 
                   \tgsucci{k_2} 
                 \bverti{2,2} 
                   \tgsucci{k_3} 
                     \cdots
                   \tgsucci{k_{n-1}}   
                 \bverti{2,n-1} 
                   \tgsucci{k_{n}} 
                 \bverti{2,n}
                   = 
                 \averti{2}$
  such that
  $\abspre{\bverti{2,n-1}} = \stringcon{\aprei{2}}{\averti{2}}$,
  and that
  $\ahom{\bverti{1,j}} = \ahom{\bverti{2,j}}$ for all $j\in\setexp{0,1,\ldots,n}$,
  and in particular
  $\ahom{\averti{1}} = \ahom{\averti{2}}$.
\end{lemma}

\begin{proof}
  Let $\atg = \tuple{\verts,\svlab,\svargs,\sroot}$, $\atgacc = \tuple{\vertsacc,\svlabacc,\svargsacc,\srootacc}$, $\sahom$,
  and $\bverti{1},\bverti{2}\in\verts$ with $\ahom{\bverti{1}} = \ahom{\bverti{2}}$ as in the lemma. 
  Furthermore,  suppose that
  $\abspre{\bverti{1}} = \stringcon{\aprei{1}}{\stringcon{\averti{1}}{\bprei{1}}}$
  with $\averti{1}\in\verts$, and $\aprei{1},\bprei{1}\in\verts^*$,
  and a path:
  \begin{equation*}
    \apathi{1} \funin 
                 \bverti{1} = \bverti{1,0}
                   \tgsucci{k_1} 
                 \bverti{1,1} 
                   \tgsucci{k_2} 
                 \bverti{1,2} 
                   \tgsucci{k_3} 
                     \cdots
                   \tgsucci{k_{n-1}}   
                 \bverti{1,n-1} 
                   \tgsucci{k_n}
                 \bverti{1,n}
                   =  
                 \averti{1}
  \end{equation*}
  from $\bverti{1}$ in $\atg$    
  with $\abspre{\bverti{1,n-1}} = \stringcon{\aprei{1}}{\averti{1}}$.
  
  By Lemma~\ref{lem:ltg}, (i), 
  from $\abspre{\bverti{1,n-1}} = \stringcon{\aprei{1}}{\averti{1}}$ it follows that $\abspre{\averti{1}} = \aprei{1}$.  
  Hence the final transition in $\apathi{1}$ must be one via the back-binding
  of a variable vertex or a delimiter vertex. 
  Therefore
  either $\bverti{1,n-1}\in\vertsof{\snlvar}$ and $k_n = 0$,
  or $\bverti{1,n-1}\in\vertsof{\snlvarsucc}$ and $k_n = 1$. 
  Furthermore by Proposition~\ref{prop:hom:tgs:paths},~(\ref{prop:hom:tgs:paths:item:i}), 
  it follows that there is a path in $\atgacc$ from $\ahom{\bverti{1}}$ of the form:
  \begin{equation*}
    \ahom{\apathi{1}} \funin 
                 \ahom{\bverti{1}} = \ahom{\bverti{1,0}}
                   \tgsucci{k_1} 
                 \ahom{\bverti{1,1}} 
                   \tgsucci{k_2} 
                 \ahom{\bverti{1,2}} 
                   \tgsucci{k_3} 
                     \cdots
                   \tgsucci{k_{n-1}}   
                 \ahom{\bverti{1,n-1}} 
                   \tgsucci{k_n}
                 \ahom{\bverti{1,n}}
                   =  
                 \ahom{\averti{1}}
  \end{equation*}
  From this path, Proposition~\ref{prop:hom:tgs:paths},~(\ref{prop:hom:tgs:paths:item:ii}),
  yields a path $\apathi{2}$ in $\atg$ from $\bverti{2}$ of the form: 
  \begin{equation*}
    \apathi{2} \funin 
                 \bverti{2} = \bverti{2,0}
                   \tgsucci{k_1} 
                 \bverti{2,1} 
                   \tgsucci{k_2} 
                 \bverti{2,2} 
                   \tgsucci{k_3} 
                     \cdots
                   \tgsucci{k_{n-1}}   
                 \bverti{2,n-1} 
                   \tgsucci{k_n}
                 \bverti{2,n}
                   =  
                 \averti{2}
  \end{equation*}
  such that $\ahom{\apathi{2}} = \ahom{\apathi{1}}$.
  Therefore it holds
  $\ahom{\bverti{1,j}} = \ahom{\bverti{2,j}}$ for all $j\in\setexp{1,\ldots,n}$,
  and in particular
  $\ahom{\averti{1}} = \ahom{\averti{2}}$.
  Since $\sahom$ is a homomorphism also
  $\vlab{\bverti{1,j}} = \vlab{\bverti{2,j}}$ follows for all $j\in\setexp{0,1,\ldots,n}$.
  Therefore, and due to Lemma~\ref{lem:ltg},~\ref{lem:ltg:item:iv}, 
  the \absprefix\ function $\sabspre$ quantitatively behaves the same  
  when stepping through $\apathi{2}$
  as when stepping through $\apathi{1}$.
  It follows
  that $\abspre{\bverti{2}} = \stringcon{\aprei{2}}{\stringcon{\avert}{\bprei{2}}}$,
       $\abspre{\bverti{2,n-1}} = \stringcon{\aprei{2}}{\avert}$,
       and
       $\abspre{\averti{2}} = \aprei{2}$
  for some $\avert\in\verts$, and $\aprei{2},\bprei{2}\in\verts^*$ with $\length{\bprei{2}} = \length{\bprei{1}}$.  
  Due to 
  $\vlab{\bverti{2,n-1}} = \vlab{\bverti{1,n-1}}$,
  and since the final transition 
  $\bverti{1,n-1} \tgsucci{k_n} \averti{1}$ in $\apathi{1}$  
  is a one along a back-binding of a variable or delimiter vertex,
  this also holds for the final transition
  $\bverti{2,n-1} \tgsucci{k_n} \averti{2}$ in $\apathi{2}$.
  Then it follows by the correctness condition $(\snlvar)_1$ or $(\snlvarsucc)_2$, respectively,
  that $\avert = \averti{2}$,
  and therefore that 
  $\abspre{\bverti{2}} = \stringcon{\aprei{2}}{\stringcon{\averti{2}}{\bprei{2}}}$,
  and 
  $\abspre{\bverti{2,n-1}} = \stringcon{\aprei{2}}{\averti{2}}$.
\end{proof}

Relying on this lemma, we can now give a rather straightforward proof of Lemma~\ref{lem:preserve:ltgs}.

\vspace{0.75ex}
\begin{proofof}{Proof of Lemma~\ref{lem:preserve:ltgs}}
  Let $\atg$, $\atgacc$ be as assumed in the lemma,
  and let $\sahom$ be a functional bisimulation that witnesses $\atg \funbisimi{\sahom} \atgacc$. 

  We first show:
  \begin{equation}\label{eq1:prf:lem:preserve:ltgs}
  \left.  
  \begin{aligned} 
    & 
    \forallstzero{\bverti{1},\bverti{2}\in\verts}
    \forallstzero{\averti{1}\in\verts}
    \forallstzero{\aprei{1},\bprei{1}\in\verts^*}
              \;
              \ahom{\bverti{1}} = \ahom{\bverti{2}}
                \;\logand\;
              \abspre{\bverti{1}} = \stringcon{\aprei{1}}{\stringcon{\averti{1}}{\bprei{1}}}
      \\
      & \hspace*{15ex}   
      \;\;\Longrightarrow\;\; 
        \existsstzero{\averti{2}\in\verts}
        \existsstzero{\aprei{2},\bprei{2}\in\verts^*}
            \abspre{\bverti{2}} = \stringcon{\aprei{2}}{\stringcon{\averti{2}}{\bprei{2}}} 
              \;\logand\;
            \ahom{\averti{1}} = \ahom{\averti{2}}
              \;\logand\;
            \length{\bprei{2}} = \length{\bprei{1}}
  \end{aligned}
  \;\;\;\;\right.
  \end{equation}
  For this, let $\bverti{1},\bverti{2}\in\verts$ be such that
  $\ahom{\bverti{1}} = \ahom{\bverti{2}}$ and 
  $\abspre{\bverti{1}} = \stringcon{\aprei{1}}{\stringcon{\averti{1}}{\bprei{1}}}$
  for $\aprei{1},\bprei{1}\in\verts^*$ and $\averti{1}\in\verts$.
  Now since $\atg$ is fully back-linked,
  there is a path  
  $\apathi{1} \funin 
                 \bverti{1} = \bverti{1,0}
                   \tgsuccstar   
                 \bverti{1,n-1} 
                   \tgsucc
                 \bverti{1,n}
                   =  
                 \averti{1}$
  in $\atg$ with, in particular, $\abspre{\bverti{1,n-1}} = \stringcon{\aprei{1}}{\averti{1}}$.
  Then Lemma~\ref{lem:preserve:paths:ltgs} 
  yields the existence of $\aprei{1},\bprei{1}\in\verts^*$ and $\averti{1}\in\verts$
  with
  $\length{\bprei{2}} = \length{\bprei{1}}$
  such that furthermore
  $\ahom{\averti{1}} = \ahom{\averti{2}}$, and
  $\abspre{\bverti{2}} = \stringcon{\aprei{2}}{\stringcon{\averti{2}}{\bprei{2}}}$.
  This establishes \eqref{eq1:prf:lem:preserve:ltgs}.

  As an easy consequence of \eqref{eq1:prf:lem:preserve:ltgs} we then obtain:
  \begin{equation*}
  \left.  
  \begin{aligned} 
    \forallstzero{\bverti{1},\bverti{2}\in\verts}
    \forallstzero{\cprei{1}\in\verts^*}
              \; &
              \ahom{\bverti{1}} = \ahom{\bverti{2}}
                \;\logand\;
              \abspre{\bverti{1}} = \cprei{1} 
      \\
      & 
      \;\;\Longrightarrow\;\; 
        \existsstzero{\dprei{2},\cprei{2}\in\verts^*}
            \abspre{\bverti{2}} = \stringcon{\dprei{2}}{\cprei{2}}
              \;\logand\;
            \ahomext{\cprei{1}} = \ahomext{\cprei{2}}
  \end{aligned}
  \;\;\;\;\right.
  \end{equation*}
  But since in this statement the roles of $\bverti{1}$ and $\bverti{2}$ can be exchanged,
  it follows that
  $\ahom{\bverti{1}} = \ahom{\bverti{2}}$
  always entails 
  $\length{\abspre{\bverti{1}}} = \length{\abspre{\bverti{2}}}$,
  and hence
  $\ahomext{\abspre{\bverti{1}}} = \ahomext{\abspre{\bverti{2}}}$.
  This establishes \eqref{eq:lem:preserve:ltgs}.
\end{proofof}

Lemma~\ref{lem:preserve:ltgs} is the crucial stepping stone for the proof of the following theorem,
our main theorem.

\begin{theorem}
    [preservation under $\sfunbisim$ of \lambdatg{s} over $\siglambdaij{1}{2}$ in subclasses]
    \label{thm:preserve:ltgs}%
  Let $\altg$ be a \lambdatg\ over $\siglambdaij{1}{2}$ (that is, $\atg\in\classtgssiglambdaij{1}{2}$),
  and suppose that $\sahom$ is a functional bisimulation from $\altg$ 
  to a term graph $\atgacc$ over $\siglambdaij{1}{2}$
  (i.e.\ $\atgacc\in\classtgssiglambdaij{1}{2}\,$, and $\sahom$ witnesses $\atg \funbisimi{\sahom} \atgacc$). 
  %
  Then the following two statements hold:
  \begin{enumerate}[(i)]
    \item{}\label{thm:preserve:ltgs:item:i}
            If $\atg$ is fully back-linked, then also $\atgacc$ is a \lambdatg\ (i.e.\ $\atgacc\in\classltgsij{1}{2}$),
      which is fully back-linked.
    \item{}\label{thm:preserve:ltgs:item:ii}
            If $\atg$ is eager-scope, then also $\atgacc$ is a \lambdatg\ (i.e.\ $\atgacc\in\classltgsij{1}{2}$),
      which is eager-scope as well.
  \end{enumerate}
\end{theorem}

\begin{proof} 
  Let 
  $\altg = \tuple{\verts,\svlab,\svargs,\sroot} \in\classltgsij{1}{2}$
  and 
  $\atgacc = \tuple{\vertsacc,\svlabacc,\svargsacc,\srootacc}\in\classtgssiglambdaij{1}{2}$,
  and $\sahom$ a homomorphism from $\altg$ to $\atgacc$, i.e.\ 
  $\altg \funbisimi{\sahom} \atgacc$ holds.
  
  For showing statement~(\ref{thm:preserve:ltgs:item:i}) of the theorem, 
  we suppose that $\altg$ is fully back-linked.
  On $\atgacc$ we define the following \absprefix\ function:
  \begin{equation}\label{eq1:prf:thm:preserve:ltgs}
    \sabspreacc \funin \vertsacc \to (\vertsacc)^* ,
      \hspace*{2ex}
    \bvertacc \mapsto \abspreacc{\bvertacc} 
                          \defdby
                        \ahomext{\abspre{\bvert}} 
                        \text{ for $\bvert\in\verts$ (arbitrary) with $\ahom{\bvert} = \bvertacc$.} 
  \end{equation}
  This function is well-defined because:
  first,
  for every $\bvertacc\in\vertsacc$ there exists a $\bvert\in\verts$ with $\bvertacc = \ahom{\bvert}$,
  due to Proposition~\ref{prop:hom:tgs:paths},~(\ref{prop:hom:tgs:paths:item:iii}),
  and therefore $\abspreacc{\bvertacc}$ is definable by the defining clause in \eqref{eq1:prf:thm:preserve:ltgs};
  and second, due to Lemma~\ref{lem:preserve:ltgs}
  the definition of $\abspreacc{\bvertacc}$ does not depend on the particular chosen $\bvert\in\verts$ with $\bvertacc = \ahom{\bvert}$.
  
  For the so defined \absprefix\ function $\sabspreacc$ it holds:
  \begin{equation*}
    \forallstzero{\bvert\in\verts} \;\,
    \ahomext{\abspre{\bvert}} = \abspreacc{\ahom{\bvert}} \punc{,}
  \end{equation*}
  and hence $\sabspreacc$ is the homomorphic image of $\sabspre$ under $\sahom$
  in the sense of Definition~\ref{def:hom:image:absprefix:function:ltgs}. 
  
  It remains to show that $\sabspreacc$ is a correct \absprefix\ function for $\atgacc$,
  and that $\atgacc$ is fully back-linked. 
  Both properties follow from statements established earlier
  by using that $\atgacc$ and $\sabspreacc$ are the homomorphic images under $\sahom$
  of $\atg$ and $\sabspre$, respectively:
  That $\sabspreacc$ is a correct \absprefix\ function for $\atgacc$
  follows from Lemma~\ref{lem:hom:image:absprefix:function:ltgs},~(\ref{lem:hom:image:absprefix:function:ltgs:item:i})
  which entails that the homomorphic image of a correct \absprefix\ function is correct.
  And that $\atgacc$ is fully back-linked follows from
  Proposition~\ref{prop:hom:ltgs:preserve:reflect:eagscope:fb},
  which states that the homomorphic image of a fully back-linked \lambdatg\ is again fully back-linked.
  
  In this way we have established statement~(\ref{thm:preserve:ltgs:item:i}) of the theorem.
  
  For showing statement~(\ref{thm:preserve:ltgs:item:ii}),
  let $\atg$ be an \eagscope\ \lambdatg\ over $\siglambdaij{1}{2}$ .
  By Proposition~\ref{prop:eager:scope:fully:back-linked} it follows that $\atg$ is also fully back-linked.
  Therefore the just established statement~(\ref{thm:preserve:ltgs:item:i}) of the theorem
  is applicable, and it yields that $\atgacc$ is a \lambdatg. 
  Since by Proposition~\ref{prop:hom:ltgs:preserve:reflect:eagscope:fb}
  also the property of being \eagscope\ is preserved by homomorphism,
  it follows that $\atgacc$ is \eagscope, too.
  This establishes statement~(\ref{thm:preserve:ltgs:item:ii}) of the theorem. 
\end{proof}

The following corollary is a straightforward reformulation of the statement of Theorem~\ref{thm:preserve:ltgs}
as a property of the subclasses of $\classltgsij{1}{2}$ that consist 
of all fully back-linked, and respectively of all eager-scope, \lambdatg{s} over $\siglambdaij{1}{2}$.  

\begin{corollary}[preservation under $\sfunbisim$ of subclasses of \lambdatg{s} over $\siglambdaij{1}{2}$]%
    \label{cor:closed:under:fun:bisim}\label{cor:thm:preserve:ltgs} 
  The following two\\ subclasses of 
  the class $\classltgsij{1}{2}$ of all \lambdatg{s} over $\siglambdaij{1}{2}$
  are closed under functional bisimulation:
  the class $\classfblltgsij{1}{2}$ of fully back-linked \lambdatg{s}, 
  and the class $\classeagltgsij{1}{2}$ of eager-scope \lambdatg{s}.
\end{corollary}

\begin{remark}\normalfont
  Note that statements analogously to Theorem~\ref{thm:preserve:ltgs} and Corollary~\ref{cor:thm:preserve:ltgs}
  do not hold for \lambdatg{s} over $\siglambdaij{1}{1}$:
  The classes
  $\classfblltgsij{1}{1}$ and $\classeagltgsij{1}{1}$ 
  are not closed under functional bisimulation.
  This is witnessed by the counterexample in the proof of 
  Proposition~\ref{prop:ltgs:not:closed:under:funbisim:convfunbisim},
      (\ref{prop:ltgs:not:closed:under:funbisim:convfunbisim:item:iii}),
  which maps an eager-scope (and hence fully-backlinked) \lambdatg\
  to a term graph that is not a \lambdatg.
  
  Similarly the statements of Theorem~\ref{thm:preserve:ltgs} and Corollary~\ref{cor:thm:preserve:ltgs}
  do not carry over to \lambdatg{s} over $\siglambdaij{0}{1}$ or over $\siglambdaij{0}{2}$
  that are eager-scope in the sense of Remark~\ref{rem:eag-scope:lambdatgs:siglambdaij:0}.
  This is witnessed by the eager-scope term graphs that are used in the proof
  Proposition~\ref{prop:ltgs:not:closed:under:funbisim:convfunbisim},
      (\ref{prop:ltgs:not:closed:under:funbisim:convfunbisim:item:i}),
      (\ref{prop:ltgs:not:closed:under:funbisim:convfunbisim:item:ii}).
\end{remark}

Another direct consequence of Theorem~\ref{thm:preserve:ltgs} is the following corollary.
Recall the notation \eqref{eq:def:eqclin:succsofordin} for $\sfunbisim$\nb-successors
of $\siso$\nb-equivalence classes of term graphs.

\begin{corollary}\label{cor:2:thm:preserve:ltgs}
  The following statements holds:
  \begin{enumerate}[(i)]
    \item 
      For every \fb\ \lambdatg~$\altg$ over $\siglambdaij{1}{2}$
      it holds:
      $\succsoford{\altgiso}{\,\sfunbisim}
         =
       \succsofordin{\altgiso}{\,\sfunbisim}{\classfblltgsij{1}{2}}
         =
       \succsofordin{\altgiso}{\,\sfunbisim}{\classltgsij{1}{2}}$.
    \item
      For every \eagscope\ \lambdatg~$\altg$ over $\siglambdaij{1}{2}$
      it holds:
      $\succsoford{\altgiso}{\,\sfunbisim}
         =
       \succsofordin{\altgiso}{\,\sfunbisim}{\classeagltgsij{1}{2}}
         =
       \succsofordin{\altgiso}{\,\sfunbisim}{\classltgsij{1}{2}}$.   
  \end{enumerate}
\end{corollary}

This corollary will be central to proving, in Section~\ref{sec:transfer} 
the complete lattice property for $\sfunbisim$\nb-successors of ($\siso$\nb-equivalence classes of)
\lambdatgs\ over $\siglambdaij{1}{2}$ and \lambdaaphotgs\ over $\siglambdai{1}$.

\begin{remark}[Generalizing the result, dropping restrictions `\eagscope' or `\fb']%
              \label{rem:generalization:to:non:eagscope}
              \normalfont
              \mbox{}\\
  While it has now been established by Corollary~\ref{cor:thm:preserve:ltgs}
  that the class $\classeagltgsij{1}{2}$ of eager-scope \lambdatg{s} over $\siglambdaij{1}{2}$ 
  is closed under functional bisimulation $\sfunbisim$,
  we saw earlier 
  in Proposition~\ref{prop:ltgs:not:closed:under:funbisim:convfunbisim},~(\ref{prop:ltgs:not:closed:under:funbisim:convfunbisim:item:iv}),
  that the underlying larger class $\classltgsij{1}{2}$ of all \lambdatg{s} over $\siglambdaij{1}{2}$ 
  is not closed under $\sfunbisim$. 
  The reason for this failure is indicated by the counterexample in \eqref{eq:counterexample:classltgsij12}:
  in the scopes of the two parallel \lambdaabstractions\ above in $\atg'''_{0}$,
  which are different and uncomparable,
  the two subgraphs representing $\labs{\avar}{\avar}$ 
  are `dangling', since the scopes of the two parallel \lambdaabstractions\ above are not closed;
  therefore the mentioned subgraphs can be shared in the homomorphic image $\atg'''_{1}$, which leads to overlapping scopes,
  thus preventing the existence of a correct \absprefix\ function on $\atg'''_{1}$.  
  
  The problem disappeared by considering the eager-scope variant $\tilde{\atg}_{0}$ of $\atg'''_{0}$ in \eqref{example:classeagltgsij12}:
  in this case the scopes of the two topmost, parallel \lambdaabstractions\ above are closed as early as possible, and as a consequence
  the subgraphs can now be shared without leading to overlapping scopes. 
  
  However, there is a possibility to remedy non-closedness of $\classltgsij{1}{2}$ under $\sfunbisim$
  by going over to a slightly different \lambdatg\ representation.
  The idea is to close scopes that formally are yet unclosed when encountering a variable vertex 
  by trailing $\snlvarsucc$\nb-nodes that link back to the abstraction that is closed.  
  In doing so for the \lambdatg\ $\atg'''_{0}$ from \eqref{eq:counterexample:classltgsij12}
  we obtain the term graph ${\atg}''''_{0}$ in Figure~\ref{fig:ltgs:over:siglambdaij:2:2},
  whose bisimulation collapse ${\atg}''''_{1}$ does not exhibit any overlapping scopes: 
\begin{figure}[t]  
  \begin{equation*}
    \mbox{%
    \graphs{
          \graph{\atg''''_{1}}{lambdatgs_over_siglambda12_late_s_g1}
          \hspace*{13ex}
          \graph{\atg''''_{0}}{lambdatgs_over_siglambda12_late_s_g0}
          }}
  \end{equation*}
  \caption{\label{fig:ltgs:over:siglambdaij:2:2}
           Extending closedness under $\sfunbisim$ to non-\eagscope\ \lambdatgs\ over signature $\siglambdaij{2}{2}$
           with `late delimiter vertices':
           representation of the non-\eagscope\ \lambdatg~$\atg'''$ over $\siglambdaij{1}{2}$ from \eqref{eq:counterexample:classltgsij12}
           as a \lambdatg~$\atg''''_{1}$ over $\siglambdaij{2}{2}$,
           and its collapse, the \lambdatg~$\atg''''_{0}$ over $\siglambdaij{2}{2}$.}
\end{figure}  
  These term graphs require variable vertices with two outgoing edges, and separate trailer nodes.
  They are term graphs over the signature
  $\siglambdaij{2}{2} \defdby \siglambda \cup \setexp{ \snlvar, \snlvarsucc, \strailer }$
  with $\arity{\snlvar} = \arity{\snlvarsucc} = 2$ and $\arity{\strailer} = 0$,
  forming the class $\classtgssiglambdaij{2}{2}$ of such term graphs. 
  Accordingly, the definition of correct \absprefix\ function
  has to be adapted for term graphs over $\siglambdaij{2}{2}$:
  it extends the requirements that the conditions 
  $(\text{root})$,
  $(\sslabs)$,
  $(\sslapp)$,
  $(\snlvar)_0$,
  $(\snlvar)_1$, 
  $(\snlvarsucc)_1$,
  and 
  $(\snlvarsucc)_2$ from 
  Definition~\ref{def:abspre:function:siglambdaij} 
  hold for all $\bvert,\bverti{0},\bverti{1}\in\verts$ and $k\in\setexp{0,1}$
  by the following stipulations for the trailer vertices, and for the additional second outgoing edges of the variable vertices: 
  \begin{align*}
    \bvert\in\vertsof{\strailer}
      \;\; & \Rightarrow \;\;
    \left\{\,  
    \begin{aligned}[c]
      &
    \abspre{\bvert} = \emptyword\;\;\slogor\;\;
        \existsstzero{\bverti{0}\in\vertsof{\snlvar}} 
        \existsstzero{\bverti{1},\ldots,\bverti{n-1}\in\vertsof{\snlvarsucc}}
      \\[-0.25ex] 
      &
      \phantom{\bverti{1}\in\vertsof{\strailer}\;\;\slogor\;\;\bverti{1}\in\vertsof{\snlvarsucc}\bverti{1}\;\,\slogand\,\;}
      \hspace*{3ex}
      \bverti{0} 
        \tgsucci{1} 
      \bverti{1} 
        \tgsucci{1} 
      \ldots
        \tgsucci{1}
      \bverti{n-1}
        \tgsucci{1}
      \bvert     
    \end{aligned} 
    \;\;\right\}
    &
    (\strailer) 
    \\
    \bvert\in\vertsof{\snlvar}
      \;\logand\;
    \bvert \tgsucci{1} \bverti{1}
      \;\; & \Rightarrow \;\;
    \left\{\,  
    \begin{aligned}[c]
      &
      \abspre{\bvertbp{1}{}} \avert
          =
      \abspre{\bvert}
        \;\;   
      \text{for some $\avert\in\verts$} \hspace*{1ex}
      \\[-0.25ex]
      &
      \;\,\slogand\,\;
      \bigl(
      \bverti{1}\in\vertsof{\strailer} 
      \\[-0.5ex]
      &
      \phantom{\;\,\slogand\,\;\bigl(} 
        \;\;\slogor\;\;
      \bverti{1}\in\vertsof{\snlvarsucc}
          \;\,\slogand\,\;
        \existsstzero{\bverti{2},\ldots,\bverti{n-1}\in\vertsof{\snlvarsucc}}
        \existsstzero{\bverti{n}\in\vertsof{\strailer}} 
      \\[-0.25ex] 
      &
      \phantom{\;\,\slogand\,\;\bigl(\;\;\slogor\;\;\bverti{1}\in\vertsof{\snlvarsucc}\bverti{1}\;\,\slogand\,\;}
      \hspace*{3ex}
      \bverti{1} 
        \tgsucci{1} 
      \bverti{2} 
        \tgsucci{1} 
      \ldots
        \tgsucci{1}
      \bverti{n-1}
        \tgsucci{1}
      \bverti{n} 
      \,\bigr)   
    \end{aligned}
    \;\;\right\}  
    &
    (\snlvar)_{2}
  \end{align*}
  for all $\bvert,\bverti{1}\in\verts$.
  The condition $(\strailer)$ demands that a trailer vertex must always be reachable from
  a variable vertex by a path via zero, one, or more delimiter vertices such that all edges
  passed are edges with index~1 (second outgoing edges). 
  The condition $(\snlvar)_{2}$ requires that below a variable vertex (that is, always using the second outgoing edge)   
  only a sequence of delimiter vertices followed by a single trailer vertex can occur;
  it also demands, in conjunction with $(\snlvar)_{0}\,$, that the variable vertex itself
  acts as a delimiter vertex in that it closes the scope of the corresponding abstraction vertex. 
\end{remark}

\section{Transfer of complete lattice property to \lambdahotgs}
  \label{sec:transfer}

In this section we first establish that 
sets of $\sfunbisim$\nb-successors of (isomorphism equivalence classes of) a given \lambdatg\ over $\siglambdaij{1}{2}$
are complete lattices under the sharing order.
For this we use the fact that this is the case for first-order term graphs in general (see Proposition~\ref{prop:funbisim:succs:of:tgs:iso:complete:lattice}),
and we apply the results developed so far. 
And second, we transfer this complete lattice property to the higher-order \lambdaaphotgs\ over $\siglambdai{1}$
via the correspondences 
                        established in Section~\ref{sec:ltgs}.


The following proposition specializes Proposition~\ref{prop:funbisim:succs:of:tgs:iso:complete:lattice}
to term graphs over one of the signatures $\siglambdaij{i}{j}$. 

\begin{proposition}\label{prop:funbisim:succs:of:ltgs:iso:complete:lattice}
  $\pair{\succsoford{\altgiso}{\,\sfunbisim}}{\sfunbisim}$
  is a complete lattice
  for all $\altg\in\classtgssiglambdaij{i}{j}$
  with $i\in\setexp{0,1}$ and $j\in\setexp{1,2}$. 
\end{proposition}

As a consequence of the result in Section~\ref{sec:closed}
that the classes of \fb\ and \eagscope\ \lambdatgs\ over $\siglambdaij{1}{2}$ 
are closed under functional bisimulation,
and in particular as an immediate consequence of
Corollary~\ref{cor:2:thm:preserve:ltgs} and Proposition~\ref{prop:funbisim:succs:of:ltgs:iso:complete:lattice},
we obtain the following theorem.

\begin{theorem}\label{thm:complete:lattice:fbl:eagscope:ltgs}
  $\pair{\succsofordin{\altgiso}{\,\sfunbisim}{\classltgsij{1}{2}}}{\sfunbisim}$
  is a complete lattice
  for all $\altg\in\classfblltgsij{1}{2}$, and hence also
  for all $\altg\in\classeagltgsij{1}{2}$.   
%
\end{theorem}


So Theorem~\ref{thm:complete:lattice:fbl:eagscope:ltgs} states
the complete lattice property for $\sfunbisim$\nb-successors of
(isomorphism equivalence classes of) \lambdatgs\ over $\siglambdaij{1}{2}$ 
only for those term graphs that are \fb\ or \eagscope.
Based only on this statement, 
we therefore cannot expect to be able to transfer the complete lattice property
to the classes of all \lambdahotgs\ or \lambdaaphotgs\ over $\siglambdai{1}$, 
but just to subclasses that correspond to the \fb\ or \eagscope\ \lambdatgs.
However, 
         the property of being \fb\
does not have a natural equivalent for \lambdaaphotgs\ and \lambdatgs.
This is because the fact that this property is fulfilled
typically depends on the existence of paths that pass via back-binding edges of delimiter vertices,
and that is, via vertices and back-binding edges that are absent in \lambdaaphotgs\ and \lambdatgs.  
This is not the case for the property of being `\eagscope',
for which even the existence of back-links for variable vertices is not essential, see Remark~\ref{rem:eag-scope:lambdatgs:siglambdaij:0}.
Indeed, the \eagscope\ property has a quite natural counterpart for \lambdaaphotgs\ and \lambdatgs. 
Its definition below is analogous to the generalization, described in Remark~\ref{rem:eag-scope:lambdatgs:siglambdaij:0},  
of the definition of the property `\eagscope' for \lambdatgs\ with variable-vertex backlinks in Definition~\ref{def:eager-scope:fully:back-linked}
to one pertaining to all \lambdatgs, including those without variable-vertex backlinks. 

\begin{definition}[eager scope \lambdaaphotg{s} over $\siglambdai{i}$]%
    \label{def:eager-scope:laphotgs}\normalfont
  Let $i\in\setexp{0,1}$, and  
  let $\alaphotg = \tuple{\verts,\svlab,\svargs,\sroot,\sabspre}$ be a \lambdaaphotg\ over $\siglambdai{1}$.
  %
  We say that $\alaphotg$ is an \emph{eager scope \lambdatg},
  or that $\alaphotg$ \emph{has the eager-scope property}, 
  if:
  \begin{equation*}
    \left.
    \begin{aligned}
    \forallstzero{\bvert,\avert\in\verts}\, 
      \forallstzero{\apre\in\verts^*\!}
                    &
                    \abspre{\bvert} = \stringcon{\apre}{\avert}
                   \;\,\Rightarrow\,\;
                    \begin{aligned}[t]
                      & 
                      \existsstzero{n\in\nats}\existsstzero{\bverti{0},\ldots,\bverti{n}\in\verts} 
                      \\[-0.25ex]
                      & \hspace*{4ex}
                        \bvert =\bverti{0} \tgsucc \bverti{1} \tgsucc \ldots \tgsucc \bverti{n} 
                          \;\:\slogand\:\;
                        \bverti{n} \in\vertsof{\snlvar}
                          \;\:\slogand\:\;
                        \abspre{\bverti{n}} = \stringcon{\apre}{\avert}
                      \\[-0.25ex]
                      & \hspace*{4ex}
                          \;\:\slogand\:\;  
                        \forallst{1\le i\le n-1}
                                 {\stringcon{\apre}{\avert} 
                                                            \le \abspre{\bverti{i}}} 
                    \end{aligned}
    \end{aligned}          
    \hspace*{-1ex}\right\}          
  \end{equation*}
  holds, or in words: if 
  for every vertex $\bvert$ in $\atg$ 
  with a non-empty abstraction-prefix $\abspre{\bvert}$ ending with $\avert$
  there exists a path from $\bvert$ to $\avert$ in $\atg$ via vertices with abstraction-prefixes that extend $\abspre{\bvert}$
  (a path in the scope of $\avert$),
  to a variable vertex~$\bverti{n}$ that is bound by the abstraction vertex $\avert$.
  (Note that if $k=1$, then $\avert$ is directly reachable from $\bverti{n}$ via a back-binding edge.)
  
  By $\classeaglaphotgsi{1}$ we denote the subclass of $\classlaphotgsi{1}$ 
  that consists of all eager-scope \lambdaaphotg{s}. 
\end{definition}

As stated by the proposition below, the properties of being \eagscope\ 
for \lambdaaphotgs\ defined above and for \lambdatgs\ defined earlier
correspond to each other via the correspondence mappings
from Proposition~\ref{prop:mappings:laphotgs:to:ltgs} and Proposition~\ref{prop:mappings:ltgs:to:laphotgs}.

\begin{proposition}\label{prop:preserve:reflect:eagscope}
  Let $i\in\setexp{0,1}$, and $j\in\setexp{1,2}$. 
  The correspondences $\slaphotgstoltgsij{i}{j}$ and $\sltgstolaphotgsij{i}{j}$ 
  between \lambdaaphotgs~$\alaphotg$ over $\siglambdai{i}$
  and     \lambdatgs~$\altg$ over $\siglambdaij{i}{j}$
  preserve and reflect the property `\eagscope'.
  That is, for all $\alaphotg\in\classlaphotgsi{i}$, and for all $\altg\in\classltgsij{i}{j}$ it holds:
  \begin{align*}
    \text{$\alaphotg$ is \eagscope}
      \;\; & \Leftrightarrow\;\;
    \text{$\laphotgstoltgsij{i}{j}{\alaphotg}$ is \eagscope}
    \\
    \text{$\ltgstolaphotgsij{i}{j}{\altg}$ is \eagscope}
      \;\; & \Leftrightarrow\;\;
    \text{$\altg$ is \eagscope}  
  \end{align*}
  Consequently, the restriction
  of $\slaphotgstoltgsij{i}{j}$ to $\classeaglaphotgsi{i}$, 
  and the restriction 
  of $\sltgstolaphotgsij{i}{j}$ to $\classeagltgsij{i}{j}$
  as well as their counterparts 
  on isomorphism equivalence classes
  can be recognized as functions of following types:
  \begin{align*}
    \srestrictto{\slaphotgstoltgsij{i}{j}}{\classeaglaphotgsi{i}}
      & \funin
          \classeaglaphotgsi{i} \to \classeagltgsij{i}{j}
    &
    \srestrictto{\sltgstolaphotgsij{i}{j}}{\classeagltgsij{i}{j}}
      & \funin
          \classeagltgsij{i}{j} \to \classeaglaphotgsi{i}
    \\      
    \srestrictto{\slaphotgsisotoltgsisoij{i}{j}}{\classeaglaphotgsisoi{i}}
      & \funin
          \classeaglaphotgsisoi{i} \to \classeagltgsisoij{i}{j}
    &
    \srestrictto{\sltgsisotolaphotgsisoij{i}{j}}{\classeagltgsisoij{i}{j}}
      & \funin
          \classeagltgsisoij{i}{j} \to \classeaglaphotgsisoi{i}
  \end{align*}
\end{proposition}

Furthermore, the statement of Theorem~\ref{thm:corr:laphotgs:ltgs}
concerning the correspondences induced on the isomorphism equivalence classes
(in particular, those stemming from the statements 
 (\ref{thm:corr:laphotgs:ltgs:item:i}) and (\ref{thm:corr:laphotgs:ltgs:item:iii}))
can be specialized to 
                      sets of $\sfunbisim$\nb-successors
of isomorphism equivalence classes of \lambdaaphotgs\, and of \lambdatgs, as follows. 
Here and below we use the notation defined in \eqref{eq:def:eqclin:succsofordin}
for $\sfunbisim$\nb-successors of term graphs also for \lambdaaphotgs.

\begin{lemma}\label{lem:left-inverse:restrictions:to:bisim:equivclasses}
  Let $i\in\setexp{0,1}$, and $j\in\setexp{1,2}$.
  Let $\alaphotg\in\classlaphotgsi{i}$ be a \lambdaaphotg. 
  Then 
  \begin{align*}
    \srestrictto{\slaphotgsisotoltgsisoij{i}{j}}{\succsoford{\alaphotgiso}{\,\sfunbisimsubscript}}
      & \,\funin\,
          \succsoford{\alaphotgiso}{\,\sfunbisim} \to \succsofordin{\laphotgsisotoltgsisoij{i}{j}{\alaphotgiso}}{\,\sfunbisim}{\classltgsij{i}{j}}
    &
    \srestrictto{\sltgsisotolaphotgsisoij{i}{j}}{\succsofordin{\laphotgsisotoltgsisoij{i}{j}{\alaphotgiso}}{\,\sfunbisimsubscript}{\classltgsij{i}{j}}}
      & \,\funin\,
          \succsofordin{\laphotgsisotoltgsisoij{i}{j}{\alaphotgiso}}{\,\sfunbisim}{\classltgsij{i}{j}} \to \succsoford{\alaphotgiso}{\,\sfunbisim}
  \end{align*}
  and it holds that 
  $\srestrictto{\sltgsisotolaphotgsisoij{i}{j}}{\succsofordin{\laphotgsisotoltgsisoij{i}{j}{\alaphotgiso}}{\,\sfunbisimsubscript}{\classltgsij{i}{j}}}$
  is a left-inverse of 
  $\srestrictto{\slaphotgsisotoltgsisoij{i}{j}}{\succsoford{\alaphotgiso}{\,\sfunbisim}}$,
  that is, for all $\alaphotgacciso\in\succsoford{\alaphotgiso}{\,\sfunbisim}$  it holds:
  \begin{align*}
    \compfuns{\srestrictto{\sltgsisotolaphotgsisoij{i}{j}}{\succsofordin{\laphotgsisotoltgsisoij{i}{j}{\alaphotgiso}}{\,\sfunbisimsubscript}{\classltgsij{i}{j}}}}
             {\srestrictto{\slaphotgsisotoltgsisoij{i}{j}}{\succsoford{\alaphotgiso}{\,\sfunbisimsubscript}}}
             {\alaphotgacciso}
    & =
    \alaphotgacciso \punc{.}
  \end{align*}  
\end{lemma}

As a final tool for the transfer of the complete lattice property 
we formulate and prove a general lemma about partial orders. 
It states that the property of a partial order to be a complete lattice
is reflected by order homomorphisms with left-inverses.

\begin{lemma}
  \label{lem:reflect:complete:lattice}
  Let $\pair{\aset}{\spoi{\aset}}$ and $\pair{\bset}{\spoi{\bset}}$ be partial orders. 
  Suppose that 
  $\saohom \funin \aset \to \bset$, 
  and 
  $\sbohom \funin \bset \to \aset$
  are order homomorphisms such that 
  $\sbohom$ is a left-inverse of $\saohom$, that is,
  $\scompfuns{\sbohom}{\saohom} = \sidfunon{\aset}$ holds.
  Then if $\pair{\bset}{\spoi{\bset}}$ is a complete lattice, then so is $\pair{\aset}{\spoi{\aset}}$.
\end{lemma}
   
\begin{proof}
  Let $\pair{\aset}{\spoi{\aset}}$ and $\pair{\bset}{\spoi{\bset}}$ be partial orders.
  Let $\saohom \funin \aset \to \bset$, $\sbohom \funin \bset \to \aset$ be order homomorphisms
  with $\scompfuns{\sbohom}{\saohom} = \sidfunon{\aset}$.
  Suppose that $\pair{\bset}{\spoi{\bset}}$ is a complete lattice.
  We prove that $\pair{\aset}{\spoi{\aset}}$ is a complete lattice, too. 
  
  Let $S \subseteq\aset$ be arbitrary. 
  We have to show the existence of a \txtlub\ $\lub{S}$ and a \txtglb\ $\glb{S}$ of $S$ in $\pair{\aset}{\spoi{\aset}}$.
  We only show the existence of the \txtlub, 
  because the argument for the \txtglb\ is analogous.
  Since $\pair{\bset}{\spoi{\bset}}$ is a complete lattice,
  the existence of the \txtlub\ $\lub{\aohom{S}}$ is guaranteed. 
  We will show that $\lub{S} = \bohom{\lub{\aohom{S}}}$.
  
  $\bohom{\lub{\aohom{S}}}$ is an upper bound of $S$:
  Let $s\in S$ be arbitrary. Then $\aohom{s}\in\aohom{S}$, and $\aohom{s} \poi{\bset} \lub{\aohom{S}}$.  
  Since $\sbohom$ is an order homomorphism, and a left-inverse of $\saohom$, 
  it follows that 
  $s = \bohom{\aohom{s}} \poi{\aset} \bohom{\lub{\aohom{S}}}$.
  
  $\bohom{\lub{\aohom{S}}}$ is less or equal to all upper bounds for $S$:
  Let $u\in\aset$ be an upper bound of $S$. 
  As $\saohom$ is an order homomorphism, it follows that $\aohom{u}$ is an upper bound for $\aohom{S}$. 
  Consequently, $\lub{\aohom{S}} \poi{\bset} \aohom{u}$. 
  Again, since $\sbohom$ is an order homomorphism, and a left-inverse of $\saohom$,
  this entails $\bohom{\lub{\aohom{S}}} \poi{\aset} \bohom{\aohom{u}} = u$.
\end{proof}

Now we can formulate and prove the main result in this section.

\begin{theorem}
    \label{prop:funbisim:succs:of:tgs:iso:complete:lattice:lambdahotgs:lambdaaphotgs}
    For every $\alhotg\in\classeaglhotgsi{1}$ it holds that
      $\pair{\succsoford{\alaphotgiso}{\,\sfunbisim}}{\sfunbisim}$ 
    is a complete lattice.
\end{theorem}

\begin{proof}
  Let $\alhotg\in\classeaglhotgsi{1}$, i.e.\ $\alhotg$ is an \eagscope\ \lambdahotg.
  Then by Proposition~\ref{prop:preserve:reflect:eagscope} it follows that
  $\laphotgstoltgsij{1}{2}{\alaphotg} \in\classeagltgsij{1}{2}$, i.e.\ it is an \eagscope\ \lambdatg,
  and also that
  $\laphotgsisotoltgsisoij{1}{2}{\alaphotgiso} \in\classeagltgsisoij{1}{2}$.
  Now Theorem~\ref{thm:complete:lattice:fbl:eagscope:ltgs}  
  yields that
  $\pair{\succsofordin{\laphotgsisotoltgsisoij{1}{2}{\alaphotgiso}}{\,\sfunbisim}{\classltgsij{1}{2}}}{\sfunbisim}$
  is a complete lattice. 
  Furthermore note that 
  $\srestrictto{\sltgsisotolaphotgsisoij{1}{2}}{\succsofordin{\laphotgsisotoltgsisoij{1}{2}{\alaphotgiso}}{\sfunbisimsubscript}{\classltgsij{1}{2}}}$
  is a left-inverse of 
  $\srestrictto{\slaphotgsisotoltgsisoij{1}{2}}{\succsoford{\alaphotgiso}{\sfunbisimsubscript}}$
  due to Lemma~\ref{lem:left-inverse:restrictions:to:bisim:equivclasses}.
  Hence Lemma~\ref{lem:reflect:complete:lattice} is applicable
  to show that the complete lattice property of
  $\succsofordin{\laphotgsisotoltgsisoij{1}{2}{\alaphotgiso}}{\,\sfunbisim}{\classltgsij{1}{2}}$ with respect to $\sfunbisim$
  is reflected by
  $\srestrictto{\slaphotgsisotoltgsisoij{1}{2}}{\succsoford{\alaphotgiso}{\,\sfunbisimsubscript}}$,
  yielding that
  $\pair{\succsoford{\alaphotgiso}{\,\sfunbisim}}{\sfunbisim}$ 
  is a complete lattice.
\end{proof}

\section{Summary and Conclusion}
  \label{sec:conclusion}


We first defined higher-order term graph representations for cyclic \lambda-terms:
\begin{itemize}\itemizeprefs
  \item 
    \lambdahotg{s} in $\classlhotgsi{i}$,  
    an adaptation of Blom's `higher-order term graphs'~\cite{blom:2001},
    which possess a scope function that maps every abstraction vertex $\avert$ to the set of vertices that are in the scope of $\avert$. 
  \item 
    \lambdaaphotg{s} in $\classlaphotgsi{i}$,
    which instead of a scope function carry an \absprefix\ function 
    that assigns to every vertex $\bvert$ information about the scoping structure relevant for~$\bvert$. 
    Abstraction prefixes are closely related to the notion of `generated subterms' for \lambdaterms\ \cite{grab:roch:2012}.
    The correctness conditions here are simpler and more intuitive than for \lambdahotg{s}. 
\end{itemize}
These classes are defined for $i\in\setexp{0,1}$, according to whether variable occurrences have back-links
to abstractions (for $i=1$) or not (for $i=0$). 
Our main statements about these classes are:
\begin{itemize}\itemizeprefs
  \item 
    a bijective correspondence between $\classlhotgsi{i}$ and $\classlaphotgsi{i}$ 
    via mappings
    $\slhotgstolaphotgsi{i}$
    and
    $\slaphotgstolhotgsi{i}$
    that preserve and reflect the sharing order (Theorem~\ref{thm:corr:lhotgs:laphotgs});
 \item 
   the naive approach to implementing homomorphisms on theses classes 
   (ignoring all scoping information and using only the underlying first-order term graphs) 
   fails (Proposition~\ref{prop:forgetful}).
\end{itemize}
The latter was the motivation to consider first-order term graph implementations with scope delimiters:
\begin{itemize}\itemizeprefs
  \item 
    \lambdatg{s} in $\classltgsij{i}{j}$ 
    (with $i\in\setexp{0,1}$ and $j=2$ or $j=1$ for scope delimiter vertices with or without back-links, respectively),
    which are first-order term graphs without a higher-order concept, but for which correctness conditions are formulated via
    the existence of an \absprefix\ function. 
\end{itemize}
The most important results linking these classes with \lambdaaphotg{s} are:
\begin{itemize}\itemizeprefs
  \item 
    an `almost bijective' correspondence
    between the classes  $\classlaphotgsi{i}$ 
    and $\classltgsij{i}{j}$ 
    via mappings
    $\slaphotgstoltgsij{i}{j}$
    and
    $\sltgstolaphotgsij{i}{j}$
    that preserve and reflect the sharing order
    (Theorem~\ref{thm:corr:laphotgs:ltgs});
  \item 
    the subclass $\classeagltgsij{1}{2}$ of eager-scope \lambdatg{s} in $\classltgsij{1}{2}$
    is closed under homomorphism (Corollary~\ref{cor:closed:under:fun:bisim}).
\end{itemize}
The correspondences together with the closedness result allow us to derive
methods to handle homomorphisms between eager-scope higher-order term graphs in $\classlhotgsi{1}$ and $\classlaphotgsi{1}$
in a straightforward manner by implementing 
them via homomorphisms between first-order term graphs in $\classltgsij{1}{2}$. 


\begin{align*}  
\begin{gathered}[c]  
\begin{tikzpicture}[>=stealth]
\matrix[row sep=-5.35mm,column sep=1.2cm,ampersand replacement=\&]{
\node(tl){\phantom{I}};\&
\node(ml){\phantom{I}};\&
\node(bl){\phantom{I}};\\
\node(t){$\classeaglhotgsi{1}$};\&
\node(m){$\classeaglaphotgsi{1}$};\&
\node(b){$\classeagltgsij{1}{2}$};\\
\node(tr){\phantom{I}};\&
\node(mr){\phantom{I}};\&
\node(br){\phantom{I}};\\
};
\draw[->]($($(tl)!.5!(ml)$)!0.75cm!(tl)$) to node[above]{$\slhotgstolaphotgsi1$}  ($($(tl)!.5!(ml)$)!0.5cm!(ml)$);
\draw[->]($($(bl)!.5!(ml)$)!0.65cm!(ml)$) to node[above]{$\slaphotgstoltgsij12$}  ($($(bl)!.5!(ml)$)!0.6cm!(bl)$);
\draw[->]($($(br)!.5!(mr)$)!0.6cm!(br)$) to node[below]{$\sltgstolaphotgsij12$} ($($(br)!.5!(mr)$)!0.65cm!(mr)$);
\draw[->]($($(tr)!.5!(mr)$)!0.5cm!(mr)$) to node[below]{$\slaphotgstolhotgsi1$} ($($(tr)!.5!(mr)$)!0.75cm!(tr)$);
\end{tikzpicture}
\end{gathered}
& \hspace*{3.5ex} &
\begin{gathered}[c]
\begin{tikzpicture}[>=stealth]
\matrix[row sep=1.4cm,column sep=1.4cm,ampersand replacement=\&]{
\node(tl){$\alhotg$};\&
\node(ml){$\alaphotg'$};\&
\node(bl){$\altg$};\&\\
\node(tr){$\alhotgi{0}$};\&
\node(mr){$\alaphotgi{0}'$};\&
\node(br){$\altgi{0}$};\\
};
\draw[funbisim](tl) to node[right]{\scriptsize $\scollC{\classeaglhotgsi{1}}$}(tr);
\draw[funbisim](ml) to node[right]{\scriptsize $\scollC{\classeaglaphotgsi{1}} $}(mr);
\draw[funbisim](bl) to node[right]{\scriptsize $\scollC{\classeagltgsij{1}{2}} $}(br);
\draw[|->](tl) to node[above]{$\slhotgstolaphotgsi{1}$} (ml);
\draw[|->](ml) to node[above]{$\slaphotgstoltgsij{1}{2}$} (bl);
\draw[|->](br) to node[above]{$\sltgstolaphotgsij{1}{2}$} (mr);
\draw[|->](mr) to node[above]{$\slaphotgstolhotgsi{1}$} (tr);
\end{tikzpicture}
\end{gathered}
\end{align*}
For example, the property that a unique maximally shared form exists 
for \lambdatg{s} in $\classltgsij{1}{2}$ 
(which can be computed as the bisimulation collapse that is guaranteed to exist for first-order term graphs)
can now be transferred to eager-scope \lambdaaphotg{s} and \lambdahotg{s}
via the correspondence mappings (see the diagram above). 
For this to hold it is crucial that $\classeagltgsij{1}{2}$ is closed under homomorphism,
and that the correspondence mappings preserve and reflect the sharing order.  
Then the maximally shared form  
                                $\collC{\classeaglhotgsi{1}}{\alhotg}$ 
of an eager \lambdahotg\ $\alhotg$ can then be computed as:
\begin{center}
  $\collC{\classeaglhotgsi1}{\alhotg} =
   (\scompfuns{\slaphotgstolhotgsi1}{\scompfuns{\sltgstolaphotgsij12}{\scompfuns{\scollC{\classeagltgsij{1}{2}}}{\scompfuns{\slaphotgstoltgsij12}{\slhotgstolaphotgsi1}}}})(\alhotg)$.
\end{center}
where 
      $\scollC{\classeagltgsij{1}{2}}$ maps every \lambdatg\ in $\classltgsij{1}{2}$ to its bisimulation collapse. 
For obtaining 
              $\scollC{\classeagltgsij{1}{2}}\,$,
fast algorithms for computing the bisimulation collapse of first-order term graphs can be utilized.

While we have explained this result here only for term graphs with eager
scope-closure, the approach can be generalized to non-eager-scope term graphs as well
(see Remark~\ref{rem:generalization:to:non:eagscope}).
To this end scope delimiters have to be placed also subsequent to variable
vertices in some situations. Then variable occurrences do not implicitly close all open extended
scopes, but every extended scope that is open at some position must be
closed explicitly by scope delimiters on all (maximal) paths from that position.
The resulting graphs are fully back-linked, and then Theorem~\ref{thm:preserve:ltgs}
guarantees that the arising class of \lambdatg{s} is again closed under homomorphism.

For our original intent of getting a grip on maximal subterm sharing in the
\lambda-calculus with \stxtletrec\ or \mu, however, only eager scope-closure is
practically relevant, since it facilitates a higher degree of sharing.

Ultimately we expect that these results allow us to develop solid formalizations
and methods for subterm sharing in higher order languages with sharing
constructs. 

\vspace*{-1ex}
\paragraph{Acknowledgement.}
  We thank the reviewers of our submission for the workshop TERMGRAPH~2013,
  which led to the proceedings version \cite{grab:roch:2013:TERMGRAPH} on which this report is based,  
  for their encouraging comments, and for pointing out some inaccuracies in that submission.

\bibliography{ltgs}  

\end{document}